\documentclass{amsart}
\usepackage{amsmath,amsfonts,amssymb,graphicx,amsthm,amscd}
\usepackage{color,comment}
\usepackage[T1]{fontenc}
\usepackage[latin1]{inputenc}

\usepackage[lmargin=2cm,rmargin=2cm,tmargin=2cm,bmargin=2cm]{geometry}
\usepackage[toc,page]{appendix}
\usepackage[linktocpage=true]{hyperref}

\let\oldtocsection=\tocsection
\let\oldtocsubsection=\tocsubsection
\renewcommand{\tocsection}[2]{\hspace{0em}\oldtocsection{#1}{#2}}
\renewcommand{\tocsubsection}[2]{\hspace{1.5em}\oldtocsubsection{#1}{#2}}

\usepackage{tikz}
\usetikzlibrary{decorations.pathreplacing,decorations.markings,angles,quotes}

\newtheorem{theorem}{Theorem}[section]
\newtheorem{lemma}[theorem]{Lemma}
\newtheorem{corollary}[theorem]{Corollary}
\newtheorem{proposition}[theorem]{Proposition}
\newtheorem{definition}[theorem]{Definition}

\newtheorem{defi/prop}[theorem]{Definition/Proposition}

\newtheorem{fact}[theorem]{Fact}

\newcommand{\N}{\mathbf{N}}

\newcommand{\R}{\mathbf{R}}
\newcommand{\C}{\mathbf{C}}
\renewcommand{\P}{\mathbf{P}}

\renewcommand{\leq}{\leqslant}
\renewcommand{\geq}{\geqslant}

\newcommand{\st}{\  : \ }

\newcommand{\Id}{\mathrm{Id}}

\DeclareMathOperator{\Tr}{Tr}
\DeclareMathOperator{\E}{\mathbf{E}}

\newcommand{\braket}[2]{\langle #1 | #2\rangle}
\newcommand{\ketbra}[2]{| #1 \rangle\!\langle #2 |}
\newcommand{\bra}[1]{\langle #1 |}
\newcommand{\ket}[1]{| #1 \rangle}

\title{Correlation length in random MPS and PEPS}

\author{C\'{e}cilia Lancien, David P\'{e}rez-Garc\'{\i}a}

\address{\textbf{C\'{e}cilia Lancien:} Institut de Math\'{e}matiques de Toulouse \& CNRS, Universit\'{e} Paul Sabatier, F-31062 Toulouse Cedex 9, France.}
\email{clancien@math.univ-toulouse.fr}
\address{\textbf{David P\'{e}rez-Garc\'{\i}a:} Departamento de An\'{a}lisis Matem\'{a}tico, Universidad Complutense de Madrid, 28040 Madrid, Spain \& Instituto de Ciencias Matem\'{a}ticas, 28049 Madrid, Spain.}
\email{dperezga@ucm.es}

\date{November 25th 2025}

\keywords{Non-asymptotic theory of random matrices and random tensor networks, Many-body quantum systems, Decay of correlations.}

\begin{document}
	
\begin{abstract}
	Tensor network states are used extensively as a mathematically convenient description of physically relevant states of many-body quantum systems. Those built on regular lattices, i.e.~matrix product states (MPS) in dimension $1$ and projected entangled pair states (PEPS) in dimension $2$ or higher, are of particular interest in condensed matter physics. The general goal of this work is to characterize which features of MPS and PEPS are generic and which are, on the contrary, exceptional. This problem can be rephrased as follows: given an MPS or PEPS sampled at random, what are the features that it displays with either high or low probability? One property which we are particularly interested in is that of having either rapidly decaying or long-range correlations. In a nutshell, our main result is that translation-invariant MPS and PEPS typically exhibit exponential decay of correlations at a high rate. We have two distinct ways of getting to this conclusion, depending on the dimensional regime under consideration. Both yield intermediate results which are of independent interest, namely: the parent Hamiltonian and the transfer operator of such MPS and PEPS typically have a large spectral gap. In all these statements, our aim is to get a quantitative estimate of the considered quantity (generic correlation length or spectral gap), which has the best possible dependency on the physical and bond dimensions of the random MPS or PEPS. 
\end{abstract}

\maketitle

\tableofcontents

\newpage

\section{Introduction}

\subsection{Motivations} \hfill\par\smallskip

One of the main practical problems when dealing with many-body quantum systems is the curse of dimensionality: a system composed of $N$ $d$-dimensional particles has dimension $d^N$, a number of degrees of freedom too large to handle in most computations as soon as more than a few particles are involved. However, it is known that, in many contexts, physically relevant states of many-body quantum systems are actually well approximated by states living in a very small subset of the whole exponentially large state space, namely the one of tensor network states. Intuitively, these should be a mathematically convenient way of representing states of systems composed of many sub-systems having a certain geometry and subject to interactions respecting this geometry.

Tensor network states are constructed as follows: Given a non-oriented graph $G$ with vertex set $V$ and edge set $E$, we put at each $v\in V$ a tensor $\ket{\chi_v}\in\C^d\otimes(\C^D)^{\otimes d(v)}$, where $d(v)$ denotes the degree of $v$ (i.e.~the number of edges at $v$). We get in this way a tensor $\ket{\hat{\chi}_G}\in (\C^d)^{\otimes |V|}\otimes(\C^D)^{\otimes 2|E|}$. Then, we contract together the indices of $\ket{\hat{\chi}_G}$ corresponding to a same edge to obtain a tensor $\ket{\chi_G}\in(\C^d)^{\otimes |V|}$. The $D$-dimensional indices are thus called bond indices while the $d$-dimensional ones are called physical indices. This construction procedure is exemplified in Figure \ref{fig:TNS} (using a graphical representation of tensors to be explained in more details afterwards). It is clear from the construction that tensor network states have the practical advantage of requiring few parameters to be described: if $G$ has $N$ vertices, each of them having degree at most $r$, then the resulting tensor network state $\ket{\chi_G}$ is described by at most $ND^rd$ parameters, which is linear rather than exponential in $N$. 

\begin{figure}[h]
	\caption{Tensor network state construction}
	\label{fig:TNS}
\begin{center}
	\begin{tikzpicture} [scale=0.8]
	\draw[color=black] (0,0) -- (1.5,0); \draw[color=black] (0.5,1) -- (2,1); \draw[color=black] (0,0) -- (0.5,1); \draw[color=black] (1.5,0) -- (2,1); \draw[color=black] (1.5,0) -- (2.8,0.4); \draw[color=black] (2,1) -- (2.8,0.4); \draw[color=black] (3.8,0.4) -- (2.8,0.4);
	\draw (0,0) node {{\small $\bullet$}}; \draw (1.5,0) node {{\small $\bullet$}}; \draw (0.5,1) node {{\small $\bullet$}}; \draw (2,1) node {{\small $\bullet$}}; \draw (2.8,0.4) node {{\small $\bullet$}}; \draw (3.8,0.4) node {{\small $\bullet$}};
	\draw (1.9,-0.5) node {$G$ with $6$ vertices and $7$ edges};
	
	\begin{scope}[xshift=6cm]
	\draw[color=gray] (0,0) -- (0.7,0); \draw[color=gray] (0.9,0) -- (1.5,0); \draw[color=gray] (0.5,1) -- (1.2,1); \draw[color=gray] (1.4,1) -- (2,1); \draw[color=gray] (0,0) -- (0.22,0.44); \draw[color=gray] (0.28,0.56) -- (0.5,1); \draw[color=gray] (1.5,0) -- (1.72,0.44); \draw[color=gray] (1.78,0.56) -- (2,1); \draw[color=gray] (1.5,0) -- (2.08,0.17); \draw[color=gray] (2.22,0.23) -- (2.8,0.4); \draw[color=gray] (2,1) -- (2.35,0.75); \draw[color=gray] (2.45,0.65) -- (2.8,0.4); \draw[color=gray] (3.8,0.4) -- (3.37,0.4); \draw[color=gray] (3.22,0.4) -- (2.8,0.4);
	\draw[color=brown] (0,0) -- (0,0.8); \draw[color=brown] (1.5,0) -- (1.5,0.8); \draw[color=brown] (0.5,1) -- (0.5,1.8); \draw[color=brown] (2,1) -- (2,1.8); \draw[color=brown] (2.8,0.4) -- (2.8,1.2); \draw[color=brown] (3.8,0.4) -- (3.8,1.2);
	\draw (0,0) node {{\small $\bullet$}}; \draw (1.5,0) node {{\small $\bullet$}}; \draw (0.5,1) node {{\small $\bullet$}}; \draw (2,1) node {{\small $\bullet$}}; \draw (2.8,0.4) node {{\small $\bullet$}}; \draw (3.8,0.4) node {{\small $\bullet$}};
	\draw (1.9,-0.5) node {$\ket{\hat{\chi}_G}\in (\C^d)^{\otimes 6}\otimes(\C^D)^{\otimes 14}$};
	\end{scope}
	
	\begin{scope}[xshift=12cm]
	\draw[color=gray] (0,0) -- (1.5,0); \draw[color=gray] (0.5,1) -- (2,1); \draw[color=gray] (0,0) -- (0.5,1); \draw[color=gray] (1.5,0) -- (2,1); \draw[color=gray] (1.5,0) -- (2.8,0.4); \draw[color=gray] (2,1) -- (2.8,0.4); \draw[color=gray] (3.8,0.4) -- (2.8,0.4);
	\draw[color=brown] (0,0) -- (0,0.8); \draw[color=brown] (1.5,0) -- (1.5,0.8); \draw[color=brown] (0.5,1) -- (0.5,1.8); \draw[color=brown] (2,1) -- (2,1.8); \draw[color=brown] (2.8,0.4) -- (2.8,1.2); \draw[color=brown] (3.8,0.4) -- (3.8,1.2);
	\draw (0,0) node {{\small $\bullet$}}; \draw (1.5,0) node {{\small $\bullet$}}; \draw (0.5,1) node {{\small $\bullet$}}; \draw (2,1) node {{\small $\bullet$}}; \draw (2.8,0.4) node {{\small $\bullet$}}; \draw (3.8,0.4) node {{\small $\bullet$}};
	\draw (1.9,-0.5) node {$\ket{\chi_G}\in (\C^d)^{\otimes 6}$};
	\end{scope}
	\end{tikzpicture}
\end{center}
\end{figure}

Often the underlying graph $G$ is taken to be a regular lattice. The corresponding tensor network state $\ket{\chi_G}$ is then usually referred to as a matrix product state (MPS) in dimension $1$ and a projected entangled pair state (PEPS) in dimension $2$ or higher. MPS and PEPS are especially interesting in the context of condensed matter physics. Indeed, it is rigorously proven in some cases and conjectured in others that they are good approximations of ground states of gapped local many-body Hamiltonians \cite{HasA,HasB,AKLV,LVV}. They are therefore used (amongst other) as Ansatz in ground energy computations, allowing for optimization over a tractable number of parameters, even when a large number of particles are involved. They have also been used very successfully to obtain analytical results concerning the understanding and classification of quantum phases of matter (see \cite{Cirac20} for a recent review on this topic). 

This brings us to the general problem we are interested in, which is, very broadly speaking, the following: are common beliefs about MPS and PEPS at least true generically? Or to rephrase it a bit more precisely: which features of MPS and PEPS are typical and which are exceptional? The kind of features that we have in mind include: being the ground state of a parent Hamiltonian which is either gapped or gapless, exhibiting either rapidly decaying or long-range correlations etc. One possible route to tackle this question is to sample MPS and PEPS at random (in a way which should be as physically relevant as possible) and study what are the characteristics that these generically display.

Note that random tensor network states have already been successfully studied in the context of holography. Indeed, tensor network states also provide a natural framework for studying AdS/CFT correspondence. And it turns out that random ones actually reproduce several conjectured properties in this theory \cite{HNQTWY} (see also \cite{Has2} for a pioneer work in this direction). Finally, let us point out that the formalism of tensor networks has also been recently applied in several areas beyond quantum physics: machine learning \cite{Carleo19, Cichoki17b}, PDEs \cite{Bachmayr16, Garcia-Ripoll19}, finance \cite{Mugel20} or signal processing \cite{Garcia-Ripoll19}. The underlying reason is again the fact that tensor networks can parametrize efficiently highly complex multidimensional tensors. 

In this work, we will be looking at a Gaussian model of random tensor network states (described rigorously below). The main reason for this is that it is arguably the first model that one should look at. Indeed, the most important take-home message from standard random matrix theory is that all `sufficiently well-behaved' random matrix models exhibit features which are somehow similar to that of their Gaussian counterpart. It is thus not unreasonable to conjecture that the same should be true for at least some of the properties of random tensor network states that we study. The Gaussian setting has the advantage of allowing for the use of powerful tools to compute averages explicitly and of often displaying provably strong concentration around these averages. So it really serves as a benchmark for other settings, where we might expect averages to be of the same order of magnitude (but where concentration is usually much weaker). Even in terms of proof techniques, it is often the case that results on a given random matrix/tensor model are established by first showing that it can be suitably compared to a Gaussian model, and then using the known result for the latter. Hence, as a first investigation of the question of typical spectral gaps and correlation lengths in random tensor network states, it does make sense to start with a Gaussian model.

\subsection{The model} \hfill\par\smallskip

In the sequel, we will always denote by $d\in\N$ the physical dimension and by $D\in\N$ the bond dimension. We will use the following graphical notation: A vertex with $1$ brown edge and $2p$ grey edges represents a random vector in $\C^d\otimes(\C^D)^{\otimes 2p}$ whose entries are independent complex Gaussians with mean $0$ and variance $1/dD^p$, as exemplified in Figure \ref{fig:diagrams}. 

\begin{figure}[h]
 	\caption{Graphical representation of a Gaussian vector in $\C^d\otimes(\C^D)^{\otimes 2p}$, for $p=1,2$}
 	\label{fig:diagrams}
 	\begin{center}
 		\begin{tikzpicture} [scale=0.8]
 		\begin{scope}[decoration={markings,mark=at position 0.5 with {\arrow{>}}}] 
 		\draw[color=brown] (4,0) -- (4,1); \draw[color=gray] (3,0) -- (5,0);
 		\draw (4.2,0.5) node {{\color{brown} $d$}}; \draw (3.5,-0.25) node {{\color{gray} $D$}}; \draw (4.5,-0.25) node {{\color{gray} $D$}};
 		\draw (4,-0.9) node {$\ket{g}\in\C^d\otimes(\C^D)^{\otimes 2}$};
 		\draw (4,-1.4) node {$g_i\sim\mathcal{N}_{\C}(0,1/dD)$};
 		
 		\draw[color=brown] (9,0) -- (9,1); \draw[color=gray] (8,0) -- (10,0); \draw[color=gray] (8.3,-0.3) -- (9.7,0.3);
 		\draw (9.2,0.5) node {{\color{brown} $d$}}; \draw (8.5,0.2) node {{\color{gray} $D$}}; \draw (9.5,-0.25) node {{\color{gray} $D$}}; \draw (9.9,0.3) node {{\color{gray} $D$}}; \draw (8.05,-0.3) node {{\color{gray} $D$}};
 		\draw (9,-0.9) node {$\ket{h}\in\C^d\otimes(\C^D)^{\otimes 4}$};
 		\draw (9,-1.4) node {$h_i\sim\mathcal{N}_{\C}(0,1/dD^2)$};
 		\end{scope}
 		\end{tikzpicture}
 	\end{center}
 \end{figure}

Two such diagrams next to one another represent the tensor product of the corresponding vectors, while two merged edges represent the contraction on the corresponding indices. And when arrows are added on some edges of a given diagram it means that it has to be viewed as an operator rather than a vector (the direction of the arrows indicating which are the input and output spaces). Finally, we will be facing the case at some point where $3$ copies of $\C^D$ play the same role, and we will for simplicity replace the corresponding $3$ grey edges by $1$ thick grey edge. All these `composition' and `decoration' operations on Gaussian diagrams are illustrated in Figure \ref{fig:diagrams'}. 

\begin{figure}[h] 
	\caption{Composition and decoration of Gaussian diagrams}
	\label{fig:diagrams'}
	\begin{center}
		\begin{tikzpicture} [scale=0.79]
		\begin{scope}[decoration={markings,mark=at position 0.5 with {\arrow{>}}}] 
		\draw[color=brown] (4.5,0) -- (4.5,1); \draw[color=brown] (5.5,0) -- (5.5,1); 
		\draw[color=gray] (4.5,0) -- (5.5,0); \draw[color=gray] (5.5,0) to[in=90,out=-0] (6.2,-0.3); \draw[color=gray] (4.5,0) to[in=90,out=-180] (3.8,-0.3); \draw[color=gray] (5.5,-0.6) to[in=-90,out=0] (6.2,-0.3); \draw[color=gray] (4.5,-0.6) to[in=-90,out=180] (3.8,-0.3); \draw[color=gray] (5.5,-0.6) -- (4.5,-0.6);
		\draw[color=gray] (4.8,-0.3) -- (6.2,0.3);
		\draw (5,-1.1) node {$\ket{u}\in(\C^d)^{\otimes 2}\otimes(\C^D)^{\otimes 2}$};
		
		\draw[postaction={decorate}, color=brown] (10,0) -- (10,1); \draw[postaction={decorate}, color=brown] (11,0) -- (11,1); 
		\draw[color=gray] (10,0) -- (11,0); \draw[color=gray] (11,0) to[in=90,out=-0] (11.7,-0.3); \draw[color=gray] (10,0) to[in=90,out=-180] (9.3,-0.3); \draw[color=gray] (11,-0.6) to[in=-90,out=0] (11.7,-0.3); \draw[color=gray] (10,-0.6) to[in=-90,out=180] (9.3,-0.3); \draw[color=gray] (11,-0.6) -- (10,-0.6);
		\draw[postaction={decorate}, color=gray] (10.3,-0.3) -- (11,0); \draw[postaction={decorate}, color=gray] (11.7,0.3) -- (11,0);
		\draw (10.5,-1.1) node {$U:(\C^D)^{\otimes 2}\longrightarrow(\C^d)^{\otimes 2}$};
		
		\draw[color=brown] (16,0) -- (16,1); \draw[color=brown] (17,0) -- (17,1); 
		\draw[color=gray] (16,0) -- (17,0); \draw[very thick, color=gray] (17,0) to[in=90,out=-0] (17.7,-0.3); \draw[very thick, color=gray] (16,0) to[in=90,out=-180] (15.3,-0.3); \draw[very thick, color=gray] (17,-0.6) to[in=-90,out=0] (17.7,-0.3); \draw[very thick, color=gray] (16,-0.6) to[in=-90,out=180] (15.3,-0.3); \draw[very thick, color=gray] (17,-0.6) -- (16,-0.6);
		\draw[very thick, color=gray] (16.3,-0.3) -- (17,0); \draw[color=gray] (17.7,0.3) -- (17,0);
		\draw (16.5,-1.1) node {$\ket{v}\in(\C^d)^{\otimes 2}\otimes\C^D\otimes\C^{D^3}$};
		
		\draw[postaction={decorate}, color=brown] (21.5,0) -- (21.5,1); \draw[postaction={decorate}, color=brown] (22.5,0) -- (22.5,1); 
		\draw[color=gray] (21.5,0) -- (22.5,0); \draw[very thick, color=gray] (22.5,0) to[in=90,out=-0] (23.2,-0.3); \draw[very thick, color=gray] (21.5,0) to[in=90,out=-180] (20.8,-0.3); \draw[very thick, color=gray] (22.5,-0.6) to[in=-90,out=0] (23.2,-0.3); \draw[very thick, color=gray] (21.5,-0.6) to[in=-90,out=180] (20.8,-0.3); \draw[very thick, color=gray] (22.5,-0.6) -- (21.5,-0.6);
		\draw[postaction={decorate}, very thick, color=gray] (21.8,-0.3) -- (22.5,0); \draw[postaction={decorate}, color=gray] (23.2,0.3) -- (22.5,0);
		\draw (22,-1.1) node {$V:\C^D\otimes\C^{D^3}\longrightarrow(\C^d)^{\otimes 2}$};
		\end{scope}
		\end{tikzpicture}
	\end{center}
\end{figure}

In the $1$-dimensional case, we construct a random translation-invariant MPS (with periodic boundary conditions) in the following way: We pick as $1$-site tensor
\begin{equation} \label{eq:MPS} 
\ket{\chi} := \sum_{x=1}^d \sum_{l,r=1}^D g_{xlr}\ket{xlr} \in \C^d\otimes(\C^D)^{\otimes 2} \, , 
\end{equation}
where the $g_{xlr}$'s are independent complex Gaussians with mean $0$ and variance $1/dD$. We then repeat it on $N$ sites disposed on a circle and contract consecutive bond indices to obtain an $N$-site MPS $\ket{\chi^{N}}\in (\C^d)^{\otimes N}$.
The corresponding transfer operator on $\C^D\otimes\C^D$ is obtained by contracting the $d$-dimensional indices of $\ket{\chi}$ and $\ket{\bar{\chi}}$. It can thus be written as
\begin{equation} \label{eq:transfer-MPS}
T = \frac{1}{d} \sum_{x=1}^d G_x \otimes \bar{G}_x \, ,
\end{equation}
where the $G_x$'s are independent $D\times D$ matrices whose entries are independent complex Gaussians with mean $0$ and variance $1/D$.

This random MPS construction is illustrated in Figures \ref{fig:MPS1} and \ref{fig:MPS2}. The choice of variance $1/dD$ for our Gaussian $1$-site tensor might appear odd at first sight. However, as we will see later, it is precisely with this variance that the resulting random MPS is with high probability close to having norm $1$ (i.e.~to actually being a state). 

\begin{figure}[h]
	\caption{MPS: $1$-site tensor and transfer operator}
	\label{fig:MPS1}
	\begin{center}
		\begin{tikzpicture} [scale=0.8]
		\begin{scope}[decoration={markings,mark=at position 0.5 with {\arrow{>}}}] 
		\draw[color=brown] (4,0) -- (4,1); \draw[color=gray] (3,0) -- (5,0);
		\draw (4,-0.5) node {$\ket{\chi}\in\C^d\otimes(\C^D)^{\otimes 2}$};
		
		\draw[color=brown] (10,0) -- (10,1); \draw[postaction={decorate}, color=gray] (11,0) -- (10,0); \draw[postaction={decorate}, color=gray] (10,0) -- (9,0);
		\draw[postaction={decorate}, color=gray] (11,1) -- (10,1); \draw[postaction={decorate}, color=gray] (10,1) -- (9,1); 
		\draw (10,-0.5) node {$T:\C^D\otimes\C^D\longrightarrow\C^D\otimes\C^D$};
		\end{scope}
		\end{tikzpicture}
	\end{center}
\end{figure}

\begin{figure}[h]
	\caption{Random translation-invariant MPS with periodic boundary conditions}
	\label{fig:MPS2}
	\begin{center}
		\begin{tikzpicture} [scale=0.8]
		\draw[color=brown] (1,0) -- (1,1); \draw[color=brown] (2,0) -- (2,1); \draw[color=brown] (5,0) -- (5,1); \draw[color=brown] (6,0) -- (6,1);
		\draw[color=gray] (0.5,0) -- (2.7,0); \draw[color=gray] (4.3,0) -- (6.5,0);
		\draw[dashed,color=gray] (3,0) -- (4,0);
		\draw[color=gray] (0.5,0) to[in=90,out=-180] (0.2,-0.3); \draw[color=gray] (0.5,-0.6) to[in=-90,out=180] (0.2,-0.3); \draw[color=gray] (6.5,0) to[in=90,out=-0] (6.8,-0.3); \draw[color=gray] (6.5,-0.6) to[in=-90,out=0] (6.8,-0.3);
		\draw[color=gray] (0.5,-0.6) -- (6.5,-0.6);
		\draw[decoration={brace,raise=7pt},decorate]
		(0.8,1) -- node[above=9pt] {$N$} (6.2,1);
		\draw (3.5,-1.2) node {$\ket{\chi^N}\in(\C^d)^{\otimes N}$};
		\end{tikzpicture}
	\end{center}
\end{figure}

Similarly, in the $2$-dimensional case, we construct a random translation-invariant PEPS (with periodic boundary conditions) in the following way: We pick as $1$-site tensor
\begin{equation} \label{eq:PEPS} 
\ket{\chi} := \sum_{x=1}^d \sum_{l,r,a,b=1}^D g_{xlrab}\ket{xlrab} \in \C^d\otimes(\C^D)^{\otimes 4} \, , 
\end{equation}
where the $g_{xlrab}$'s are independent complex Gaussians with mean $0$ and variance $1/dD^2$. We then repeat it on $N^2$ sites disposed on a torus and contract consecutive bond indices (in both row and column directions) to obtain an $N^2$-site PEPS $\ket{\chi^{N}_N}\in(\C^d)^{\otimes N^2}$. This means that its contraction on an $N$-site column $\ket{\chi_N}\in (\C^d\otimes(\C^D)^{\otimes 2})^{\otimes N}$ is
\begin{equation} \label{eq:PEPS'} 
\ket{\chi_N} := \frac{1}{d^{N/2}} \sum_{x_1,\ldots,x_N=1}^d \frac{1}{D^{N}} \sum_{l_1,r_1,\ldots,l_N,r_N=1}^D \left( \sum_{a_1,\ldots,a_N=1}^D g_{x_1l_1r_1a_Na_1}\cdots g_{x_Nl_Nr_Na_{N-1}a_N} \right) \ket{x_1l_1r_1\cdots x_Nl_Nr_N} \, . 
\end{equation}
The corresponding transfer operator on $(\C^D\otimes\C^D)^{\otimes N}$ is obtained by contracting the $d$-dimensional indices of $\ket{\chi_N}$ and $\ket{\bar{\chi}_N}$. It can thus be written as
\begin{equation} \label{eq:transfer-PEPS}
T_N = \frac{1}{D^N} \sum_{a_1,b_1,\ldots,a_N,b_N=1}^D \frac{1}{d^N} \sum_{x_1,\ldots,x_N=1}^d G_{a_Na_1x_1} \otimes \bar{G}_{b_Nb_1x_1}\otimes \cdots\otimes G_{a_{N-1}a_Nx_N} \otimes \bar{G}_{b_{N-1}b_Nx_N} \, ,
\end{equation}
where the $G_{a_{i-1}a_ix_i}$'s are independent $D\times D$ matrices whose entries are independent complex Gaussians with mean $0$ and variance $1/D$.

This random PEPS construction is illustrated in Figure \ref{fig:PEPS}. Just as in the MPS case, the seemingly odd choice of variance $1/dD^2$ for our Gaussian $1$-site tensor is only to guarantee that the resulting random PEPS is with high probability close to having norm $1$.

\begin{figure}[h]
	\caption{PEPS: $1$-site tensor, $N$-site column tensor and transfer operator}
	\label{fig:PEPS}
	\begin{center}
		\begin{tikzpicture} [scale=0.8]
		\begin{scope}[decoration={markings,mark=at position 0.5 with {\arrow{>}}}] 
		\draw[color=brown] (-3,0) -- (-3,1); \draw[color=gray] (-4,0) -- (-2,0); \draw[color=gray] (-3.7,-0.3) -- (-2.3,0.3);
		\draw (-3,-0.8) node {$\ket{\chi}\in\C^d\otimes(\C^D)^{\otimes 4}$};
		
		\draw[color=brown] (1,0) -- (1,1); \draw[color=brown] (2,0) -- (2,1); \draw[color=brown] (5,0) -- (5,1); \draw[color=brown] (6,0) -- (6,1);
		\draw[color=gray] (0.5,0) -- (2.7,0); \draw[color=gray] (4.3,0) -- (6.5,0);
		\draw[dashed,color=gray] (3,0) -- (4,0);
		\draw[color=gray] (0.5,0) to[in=90,out=-180] (0.2,-0.3); \draw[color=gray] (0.5,-0.6) to[in=-90,out=180] (0.2,-0.3); \draw[color=gray] (6.5,0) to[in=90,out=-0] (6.8,-0.3); \draw[color=gray] (6.5,-0.6) to[in=-90,out=0] (6.8,-0.3);
		\draw[color=gray] (0.5,-0.6) -- (6.5,-0.6);
		\draw[color=gray] (0.3,-0.3) -- (1.7,0.3); \draw[color=gray] (1.3,-0.3) -- (2.7,0.3); \draw[color=gray] (4.3,-0.3) -- (5.7,0.3); \draw[color=gray] (5.3,-0.3) -- (6.7,0.3);
		\draw[decoration={brace,raise=7pt},decorate]
		(0.8,1) -- node[above=9pt] {$N$} (6.2,1);
		\draw (3.5,-1.2) node {$\ket{\chi_N}\in(\C^d\otimes(\C^D)^{\otimes 2})^{\otimes N}$};
		
		\draw[color=brown] (10,0) -- (10,1); \draw[color=brown] (11,0) -- (11,1); \draw[color=brown] (14,0) -- (14,1); \draw[color=brown] (15,0) -- (15,1);
		\draw[color=gray] (9.5,0) -- (11.7,0); \draw[color=gray] (13.3,0) -- (15.5,0);
		\draw[dashed,color=gray] (12,0) -- (13,0);
		\draw[color=gray] (9.5,0) to[in=90,out=-180] (9.2,-0.3); \draw[color=gray] (9.5,-0.6) to[in=-90,out=180] (9.2,-0.3); \draw[color=gray] (15.5,0) to[in=90,out=-0] (15.8,-0.3); \draw[color=gray] (15.5,-0.6) to[in=-90,out=0] (15.8,-0.3);
		\draw[color=gray] (9.5,-0.6) -- (15.5,-0.6);
		\draw[postaction={decorate}, color=gray] (10.7,0.3) -- (10,0); \draw[postaction={decorate}, color=gray] (10,0) -- (9.3,-0.3); \draw[postaction={decorate}, color=gray] (11.7,0.3) -- (11,0); \draw[postaction={decorate}, color=gray] (11,0) -- (10.3,-0.3); \draw[postaction={decorate}, color=gray] (14.7,0.3) -- (14,0); \draw[postaction={decorate}, color=gray] (14,0) -- (13.3,-0.3); \draw[postaction={decorate}, color=gray] (15.7,0.3) -- (15,0); \draw[postaction={decorate}, color=gray] (15,0) -- (14.3,-0.3);
		\draw[color=gray] (9.5,1) -- (11.7,1); \draw[color=gray] (13.3,1) -- (15.5,1);
		\draw[dashed,color=gray] (12,1) -- (13,1);
		\draw[color=gray] (9.5,1) to[in=-90,out=180] (9.2,1.3); \draw[color=gray] (9.5,1.6) to[in=90,out=-180] (9.2,1.3); \draw[color=gray] (15.5,1) to[in=-90,out=0] (15.8,1.3); \draw[color=gray] (15.5,1.6) to[in=90,out=-0] (15.8,1.3);
		\draw[color=gray] (9.5,1.6) -- (15.5,1.6);
		\draw[postaction={decorate}, color=gray] (10.7,1.3) -- (10,1); \draw[postaction={decorate}, color=gray] (10,1) -- (9.3,0.7); \draw[postaction={decorate}, color=gray] (11.7,1.3) -- (11,1); \draw[postaction={decorate}, color=gray] (11,1) -- (10.3,0.7); \draw[postaction={decorate}, color=gray] (14.7,1.3) -- (14,1); \draw[postaction={decorate}, color=gray] (14,1) -- (13.3,0.7); \draw[postaction={decorate}, color=gray] (15.7,1.3) -- (15,1); \draw[postaction={decorate}, color=gray] (15,1) -- (14.3,0.7);
		\draw[decoration={brace,raise=7pt},decorate]
		(9.8,1.7) -- node[above=9pt] {$N$} (15.2,1.7);
		\draw (12.5,-1.2) node {$T_N:(\C^D\otimes\C^D)^{\otimes N}\longrightarrow(\C^D\otimes\C^D)^{\otimes N}$};
		\end{scope}
		\end{tikzpicture}
	\end{center}
\end{figure}

Later on in the paper, to simplify notation, we may sometimes write indices running from $1$ to $N$ modulo $N$, i.e.~identify index $N+1$ with index $1$ and index $0$ with index $N$.

In what follows, we will denote by $\ket{\psi}\in \C^D\otimes\C^D$ the maximally entangled unit vector. Letting $\{ \ket{1},\ldots,\ket{D} \}$ being the canonical orthonormal basis of $\C^D$, the latter is defined as
\[ \ket{\psi} := \frac{1}{\sqrt{D}}\sum_{\alpha=1}^{D} \ket{\alpha\alpha} \, . \]

An important property of MPS and PEPS is injectivity, and more generally normality. Let us start by recalling the definition of these terms. 

\begin{definition}[Injectivity and normality]
	An MPS, resp.~a PEPS, is called injective if its $1$-site tensor, viewed as a linear map from the bond space to the physical space, i.e.~from $(\C^D)^{\otimes 2}$ to $\C^d$, resp.~from $(\C^D)^{\otimes 4}$ to $\C^d$, is injective. It is called normal if there exists an integer $L$, resp.~integers $K,L$, such that after blocking together segments of $L$ sites, resp.~rectangles of $K\times L$ sites, it becomes injective.
\end{definition}

Let us see what the above definition means concretely in our case. The random MPS defined by the $1$-site tensor $\ket{\chi}\in\C^d\otimes(\C^D)^{\otimes 2}$ of equation \eqref{eq:MPS} is normal if there exists $L$ such that the following map $\widetilde{\chi}^L:(\C^D)^{\otimes 2}\longrightarrow(\C^d)^{\otimes L}$ is injective:
\[ \widetilde{\chi}^L:= \sum_{x_1,\ldots,x_L=1}^d \sum_{a_1,\ldots,a_{L+1}=1}^D g_{x_1a_1a_2}\cdots g_{x_La_La_{L+1}}  \ketbra{x_1\cdots x_L}{a_1a_{L+1}} \, . \]
And the random PEPS defined by the $1$-site tensor $\ket{\chi}\in\C^d\otimes(\C^D)^{\otimes 4}$ of equation \eqref{eq:PEPS} is normal if there exists $K,L$ such that the following map $\widetilde{\chi}_K^L:(\C^D)^{\otimes 2(K+L)}\longrightarrow(\C^d)^{\otimes KL}$ is injective:
\[ \widetilde{\chi}_K^L:= \sum_{\substack{x_{i,j}=1 \\ 1\leq i\leq K \\ 1\leq j\leq L}}^d \sum_{\substack{a_{i,j},b_{i,j}=1 \\ 1\leq i\leq K+1 \\ 1\leq j\leq L+1}} ^D \prod_{\substack{1\leq i\leq K \\ 1\leq j\leq L}} g_{x_{i,j}a_{i,j}a_{i,j+1}b_{i,j}b_{i+1,j}}  \ketbra{x_{1,1}\cdots x_{K,L}}{a_{1,1}a_{1,L+1}\cdots a_{K,1}a_{K,L+1}b_{1,1}b_{K+1,1}\cdots b_{1,L}b_{K+1,L}} \, . \]
If $L=1$, resp.~$K,L=1$, then the random MPS, resp.~PEPS, is injective.

\begin{fact} \label{fact:irreducible}
Our random MPS and PEPS, with $1$-site tensor as defined by equations \eqref{eq:MPS} and \eqref{eq:PEPS} respectively, are almost surely normal. 
%More precisely, for any $L$ such that $d^L\geq D^2$, our random MPS is almost surely injective after blocking $L$ sites, and for any $K,L$ such that $d^{KL}\geq D^{2(K+L)}$, our random PEPS is almost surely injective after blocking $KL$ sites. 

Additionally, we also have more precisely that, if $d\geq D^2$, resp.~$d\geq D^4$, then our random MPS, resp.~PEPS, is almost surely injective.
\end{fact}

\begin{proof}
	The first general statement follows from \cite[Theorem 3.4]{MSV}, which says that the set of $1$-site tensors on $\C^d\otimes(\C^D)^{\otimes 2}$, resp.~$\C^d\otimes(\C^D)^{\otimes 4}$, giving rise to a non-normal MPS, resp.~PEPS, is a proper sub-manifold of the space of all $1$-site tensors.
	
	The second more specific statement is immediate once noticed that the $1$-site operators $\widetilde{\chi}^1:(\C^D)^{\otimes 2}\longrightarrow\C^d$ and $\widetilde{\chi}_1^1:(\C^D)^{\otimes 4}\longrightarrow\C^d$ are simply Gaussian operators:
	\[ \widetilde{\chi}^1= \sum_{x=1}^d\sum_{l,r=1}^D g_{xlr}\ketbra{x}{lr} \ \ \text{and} \ \ \widetilde{\chi}_1^1= \sum_{x=1}^d\sum_{l,r,a,b=1}^D g_{xlrab}\ketbra{x}{lrab} \, . \]
	Hence they are almost surely injective as soon as their output dimension is larger than their input dimension.
\end{proof}

Having identified the dimensional regime where our random MPS and PEPS are almost surely injective (namely $d\geq D^2$ and $d\geq D^4$) will be of great importance for us later on. Indeed, if an MPS or a PEPS is injective, then \cite[Theorem 3]{CPGVW} tells us that it is guaranteed to be the unique ground state of its parent Hamiltonian, which additionally takes the simplest possible form. This fact will be particularly useful to us in Section \ref{sec:parent-hamiltonian}, where we need to estimate the typical spectral gap of the random MPS and PEPS parent Hamiltonians (and where it is explained in detail how these local Hamiltonians are constructed). 

What is more, we know from \cite[Footnote 1]{CPGSW} that the normality of an MPS or a PEPS implies the irreducibility of the corresponding transfer operator. Hence, as a consequence of Fact \ref{fact:irreducible} we get that $T$ as defined in equation \eqref{eq:transfer-MPS} and $T_N$ as defined in equation \eqref{eq:transfer-PEPS} are almost surely irreducible. This fact will be crucial for us in Section \ref{sec:transfer-operator}, where we have to estimate the typical spectral gap of the random MPS and PEPS transfer operators.

\subsection{Summary of our main results and connections with previous works} \hfill\par\smallskip

As already said, our goal in this work is to study what are the features that our models of random translation-invariant MPS and PEPS, with $1$-site random tensors respectively defined by equations \eqref{eq:MPS} and \eqref{eq:PEPS}, typically exhibit. We are particularly interested in understanding what is their typical correlation length. We try to tackle this question through two distinct strategies. Both ultimately boil down to determining what is the typical spectral gap of a random operator associated to our random tensor network states: either their parent Hamiltonian or their transfer operator. The first approach, through the study of the parent Hamiltonian has the advantage of yielding results which are true for any system size for both MPS and PEPS, but the disadvantage of applying only to injective MPS and PEPS. On the contrary, the approach through the transfer operator has the advantage of being valid beyond the injectivity regime, but the disadvantage of giving results which depend on the system size in the PEPS case. So they somehow complement each other.

The overall merit of our various results could perhaps be summarized as follows: In the $1$-dimensional case, none of them is qualitatively surprising. What is interesting is that we are able to prove quantitative lower bounds on the generic spectral gap of the random parent Hamiltonian and transfer operator, in terms of the physical and bond dimensions (which translate into quantitative upper bounds on the generic correlation length of the random MPS). In the $2$-dimensional case though, nothing was known (and there was not even a consensus on what should be expected). So already the qualitative result that, in some range of physical and bond dimensions, random PEPS generically exhibit exponential decay of correlation is relevant in its own.

First, in Section \ref{sec:parent-hamiltonian}, we place ourselves in the injectivity regime of our random tensor networks, i.e.~$d>D^2$ in the case of MPS and $d>D^4$ in the case of PEPS, where it makes sense to talk about their canonical parent Hamiltonian. The latter is precisely defined in Section \ref{sec:parent-hamiltonian}. For now, let us just say that it is a $2$-local (nearest-neighbour interaction) Hamiltonian which is frustration-free and has the MPS or PEPS as unique ground state. We are able to show that, at least in a `super-injectivity' regime, this parent Hamiltonian has with high probability a large spectral gap. More precisely, we obtain the following result, which appears as Theorems \ref{th:gap-H} (for the case of MPS) and \ref{th:gap-H-PEPS} (for the case of PEPS).

\begin{theorem} \label{th:gap-H-summary}
Let $d,D,N\in\N$. Denote by $H_{MPS}$, resp.~$H_{PEPS}$, the parent Hamiltonian of our random MPS, resp.~PEPS. If $d\geq D^{10+\epsilon}$ for some $\epsilon>0$, then
\[ \P \left( \Delta\left(H_{MPS}\right) \geq 1-\frac{C}{D^{\epsilon/2}} \right) \geq 1-e^{-cD^2} \, , \]
and if $d\geq D^{26+\epsilon}$ for some $\epsilon>0$, then
\[ \P \left( \Delta\left(H_{PEPS}\right) \geq 1-\frac{C}{D^{\epsilon/2}} \right) \geq 1-e^{-cD^4} \, , \]
where $C,c>0$ are universal constants.
\end{theorem} 

In words, Theorem \ref{th:gap-H-summary} tells us the following: as $d,D$ grow, with $d>D^{10}$, resp.~$d>D^{26}$, it holds that, with probability going to $1$, the spectral gap of $H_{MPS}$, resp.~$H_{PEPS}$, is going to (at least) $1$. Let us point out immediately that the constraints on the scaling of $d$ with respect to $D$ are likely to be far from optimal. It is indeed clear when looking at the whole reasoning in Section \ref{sec:parent-hamiltonian} that, in several steps, we use general bounds that might be rough in our particular case. It seems however that getting better ones could require a very careful analysis, that we do not pursue here. We build our whole reasoning upon the well-known fact that, in some range of parameters, Wishart matrices (suitably rescaled) can be approximated by the identity in a strong sense. 

Let us also mention that the question of whether random local Hamiltonians are generically gapped or gapless has recently been studied in \cite{Mov} and \cite{Lem}, with a quite different perspective than ours. In both works, the local terms composing the Hamiltonian are picked at random, while what we pick at random is the ground state of the Hamiltonian. In \cite{Mov} local terms are sampled independently, and it is shown that the obtained Hamiltonian is gapless with probability $1$ in the thermodynamic limit. In \cite{Lem}, on the contrary, only one local term is sampled and repeated, hence imposing translation-invariance of the obtained Hamiltonian, which is shown to be gapped with strictly positive probability in the thermodynamic limit. This latter setting is in fact very close to ours: in addition to being translation-invariant, the random local Hamiltonian which is studied is frustration-free (a characteristics that the parent Hamiltonian of an MPS or PEPS has by definition). And it indeed leads to a similar conclusion. 

As a consequence of Theorem \ref{th:gap-H-summary} we have that, in this same regime, our random MPS and PEPS exhibit, with probability going to $1$ as $d,D$ grow, exponential decay of correlations. This result appears as Theorems \ref{th:dc1-MPS} and \ref{th:dc1-PEPS} in Section \ref{sec:decay-correlations-1}. Even though already interesting, it still has two weaknesses. First of all, it only applies to the injectivity regime. Indeed, even if the constraints $d>D^{10}$ for MPS and $d>D^{26}$ for PEPS could probably be improved, there is no way that we can say anything about the range $d<D^2$ for MPS and $d<D^4$ for PEPS via this approach (which consists in showing exponential decay of correlations in a tensor network state through showing that its parent Hamiltonian is gapped). Second of all, it cannot give anything stronger than a constant correlation length. In our specific case, we can actually improve the latter point by using the powerful recent result of \cite{HHKL}. Indeed, we have more than a lower bound on the typical spectral gap of the parent Hamiltonian, namely an upper bound on the typical commutator of the local terms composing it. And this enables us to prove that, in fact, the correlation length typically decays as $1/\log D$ (see Theorems \ref{th:dc1'-MPS} and \ref{th:dc1'-PEPS} in Section \ref{sec:decay-correlations-1}).

However, to remedy the first problem it is necessary to look for another possible way to prove typical exponential decay of correlations in our random MPS and PEPS. A very natural one is through showing that their associated transfer operator is typically gapped, a result which is moreover of independent interest. This is what we do in Section \ref{sec:transfer-operator}, ultimately obtaining Theorems \ref{th:gap-MPS} (for the case of MPS) and \ref{th:gap-PEPS} (for the case of PEPS), which are summarized below.

\begin{theorem} \label{th:gap-T-summary}
Let $d,D\in\N$. Denote by $T_{MPS}$ the transfer operator of our random MPS. Then, 
\[ \P\left( \Delta\left(T_{MPS}\right) \geq 1 - \frac{C}{d^{1/2}} \right) \geq 1-e^{-cD} \, , \]
where $C,c>0$ are universal constants.

Let $d,D,N\in\N$. Denote by $T_{PEPS}$ the transfer operator of our random PEPS. If $d\simeq N^{\alpha}$ and $D\simeq N^{\beta}$ with $\alpha>8$ and $(\alpha+1)/3<\beta<(\alpha-2)/2$, then
\[  \P\left( \Delta\left(T_{PEPS}\right) \geq 1-\frac{C}{N^{\alpha/2-\beta-1}} \right) \geq 1-e^{-cN^{3\beta-\alpha}} \, , \]
where $C,c>0$ are universal constants.
\end{theorem}

Let us comment first on the MPS case, where the result is valid for any $d$ and $D$. The spectral gap question that we investigate has already been studied in \cite{Has1,Pis2} and \cite{GJN} on different random models (not all of them motivated by the study of random MPS), involving unitary and truncated unitary rather than Gaussian operators. In these three pieces of work the emphasis is put on having as tight as possible average results, while we mostly care about the order of magnitude but want to show that it is generic. It however remains that all approaches yield one similar result, namely an expected spectral gap of the considered random transfer operator larger than $1-C/\sqrt{d}$. Our proof strategy, to first lower bound the expected spectral gap of $T_{MPS}$, follows closely that of \cite{Pis2} and \cite{GJN}. An additional difficulty in our case comes from the fact that we are dealing with a non-normal random matrix model, which makes the analysis of its spectrum significantly more delicate. In order to then show that this lower bound is actually also typical we make use of a slightly refined version of the standard Gaussian concentration inequality. 

Concerning the PEPS case, we see that we are able to prove that the transfer operator is generically gapped only in the regime where $d,D$ grow polynomially with $N$. While this is to be expected for $D$, it seems much less natural for $d$ though. This scaling can however easily be enforced by a so-called blocking procedure, namely: We start from a square lattice with $\underline{N}\times\underline{N}$ sites, where $\underline{N}:=N\sqrt{\log N}$, each having physical dimension $\underline{d}$ and bond dimension $\underline{D}:=\underline{D}'^{\sqrt{\log N}}$. We then redefine $1$ site as being a square of $\sqrt{\log N}\times\sqrt{\log N}$ sites. We thus obtain a square lattice with $N\times N$ sites, each having physical dimension $d:=\underline{d}^{\log N}$ and bond dimension $D:=\underline{D}^{\sqrt{\log N}}=\underline{D}'^{\log N}$. Hence indeed, setting $\alpha:=\log\underline{d}$ and $\beta:=\log\underline{D}'$, we have $d=N^{\alpha}$ and $D=N^{\beta}$. Finally, for the parameters $\alpha,\beta$ to be in the valid range, we just have to impose on the parameters $\underline{d},\underline{D}'$ that they satisfy $\underline{d}>e^8$ and $(e\underline{d})^{1/3}<\underline{D}'<\underline{d}^{1/2}/e$. Let us emphasize here that the proof techniques to prove a lower bound on the typical spectral gap of $T_{PEPS}$ are, as far as we are aware of, essentially new. The basic idea is some kind of recursion procedure that uses the MPS results as building blocks. 

As a quite straightforward consequence of Theorem \ref{th:gap-T-summary} we obtain that our random MPS and PEPS typically exhibit exponential decay of correlations at a provably high rate, and in a dimension regime that goes beyond the injectivity one. This result appears as Theorems \ref{th:dc2-MPS} and \ref{th:dc2-PEPS} in Section \ref{sec:decay-correlations-2}.

Theorem \ref{th:gap-T-summary} has several other implications, some of which are studied in Section \ref{sec:implications}. In particular, we draw the path towards constructions of random quantum expanders and random dissipative evolutions.

Let us make some final comments on the extra technical difficulties when trying to extend $1$-dimensional results to $2$-dimensional ones. Our first approach, based on looking at the parent Hamiltonian, has the great advantage of not being that much more complicated for PEPS than for MPS. The reason behind this is that, in both cases, the terms composing the parent Hamiltonian have the same locality (namely they just involve $2$ sites). The only thing that makes the PEPS case slightly more subtle is that the $4$ terms which act non-trivially on a given site are of $2$ different kinds ($2$ identical `horizontal' terms and $2$ identical `vertical' terms), while in the MPS case there are only $2$ identical terms which act non-trivially on a given site. On the other hand, our second approach, based on looking at the transfer operator, is way more difficult to go through for PEPS than for MPS. Indeed, the transfer operator of a PEPS, contrary to that of an MPS, depends on the system size (since it is not constructed from its $1$-site tensor but from its $N$-site column tensor). Hence, while in the MPS case any statement about the transfer operator is automatically valid for any system size, this is not true anymore in the PEPS case. 

\subsection{A few key results in Gaussian concentration} \hfill\par\smallskip

Let us conclude this introductory part with two technical results that we will be using in multiple occasions throughout this paper. The first one is the celebrated concentration inequality for Lipschitz functions on Gaussian space, which was proved independently in \cite{Bor} and \cite{ST}. The second one is a local version of this concentration inequality, which is useful when the considered function does not have a small Lipschitz constant on the whole Gaussian space but only on a large measure subset, and which was established in \cite{ASW}.

\begin{theorem}[Gaussian concentration inequality, global version \cite{Bor,ST}] \label{th:g-global}
Let $f:\C^n\longrightarrow\R$ be $L$-Lipschitz (with respect to the Euclidean norm). For $g\in\C^n$ a Gaussian vector with mean $0$ and variance $\sigma^2$, we have
\[ \forall\ \epsilon>0,\ \P( f(g) \gtrless \E f \pm \epsilon) \leq e^{-\epsilon^2/\sigma^2L^2} \, . \] 
\end{theorem}

\begin{theorem}[Gaussian concentration inequality, local version \cite{ASW}] \label{th:g-local}
Let $\Omega\subset\C^n$ and let $f:\C^n\longrightarrow\R$ be $L$-Lipschitz on $\Omega$ (with respect to the Euclidean norm). For $g\in\C^n$ a Gaussian vector with mean $0$ and variance $\sigma^2$, we have
\[  \forall\ \epsilon>0,\ \P( f(g) \gtrless \E f \pm \epsilon) \leq e^{-\epsilon^2/\sigma^2L^2} + \P(g\notin\Omega) \, . \]
\end{theorem}

Finally, we will also use several times the fact that (suitably rescaled) large Wishart matrices of large enough parameter are with high probability close to the identity. More precisely, we will need the result below, which can be found in \cite[Proposition 6.33]{AS}.

\begin{theorem}[Strong convergence of Wishart matrices \cite{AS}] \label{th:Wishart}
	Fix $n,s\in\N$. Let $W$ be an $n\times n$ Wishart matrix of parameter $s$ (i.e.~$W=GG^*$ where $G$ is an $n\times s$ matrix whose entries are independent complex Gaussians with mean $0$ and variance $1$). Then,
	\[ \P\left( \left\|\frac{1}{s}W-\Id\right\|_{\infty} > 6\sqrt{\frac{n}{s}} \right) \leq 2e^{-n/4} \, . \]
\end{theorem}

\section{Typical spectral gap of the parent Hamiltonian of random MPS and PEPS}
\label{sec:parent-hamiltonian}

In this section we want to show that, in the dimensional regime where our random MPS and PEPS are injective, their canonical parent Hamiltonians are typically gapped. We are actually only able to establish this in a `super-injective' regime, which might be an artefact of our proof techniques. We first study the MPS case in Section \ref{sec:parent-hamiltonian-MPS} and then follow step by step the same reasoning for the PEPS case in Section \ref{sec:parent-hamiltonian-PEPS}. The final results appear as Theorems \ref{th:gap-H} and \ref{th:gap-H-PEPS} respectively. 

Let us just briefly emphasize again a point that we have already raised in the introduction. The corollaries of these results in terms of typical correlation length (see Section \ref{sec:decay-correlations-1}) are very powerful in the PEPS case but not as much as they could be in the MPS case. Indeed, for the latter, our second approach through characterizing the typical spectral gap of the transfer operator (see Section \ref{sec:transfer-operator}) yields stronger typical correlation length consequences (see Section \ref{sec:decay-correlations-2}). The interest of nevertheless carrying through the MPS analysis here is two-fold: First, the result on the typical spectral gap of the parent Hamiltonian is interesting in its own. Second, the overall reasoning can really be seen as a `warm up' before getting into the slightly more intricate to follow, but in the end entirely analogous, PEPS reasoning.

Before we move on, let us just state here one quite straightforward fact concerning the maximum norm increase of a matrix under realignment, which we will make use of several times later on. We recall that the realignment (in the canonical tensor orthonormal bases) of an $nm\times nm$ matrix $M$ is the $n^2\times m^2$ matrix $\mathcal{R}(M)$ defined by
\[ \forall\ 1\leq i,j \leq n,\ 1\leq k,l\leq m,\ \mathcal{R}(M)_{ij,kl} := M_{ik,jl} \, . \]

\begin{fact} \label{fact:realign}
	Let $n,m\in\N$. For any $nm\times nm$ matrix $M$,
	\[ \|\mathcal{R}(M)\|_{\infty} \leq \min(n,m)\|M\|_{\infty} \, . \]
\end{fact}

\subsection{The case of MPS} \hfill\par\smallskip
\label{sec:parent-hamiltonian-MPS}

We assume here that $d>D^2$, so that our random MPS $\ket{\chi^N}$ is injective with probability $1$ (see Fact \ref{fact:irreducible}). We know from \cite{CPGVW} that, in this case, there exists a canonical way of constructing a $2$-local frustration-free Hamiltonian on $(\C^d)^{\otimes N}$ whose unique ground state is $\ket{\chi^N}$, with ground energy $0$. We call it the parent Hamiltonian of $\ket{\chi^N}$ and denote it by $H_{\chi}$. We now want to show that this random Hamiltonian $H_{\chi}$ typically has a large (lower) spectral gap $\Delta(H_{\chi})$. Note that since the smallest eigenvalue of $H_{\chi}$ is $0$, $\Delta(H_{\chi})$ is actually nothing else than the second smallest eigenvalue of $H_{\chi}$. 

Let us first recall how the Hamiltonian $H_{\chi}$ is constructed. Define $V_{\chi}\subset \C^d\otimes\C^d$ as
\[ V_{\chi} := \mathrm{Span} \left\{ \sum_{x_1,x_2=1}^d \Tr(G_{x_1}G_{x_2}M)\ket{x_1x_2} : M\ D\times D\ \text{matrix} \right\} \, , \]
where the $G_x$'s are the $D\times D$ matrices appearing in equation \eqref{eq:transfer-MPS} defining the transfer operator $T$ associated to $\ket{\chi}$.
Equivalently,
\begin{equation} \label{eq:W-chi} 
V_{\chi} := \mathrm{Span} \left\{ \ket{\chi_{\upsilon}} : \ket{\upsilon}\in\C^D\otimes\C^D \right\} \, ,
\end{equation}
where $\ket{\chi_{\upsilon}}\in\C^d\otimes\C^d$ is the $2$-site MPS having $\ket{\chi}$ as $1$-site tensor and $\ket{\upsilon}$ as boundary condition.
By construction we always have $\mathrm{dim}(V_{\chi})\leq D^2<d^2$, so that $\mathrm{dim}(V_{\chi}^{\perp})\geq d^2-D^2>0$. And in our case we actually have $\mathrm{dim}(V_{\chi})=D^2$ with probability $1$. Denoting by $\Pi$ the projector on $V_{\chi}$, the parent Hamiltonian $H_{\chi}$ of $\ket{\chi^N}$ is then defined as
\begin{equation} \label{eq:H-parent} 
H_{\chi} := \sum_{i=1}^N \Pi^{\perp}_{i,i+1} \otimes\Id_{1,\ldots,i-1,i+2,\ldots,N}  \, . 
\end{equation}

\subsubsection{Approximating the local ground space projectors} \hfill\par\smallskip

The following $1$-site operator $W$ will appear repeatedly in our subsequent computations:
\begin{center}
	\begin{tikzpicture} [scale=0.6]
	\begin{scope}[decoration={markings,mark=at position 0.5 with {\arrow{>}}}] 
	\draw[postaction={decorate}, color=gray] (3,-0.5) to[in=-180,out=90] (4,0); \draw[postaction={decorate}, color=gray] (5,-0.5) to[in=-0,out=90] (4,0); \draw[postaction={decorate}, color=gray] (4,1) to[in=-90,out=180] (3,1.5); \draw[postaction={decorate}, color=gray] (4,1) to[in=-90,out=0] (5,1.5);
	\end{scope}
	\draw[color=brown] (4,0) -- (4,1); 
	\draw (4,-1) node {$W:(\C^D)^{\otimes 2}\longrightarrow(\C^D)^{\otimes 2}$};
	\end{tikzpicture}
\end{center}

What will be crucial for us is that, in the range $d>D^2$ that we are interested in, $W$ (suitably renormalized) is generically close to the identity. Indeed, as immediate corollary of Theorem \ref{th:Wishart} (applied with $n=D^2$ and $s=d$) we have the result below.

\begin{proposition} \label{prop:norm-W}
	Let $d\geq D^{2\tau}$, for some $\tau>1$. Then, there exist universal constants $c,C>0$ such that
	\[ \P\left( \left\|DW-\Id\right\|_{\infty} \leq \frac{C}{D^{\tau-1}} \right) \geq 1-e^{-cD^2} \, . \]
\end{proposition}

Let us now define the $2$-site operators $Q,P,M$ as
\begin{center}
	\begin{tikzpicture} [scale=0.6]
	\begin{scope}[decoration={markings,mark=at position 0.5 with {\arrow{>}}}] 
	\draw[postaction={decorate}, color=brown] (5,0) -- (5,1); \draw[postaction={decorate}, color=brown] (6,0) -- (6,1); \draw[color=gray] (5,0) -- (6,0); \draw[postaction={decorate}, color=gray] (4,-0.5) to[in=-180,out=90] (5,0); \draw[postaction={decorate}, color=gray] (7,-0.5) to[in=-0,out=90] (6,0); 
	\draw (5.5,-1) node {$Q:(\C^D)^{\otimes 2}\longrightarrow(\C^d)^{\otimes 2}$};
	
	\draw[postaction={decorate}, color=brown] (15,0) -- (15,1); \draw[postaction={decorate}, color=brown] (16,0) -- (16,1); \draw[postaction={decorate}, color=brown] (15,-2) -- (15,-1); \draw[postaction={decorate}, color=brown] (16,-2) -- (16,-1); \draw[color=gray] (15,0) -- (16,0); \draw[color=gray] (15,-1) -- (16,-1); \draw[color=gray] (14.2,-0.5) to[in=-180,out=90] (15,0); \draw[color=gray] (14.2,-0.5) to[in=180,out=-90] (15,-1); \draw[color=gray] (16.8,-0.5) to[in=-0,out=90] (16,0); \draw[color=gray] (16.8,-0.5) to[in=0,out=-90] (16,-1);  
	\draw (15.5,-2.5) node {$P=QQ^*:(\C^d)^{\otimes 2}\longrightarrow(\C^d)^{\otimes 2}$};
	
	\draw[color=brown] (25,-1) -- (25,0); \draw[color=brown] (26,-1) -- (26,0); \draw[color=gray] (25,-1) -- (26,-1); \draw[postaction={decorate}, color=gray] (24,-1.5) to[in=-180,out=90] (25,-1); \draw[postaction={decorate}, color=gray] (27,-1.5) to[in=-0,out=90] (26,-1); \draw[color=gray] (25,0) -- (26,0); \draw[postaction={decorate}, color=gray] (25,0) to[in=-90,out=180] (24,0.5); \draw[postaction={decorate}, color=gray] (26,0) to[in=-90,out=0] (27,0.5);
	\draw (25.5,-2) node {$M=Q^*Q:(\C^D)^{\otimes 2}\longrightarrow(\C^D)^{\otimes 2}$};
	\end{scope}
	\end{tikzpicture}
\end{center}

We will first show that, in the range $d>D^6$, just as $W$, $M$ (suitably renormalized) is close to the identity.

\begin{proposition} \label{prop:norm-M}
	Let $d\geq D^{2\tau}$, for some $\tau>3$. Then, there exist universal constants $c,C>0$ such that
	\[ \P\left( \|DM-\Id\|_{\infty} \leq \frac{C}{D^{\tau-3}} \right) \geq 1-e^{-cD^2} \, . \]
\end{proposition}

\begin{proof}
	Observe that $M=\mathcal{R}\left(\mathcal{R}(W)\mathcal{R}(W)\right)$, while $\Id=D\mathcal{R}(\ketbra{\psi}{\psi})$, i.e.~equivalently $\ketbra{\psi}{\psi}=\mathcal{R}(\Id)/D$. Thus,
	\begin{align*}
	\|DM-\Id\|_{\infty} & = D \left\| \mathcal{R}\left(\mathcal{R}(W)\mathcal{R}(W)-\ketbra{\psi}{\psi}\right) \right\|_{\infty} \\
	& \leq D^2 \left\| \mathcal{R}(W)\mathcal{R}(W)-\ketbra{\psi}{\psi} \right\|_{\infty} \\
	& = D^2 \left\| \left(\mathcal{R}(W)-\ketbra{\psi}{\psi}\right) \ketbra{\psi}{\psi} + \mathcal{R}(W) \left(\mathcal{R}(W)-\ketbra{\psi}{\psi}\right) \right\|_{\infty} \\
	& \leq D^2 \left( \left\|\ketbra{\psi}{\psi} \right\|_{\infty} + \left\| \mathcal{R}(W) \right\|_{\infty} \right) \left\| \mathcal{R}(W)-\ketbra{\psi}{\psi} \right\|_{\infty} \\
	& = D \left( 1 + \left\| \mathcal{R}(W) \right\|_{\infty} \right) \left\| \mathcal{R}(DW-\Id) \right\|_{\infty} \\
	& \leq D^2 \left( 1 + D \left\| W \right\|_{\infty} \right) \left\| DW-\Id \right\|_{\infty} \, , 
	\end{align*}
	where the first and last inequalities are by Fact \ref{fact:realign}. Now, we know by Proposition \ref{prop:norm-W} that
	\[ \P\left( \left\| DW-\Id \right\|_{\infty} > \frac{C}{D^{\tau -1}} \right) \leq e^{-cD^2} \, , \]
	and therefore also that
	\[ \P\left( \left\| W \right\|_{\infty} > \frac{2}{D} \right) \leq e^{-cD^2} \, . \]
	Hence putting everything together, we eventually get
	\[ \P\left( \|DM-\Id\|_{\infty} > \frac{3C}{D^{\tau-3}} \right) \leq 2e^{-cD^2} \, , \]
	which (suitably re-labelling $c,C$) is exactly the announced result.
\end{proof}

From now on we set $\widetilde{P}:=DP$.

\begin{proposition} \label{prop:invariant}
	Let $d\geq D^{2\tau}$, for some $\tau>3$. Let $V_{\chi}\subset{\C^d\otimes\C^d}$ be defined as in equation \eqref{eq:W-chi}. Then, there exist universal constants $c,C>0$ such that
	\[ \P\left( \forall\ \ket{\varphi}\in V_{\chi},\ \left\| \widetilde{P}\ket{\varphi} - \ket{\varphi} \right\| \leq \frac{C}{D^{\tau-3}} \|\ket{\varphi}\| \right) \geq 1-e^{-cD^2} \, . \]	
\end{proposition}

\begin{proof}
	Define $\Sigma_{\chi}:=\left\{ \ket{\upsilon}\in\C^D\otimes\C^D,\ \left\| \ket{\chi_{\upsilon}} \right\|=1 \right\}$, and let $\ket{\upsilon}\in\Sigma_{\chi}$. Note that $\ket{\chi_{\upsilon}}=Q\ket{\upsilon}$ and $P\ket{\chi_{\upsilon}}=QM\ket{\upsilon}$. Hence,
	\[ \left\| \widetilde{P}\ket{\chi_{\upsilon}} - \ket{\chi_{\upsilon}} \right\| = \left\| DQM\ket{\upsilon} - Q\ket{\upsilon} \right\| = \left\| Q(DM-\Id)\ket{\upsilon} \right\| \leq \|Q\|_{\infty}\|DM-\Id\|_{\infty} \|\ket{\upsilon}\| \, . \]
	What is more, observe that
	\[ 1 = \braket{\chi_{\upsilon}}{\chi_{\upsilon}} = \bra{\upsilon} M \ket{\upsilon} = \frac{1}{D}\left( \bra{\upsilon} \Id \ket{\upsilon} + \bra{\upsilon} DM-\Id \ket{\upsilon} \right) \geq \frac{1}{D}(1-\|DM-\Id\|_{\infty}) \braket{\upsilon}{\upsilon} \, . \]
	We thus actually have that, for all $\ket{\upsilon}\in\Sigma_{\chi}$,
	\[ \left\| \widetilde{P}\ket{\chi_{\upsilon}} - \ket{\chi_{\upsilon}} \right\| \leq \|Q\|_{\infty}\|DM-\Id\|_{\infty} \left(\frac{D}{1-\|DM-\Id\|_{\infty}}\right)^{1/2} \, . \]
	Now, we know by Proposition \ref{prop:norm-M} that
	\[ \P\left( \|DM-\Id\|_{\infty} > \frac{C}{D^{\tau-3}} \right) \leq e^{-cD^2} \, , \] 
	which implies, recalling that $M=Q^*Q$ and hence $\|Q\|_{\infty}=\|M\|_{\infty}^{1/2}$, that
	\[ \P\left( \|Q\|_{\infty} > \frac{2}{D^{1/2}} \right) \leq e^{-cD^2} \, . \]
	Therefore putting everything together,
	\[ \P\left( \exists\ \ket{\upsilon}\in\Sigma_{\chi}: \left\| \widetilde{P}\ket{\chi_{\upsilon}} - \ket{\chi_{\upsilon}} \right\|  > \frac{2C}{(1-C/D^{\tau-3})D^{\tau-3}} \right) \leq 2e^{-cD^2} \, ,	 \]
	which (suitably re-labelling $c,C$) is exactly the announced result.
\end{proof}

\begin{proposition} \label{prop:invariant'}
	Let $d\geq D^{2\tau}$, for some $\tau>5$. Let $V_{\chi}\subset{\C^d\otimes\C^d}$ be defined as in equation \eqref{eq:W-chi}. Then, there exist universal constants $c,C>0$ such that
	\[ \P\left( \forall\ \ket{\varphi}\in V_{\chi}^{\perp},\ \left\| \widetilde{P}\ket{\varphi}\right\| \leq \frac{C}{D^{\tau-5}} \|\ket{\varphi}\| \right) \geq 1-e^{-cD^2} \, . \]	
\end{proposition}

\begin{proof}
	%First, observe that $\widetilde{P}$ can be written as
	%\[ \widetilde{P}=\frac{1}{d^2D}\sum_{a,b=1}^D (G_aG_b^*)\otimes(G_aG_b^*) \, , \]
	%where the $G_a$'s are independent $D\times d$ matrices whose entries are independent complex Gaussians with mean $0$ and variance $1$. Hence,
	%\[ \Tr(\widetilde{P}) = \frac{1}{d^2D}\sum_{a,b=1}^D (\Tr(G_aG_b^*))^2 \, . \]	
	%We thus have by Corollary \ref{cor:trace-gaussian}, applied with $\epsilon=1/D^{\tau-3}$, that
	By construction, $\widetilde P$ has rank at most $D^2$. And recalling that $\widetilde P=DM^*$, so that $\|\widetilde P\|_\infty=\|DM\|_\infty$, we know by Proposition \ref{prop:norm-M} that
	\[ \P\left( \left\|\widetilde P\right\|_{\infty} > 1+\frac{C}{D^{\tau-3}} \right) \leq e^{-cD^2} \, . \]
	We thus have
	\[ \P\left( \Tr\left(\widetilde{P}\right) > D^2\left(1+\frac{C}{D^{\tau-3}}\right) \right) \leq e^{-cD^{\tau+3}} \, . \]
	Now, we also know by Proposition \ref{prop:invariant} (recalling that $\mathrm{dim}(V_{\chi})=D^2$ with probability $1$) that
	\[ \P\left( \Tr\left(\widetilde{P}_{|V_{\chi}}\right) < D^2\left(1-\frac{C'}{D^{\tau-3}}\right) \right) \leq e^{-c'D^2} \, . \]
	As a consequence of the two above inequalities, we get that
	\[ \P\left( \Tr\left(\widetilde{P}_{|V_{\chi}^{\perp}}\right) > \frac{C+C'}{D^{\tau-5}} \right) \leq e^{-cD^{\tau+3}}+e^{-c'D^2} \, . \]
	And this clearly implies that
	\[ \P\left( \left\|\widetilde{P}_{|V_{\chi}^{\perp}}\right\|_{\infty} > \frac{C+C'}{D^{\tau-5}} \right) \leq e^{-cD^{\tau+3}}+e^{-c'D^2} \, , \]
	which (suitably re-labelling $c,C$) is exactly the announced result.
\end{proof}

Putting Propositions \ref{prop:invariant} and \ref{prop:invariant'} together we immediately get that $\widetilde{P}$ is with high probability close to $\Pi$, the projector on $V_{\chi}$. More precisely, we have Corollary \ref{cor:P-Pi} below.

\begin{corollary} \label{cor:P-Pi}
	Let $d\geq D^{2\tau}$, for some $\tau>5$. Let $V_{\chi}\subset{\C^d\otimes\C^d}$ be defined as in equation \eqref{eq:W-chi} and $\Pi$ be the projector on $V_{\chi}$. Then, there exist universal constants $c,C>0$ such that
	\[ \P\left( \left\| \widetilde{P}-\Pi\right\|_{\infty} \leq \frac{C}{D^{\tau-5}} \right) \geq 1-e^{-cD^2} \, . \]	
\end{corollary}

Let us now make a comment that will be important later on. If we define the operators $Q,P,M$ as we did before, but on $3$ sites instead of $2$, all the results that we have just established remain essentially the same. Concretely, let us slightly abusively denote again by $Q,P,M$ the $3$-site operators
\begin{center}
	\begin{tikzpicture} [scale=0.6]
		\begin{scope}[decoration={markings,mark=at position 0.5 with {\arrow{>}}}] 
			\draw[postaction={decorate}, color=brown] (5,0) -- (5,1); \draw[postaction={decorate}, color=brown] (6,0) -- (6,1);
			\draw[postaction={decorate}, color=brown] (7,0) -- (7,1); \draw[color=gray] (5,0) -- (7,0); \draw[postaction={decorate}, color=gray] (4,-0.5) to[in=-180,out=90] (5,0); \draw[postaction={decorate}, color=gray] (8,-0.5) to[in=-0,out=90] (7,0); 
			\draw (6,-1) node {$Q:(\C^D)^{\otimes 2}\longrightarrow(\C^d)^{\otimes 3}$};
			
			\draw[postaction={decorate}, color=brown] (15,0) -- (15,1); \draw[postaction={decorate}, color=brown] (16,0) -- (16,1); \draw[postaction={decorate}, color=brown] (17,0) -- (17,1); \draw[postaction={decorate}, color=brown] (15,-2) -- (15,-1); \draw[postaction={decorate}, color=brown] (16,-2) -- (16,-1); \draw[postaction={decorate}, color=brown] (17,-2) -- (17,-1); 
			\draw[color=gray] (15,0) -- (17,0); \draw[color=gray] (15,-1) -- (17,-1); \draw[color=gray] (14.2,-0.5) to[in=-180,out=90] (15,0); \draw[color=gray] (14.2,-0.5) to[in=180,out=-90] (15,-1); \draw[color=gray] (17.8,-0.5) to[in=-0,out=90] (17,0); \draw[color=gray] (17.8,-0.5) to[in=0,out=-90] (17,-1);  
			\draw (16,-2.5) node {$P=QQ^*:(\C^d)^{\otimes 3}\longrightarrow(\C^d)^{\otimes 3}$};
			
			\draw[color=brown] (25,-1) -- (25,0); \draw[color=brown] (26,-1) -- (26,0); \draw[color=brown] (27,-1) -- (27,0); 
			\draw[color=gray] (25,-1) -- (27,-1); \draw[postaction={decorate}, color=gray] (24,-1.5) to[in=-180,out=90] (25,-1); \draw[postaction={decorate}, color=gray] (28,-1.5) to[in=-0,out=90] (27,-1); \draw[color=gray] (25,0) -- (27,0); \draw[postaction={decorate}, color=gray] (25,0) to[in=-90,out=180] (24,0.5); \draw[postaction={decorate}, color=gray] (27,0) to[in=-90,out=0] (28,0.5);
			\draw (26,-2) node {$M=Q^*Q:(\C^D)^{\otimes 2}\longrightarrow(\C^D)^{\otimes 2}$};
		\end{scope}
	\end{tikzpicture}
\end{center}
and we set $\widetilde P:=DP$.

We first have the following analogue of Proposition \ref{prop:norm-M}.

\begin{proposition} \label{prop:norm-M'}
	Let $d\geq D^{2\tau}$, for some $\tau>3$. Then, there exist universal constants $c,C>0$ such that
	\[ \P\left( \|DM-\Id\|_{\infty} \leq \frac{C}{D^{\tau-3}} \right) \geq 1-e^{-cD^2} \, . \]
\end{proposition}

\begin{proof}
	Observe that $M=\mathcal{R}\left(\mathcal{R}(W)^3\right)$, while $\Id=D\mathcal{R}(\ketbra{\psi}{\psi})$, i.e.~equivalently $\ketbra{\psi}{\psi}=\mathcal{R}(\Id)/D$. Thus,
	\begin{align*}
		\|DM-\Id\|_{\infty} & = D \left\| \mathcal{R}\left(\mathcal{R}(W)^3-\ketbra{\psi}{\psi}\right) \right\|_{\infty} \\
		& \leq D^2 \left\|\mathcal{R}(W)^3-\ketbra{\psi}{\psi}\right\|_{\infty} \\
		& = D^2 \left\| \left(\mathcal{R}(W)-\ketbra{\psi}{\psi}\right) \ketbra{\psi}{\psi} + \mathcal{R}(W) \left(\mathcal{R}(W)-\ketbra{\psi}{\psi}\right)\ketbra{\psi}{\psi} + \mathcal{R}(W)^2 \left(\mathcal{R}(W)-\ketbra{\psi}{\psi}\right) \right\|_{\infty} \\
		& \leq D^2 \left( \left\|\ketbra{\psi}{\psi}\right\|_{\infty} + \left\|\mathcal{R}(W)\right\|_{\infty} \left\|\ketbra{\psi}{\psi}\right\|_{\infty} + \left\|\mathcal{R}(W)\right\|_{\infty}^2 \right) \left\| \mathcal{R}(W)-\ketbra{\psi}{\psi} \right\|_{\infty} \\
		& = D \left( 1 + \left\|\mathcal{R}(W)\right\|_{\infty} + \left\|\mathcal{R}(W)\right\|_{\infty}^2 \right) \left\|\mathcal{R}(DW-\Id)\right\|_{\infty} \\
		& \leq D^2 \left( 1 + D\left\|W\right\|_{\infty} + D^2\left\|W\right\|_{\infty}^2 \right) \left\|DW-\Id\right\|_{\infty} \, , 
	\end{align*}
	where the first and last inequalities are by Fact \ref{fact:realign}. Now, we know by Proposition \ref{prop:norm-W} that
	\[ \P\left( \left\| DW-\Id \right\|_{\infty} > \frac{C}{D^{\tau -1}} \right) \leq e^{-cD^2} \, , \]
	and therefore also that
	\[ \P\left( \left\| W \right\|_{\infty} > \frac{2}{D} \right) \leq e^{-cD^2} \, . \]
	Hence putting everything together, we eventually get
	\[ \P\left( \|DM-\Id\|_{\infty} > \frac{7C}{D^{\tau-3}} \right) \leq 2e^{-cD^2} \, , \]
	which (suitably re-labelling $c,C$) is exactly the announced result.
\end{proof}

We then also have the following analogue of Corollary \ref{cor:P-Pi}, where we again abuse notations by defining $V_\chi$ on $3$ sites as it is defined on $2$ sites in equation \eqref{eq:W-chi}, i.e.
\begin{equation} \label{eq:W-chi'} 
	V_{\chi} := \mathrm{Span} \left\{ \ket{\chi_{\upsilon}} : \ket{\upsilon}\in\C^D\otimes\C^D \right\} \, ,
\end{equation}
where $\ket{\chi_{\upsilon}}\in(\C^d)^{\otimes 3}$ is the $3$-site MPS having $\ket{\chi}$ as $1$-site tensor and $\ket{\upsilon}$ as boundary condition. By construction we still have $\mathrm{dim}(V_{\chi})=D^2$ with probability $1$.

\begin{corollary} \label{cor:P-Pi'}
	Let $d\geq D^{2\tau}$, for some $\tau>5$. Let $V_{\chi}\subset(\C^d)^{\otimes 3}$ be defined as in equation \eqref{eq:W-chi'} and $\Pi$ be the projector on $V_{\chi}$. Then, there exist universal constants $c,C>0$ such that
	\[ \P\left( \left\| \widetilde{P}-\Pi\right\|_{\infty} \leq \frac{C}{D^{\tau-5}} \right) \geq 1-e^{-cD^2} \, . \]	
\end{corollary}

\begin{proof}
	Corollary \ref{cor:P-Pi'} is a consequence of the two following intermediate results, which are the analogues of Propositions \ref{prop:invariant} and \ref{prop:invariant'}, respectively:
	\begin{equation} \label{eq:invariant} \P\left( \forall\ \ket{\varphi}\in V_{\chi},\ \left\| \widetilde{P}\ket{\varphi} - \ket{\varphi} \right\| \leq \frac{C}{D^{\tau-3}} \|\ket{\varphi}\| \right) \geq 1-e^{-cD^2} \, , \end{equation}
	which is valid as soon as $d\geq D^{2\tau}$ for some $\tau>3$, and
	\begin{equation} \label{eq:invariant'} \P\left( \forall\ \ket{\varphi}\in V_{\chi}^{\perp},\ \left\| \widetilde{P}\ket{\varphi}\right\| \leq \frac{C}{D^{\tau-5}} \|\ket{\varphi}\| \right) \geq 1-e^{-cD^2} \, , \end{equation}	 
	which is valid as soon as $d\geq D^{2\tau}$ for some $\tau>5$.
	
	Indeed, following the exact same steps as in the proof of Proposition \ref{prop:invariant} and using Proposition \ref{prop:norm-M'} in place of Proposition \ref{prop:norm-M} provides equation \eqref{eq:invariant}. While following the exact same steps as in the proof of Proposition \ref{prop:invariant'} and using Proposition \ref{prop:norm-M'} and equation \eqref{eq:invariant} in place of Proposition \ref{prop:norm-M} and Proposition \ref{prop:invariant}, respectively, provides equation \eqref{eq:invariant'}.
\end{proof}

We now need one final result, comparing the product of $2$-site operators $\widetilde{P}$ overlapping on $1$ site to the $3$-site operator $\widetilde{P}$.

\begin{comment}
In order to state Proposition \ref{prop:P-commute}, we introduce $\widetilde{P}_{123}:=DP_{123}$, where $P_{123}$ is defined as
\begin{center}
	\begin{tikzpicture} [scale=0.6]
		\begin{scope}[decoration={markings,mark=at position 0.5 with {\arrow{>}}}] 
			\draw[postaction={decorate}, color=brown] (15,0) -- (15,1); \draw[postaction={decorate}, color=brown] (16,0) -- (16,1); \draw[postaction={decorate}, color=brown] (17,0) -- (17,1); \draw[postaction={decorate}, color=brown] (15,-2) -- (15,-1); \draw[postaction={decorate}, color=brown] (16,-2) -- (16,-1); \draw[postaction={decorate}, color=brown] (17,-2) -- (17,-1); \draw[color=gray] (15,0) -- (16,0); \draw[color=gray] (15,-1) -- (16,-1); \draw[color=gray] (16,0) -- (17,0); \draw[color=gray] (16,-1) -- (17,-1); \draw[color=gray] (14.2,-0.5) to[in=-180,out=90] (15,0); \draw[color=gray] (14.2,-0.5) to[in=180,out=-90] (15,-1); \draw[color=gray] (17.8,-0.5) to[in=-0,out=90] (17,0); \draw[color=gray] (17.8,-0.5) to[in=0,out=-90] (17,-1);  
			\draw (16,-2.5) node {$P_{123}:\C^d\otimes\C^d\otimes\C^d\longrightarrow\C^d\otimes\C^d\otimes\C^d$};
		\end{scope}
	\end{tikzpicture}
\end{center}
\end{comment}

\begin{proposition} \label{prop:P-commute}
	Let $d\geq D^{2\tau}$, for some $\tau>3$. Then, there exist universal constants $c,C>0$ such that
	\[ \P\left( \left\| \left(\widetilde{P}_{12}\otimes\Id_3\right) \left(\Id_1\otimes\widetilde{P}_{23}\right) - \widetilde{P}_{123} \right\|_{\infty} \leq \frac{C}{D^{\tau-1}} \right) \geq 1-e^{-cD^2} \, . \]	
\end{proposition}

%In order to prove Proposition \ref{prop:P-commute}, i.e.~that $(\widetilde{P}_{12}\otimes\Id_3) (\Id_1\otimes\widetilde{P}_{23})$ and $(\Id_1\otimes\widetilde{P}_{23}) (\widetilde{P}_{12}\otimes\Id_3)$ are close to one another, we will actually show that both of them are close to $\widetilde{P}_{123}:=DP_{123}$, where $P_{123}$ is defined as
Before proving Proposition \ref{prop:P-commute}, there are still two operators $N,N'$ that we have to introduce:
\begin{center}
	\begin{tikzpicture} [scale=0.6]
	\begin{scope}[decoration={markings,mark=at position 0.5 with {\arrow{>}}}] 
	\draw[postaction={decorate}, color=brown] (15,0) -- (15,1); \draw[postaction={decorate}, color=brown] (16,0) -- (16,1); \draw[postaction={decorate}, color=brown] (15,-2) -- (15,-1); \draw[color=gray] (15,0) -- (16,0); \draw[postaction={decorate},  color=gray] (16,-1) -- (15,-1); \draw[color=gray] (14.2,-0.5) to[in=-180,out=90] (15,0); \draw[color=gray] (14.2,-0.5) to[in=180,out=-90] (15,-1); \draw[postaction={decorate}, color=gray] (16.8,-0.5) to[in=-0,out=90] (16,0); 
	\draw (15.5,-2.5) node {$N:\C^d\otimes(\C^D)^{\otimes 2}\longrightarrow(\C^d)^{\otimes 2}$};
	
	\draw[postaction={decorate}, color=brown] (25,0) -- (25,1); \draw[postaction={decorate}, color=brown] (26,0) -- (26,1); \draw[postaction={decorate}, color=brown] (26,-2) -- (26,-1); \draw[color=gray] (25,0) -- (26,0); \draw[postaction={decorate}, color=gray] (25,-1) -- (26,-1); \draw[postaction={decorate}, color=gray] (24.2,-0.5) to[in=-180,out=90] (25,0); \draw[color=gray] (26.8,-0.5) to[in=0,out=-90] (26,-1); \draw[color=gray] (26.8,-0.5) to[in=-0,out=90] (26,0); 
	\draw (25.5,-2.5) node {$N':(\C^D)^{\otimes 2}\otimes\C^d\longrightarrow(\C^d)^{\otimes 2}$};
	\end{scope}
	\end{tikzpicture}
\end{center}

\begin{proof}
	Note that $P_{123}=(N_{12}\otimes\Id_3)(\Id_1\otimes N'^*_{23})$ and $(P_{12}\otimes\Id_3)(\Id_1\otimes P_{23})= (N_{12}\otimes\Id_3)(\Id_1\otimes W_2\otimes\Id_3)(\Id_1\otimes N'^*_{23})$. Hence,
	\begin{align*}
	\left\| \left(\widetilde{P}_{12}\otimes\Id_3\right)\left(\Id_1\otimes \widetilde{P}_{23}\right) - \widetilde{P}_{123} \right\|_{\infty} & = D\left\| (N_{12}\otimes\Id_3)(\Id_1\otimes (DW_2-\Id_2)\otimes\Id_3)(\Id_1\otimes N'^*_{23}) \right\|_{\infty} \\
	& \leq D \|N_{12}\otimes\Id_3\|_{\infty} \|\Id_1\otimes N'^*_{23}\|_{\infty} \|\Id_1\otimes (DW_2-\Id_2)\otimes\Id_3\|_{\infty} \\
	& = D \|N\|_{\infty} \|N'^*\|_{\infty} \|DW-\Id\|_{\infty} \\
	& = D \|NN^*\|_{\infty} \|DW-\Id\|_{\infty} \, .
	\end{align*}
	Next, observe that $NN^*=Q_{12}(\hat{W}_1\otimes\Id_2)Q^*_{12}$, where $\hat{W}_1:=\Tr_2(W_{12})$. Consequently,
	\[ \|NN^*\|_{\infty} \leq \|Q\|_{\infty} \|Q^*\|_{\infty} \|\hat{W}\|_{\infty} \leq D \|Q\|_{\infty}^2 \|W\|_{\infty} \, . \]
	Now, we have already argued that, as a consequence of Proposition \ref{prop:norm-M}, we have
	\[ \P\left( \|Q\|_{\infty} > \frac{2}{D^{1/2}} \right) \leq e^{-cD^2} \, , \]
	while as an immediate consequence of Proposition \ref{prop:norm-W}, we have
	%\[ \P\left( \|DW-\Id\|_{\infty} > \frac{C}{D^{\tau-1}} \right) \leq e^{-c'D^2} \, , \]
	%and therefore also that 
	\[ \P\left( \|W\|_{\infty} > \frac{2}{D} \right) \leq e^{-c'D^2} \, . \]
	Hence putting everything together, we eventually get
	\[ \P\left( \left\| \left(\widetilde{P}_{12}\otimes\Id_3\right)\left(\Id_1\otimes \widetilde{P}_{23}\right) - \widetilde{P}_{123} \right\|_{\infty} > \frac{8C}{D^{\tau-1}} \right) \leq e^{-cD^2} + 2e^{-c'D^2} \, , \]
	%Similarly, we can show that
	%\[ \P\left( \left\| \left(\Id_1\otimes \widetilde{P}_{23}\right)\left(\widetilde{P}_{12}\otimes\Id_3\right) - \widetilde{P}_{123} \right\|_{\infty} > \frac{8C}{D^{\tau-1}} \right) \leq e^{-cD^2} + 2e^{-c'D^2} \, . \]
	%As a consequence, we have
	%\[ \P\left( \left\| \left(\widetilde{P}_{12}\otimes\Id_3\right)\left(\Id_1\otimes \widetilde{P}_{23}\right)- \left(\Id_1\otimes \widetilde{P}_{23}\right)\left(\widetilde{P}_{12}\otimes\Id_3\right) \right\|_{\infty} > \frac{16C}{D^{\tau-1}} \right) \leq 2\left(e^{-cD^2} + 2e^{-c'D^2}\right) \, , \]
	which (suitably re-labelling $c,C$) implies precisely the announced result.
\end{proof}

\subsubsection{Conclusions} \hfill\par\smallskip

Combining Corollary \ref{cor:P-Pi} and Proposition \ref{prop:P-commute}, we immediately get Theorem \ref{th:Pi-commute} below.

\begin{theorem} \label{th:Pi-commute}
	Let $d\geq D^{2\tau}$, for some $\tau>5$. Then, there exist universal constants $c,C>0$ such that
	\[ \P\left( \left\| \left(\Pi_{12}\otimes\Id_3\right) \left(\Id_1\otimes\Pi_{23}\right)-\Pi_{123} \right\|_{\infty} \leq \frac{C}{D^{\tau-5}} \right) \geq 1-e^{-cD^2} \, . \]	
\end{theorem}

\begin{proof}
	Setting $R_{12}:=\Pi_{12}-\widetilde{P}_{12}$, $R_{23}:=\Pi_{23}-\widetilde{P}_{23}$ and $S_{123}:=\Pi_{123}-\widetilde{P}_{123}$ we have
	\[  \left(\Pi_{12}\otimes\Id_3\right) \left(\Id_1\otimes\Pi_{23}\right)-\Pi_{123} = \left(\widetilde{P}_{12}\otimes\Id_3\right) \left(\Id_1\otimes\widetilde{P}_{23}\right)-\widetilde{P}_{123} + T_{123} \, , \] %= \left((\widetilde{P}+R)_{12}\otimes\Id_3\right) \left(\Id_1\otimes(\widetilde{P}+R)_{23}\right)-(\widetilde{P}+S)_{123} 
	with $\|T\|_{\infty} \leq 2\|\widetilde{P}\|_{\infty}\|R\|_{\infty} + \|R\|_{\infty}^2 + \|S\|_\infty$ and $\|\widetilde{P}\|_{\infty}\leq\|\Pi\|_\infty+\|R\|_\infty=1+\|R\|_\infty$. Now on the one hand, we know by Proposition \ref{prop:P-commute} that
	\[  \P\left( \left\| \left(\widetilde{P}_{12}\otimes\Id_3\right) \left(\Id_1\otimes\widetilde{P}_{23}\right)-\widetilde{P}_{123} \right\|_{\infty} > \frac{C}{D^{\tau-1}} \right) \leq e^{-cD^2} \, . \]
	While on the other hand, we know by Corollary \ref{cor:P-Pi} that
	\[ \P\left( \|R\|_{\infty}> \frac{C'}{D^{\tau-5}} \right) \leq e^{-c'D^2} \, , \]
	and by Corollary \ref{cor:P-Pi'} that 
	\[ \P\left( \|S\|_{\infty}> \frac{C''}{D^{\tau-5}} \right) \leq e^{-c''D^2} \, . \]
	We thus get that
	\[ \P\left( \left\| \left(\Pi_{12}\otimes\Id_3\right) \left(\Id_1\otimes\Pi_{23}\right)-\Pi_{123} \right\|_{\infty} > \frac{C}{D^{\tau-1}} + \frac{5C'+C''}{D^{\tau-5}} \right) \leq e^{-cD^2} + e^{-c'D^2} + e^{-c''D^2} \, , \]
	which (suitably re-labelling $c,C$) is exactly the announced result.
\end{proof}

We are now in position to prove that the parent Hamiltonian $H_{\chi}$ of $\ket{\chi^N}$, as defined by equation \eqref{eq:H-parent}, is typically gapped. 

\begin{theorem} \label{th:gap-H}
	Let $d\geq D^{2\tau}$, for some $\tau>5$. Let $H_{\chi}$ be the parent Hamiltonian of $\ket{\chi^N}$, as defined by equation \eqref{eq:H-parent}. Then, there exist universal constants $c,C>0$ such that
	\[ \P\left( \Delta(H_{\chi}) \geq 1-\frac{C}{D^{\tau-5}} \right) \geq 1-e^{-cD^2} \, . \]
\end{theorem}

\begin{proof}
	In order to prove Theorem \ref{th:gap-H} we will use a now standard strategy, first introduced in \cite[Theorem 6.4]{FNW}. Namely, we will actually show that, with probability larger than $1-e^{-cD^2}$, $H_{\chi}^2\geq (1-C/D^{\tau-5})H_{\chi}$. For the sake of simplifying notation we set $\hat{\Pi}_i:= \Pi_{i,i+1}^{\perp}\otimes\Id_{1,\ldots,i-1,i+2,\ldots,N}$, $1\leq i\leq N$, so that
	\[ H_{\chi} = \sum_{i=1}^N \hat{\Pi}_i \, . \] 	
	And thus,
	\[ H_{\chi}^2 = \sum_{i,j=1}^N \hat{\Pi}_i\hat{\Pi}_j = \sum_{i=1}^N \hat{\Pi}_i^2 + \sum_{i=1}^{N-1} \left(\hat{\Pi}_i\hat{\Pi}_{i+1} + \hat{\Pi}_{i+1}\hat{\Pi}_i \right) + \sum_{\substack{i,j=1 \\ |i-j|>1}}^N \hat{\Pi}_i\hat{\Pi}_j \, . \]
	Now first, since the $\hat{\Pi}_i$'s are projectors,
	\[ \sum_{i=1}^N \hat{\Pi}_i^2 = \sum_{i=1}^N \hat{\Pi}_i = H_{\chi} \, . \]
	Then, since $\hat{\Pi}_i$ and $\hat{\Pi}_j$ commute for $|i-j|>1$,
	\[ \sum_{\substack{i,j=1 \\ |i-j|>1}}^N \hat{\Pi}_i\hat{\Pi}_j \geq 0 \, . \]
	And finally, we know by \cite[Lemma 3]{FNW} that
	\[ \hat{\Pi}_i\hat{\Pi}_{i+1} + \hat{\Pi}_{i+1}\hat{\Pi}_i \geq - \left\|\hat{\Pi}_i\hat{\Pi}_{i+1}-\hat{\Pi}_i\wedge\hat{\Pi}_{i+1}\right\|_\infty \left(\hat{\Pi}_i+\hat{\Pi}_{i+1}\right) \, , \]
	where $\hat{\Pi}_i\wedge\hat{\Pi}_{i+1}$ stands for the largest lower bound on $\hat{\Pi}_i$ and $\hat{\Pi}_{i+1}$ in the lattice of projections. Now, by injectivity of the MPS we have that $\hat{\Pi}_i\wedge\hat{\Pi}_{i+1}=\Pi_{i,i+1,i+2}^\perp\otimes\Id_{1,\ldots,i-1,i+3,\ldots,N}$. Hence by Theorem \ref{th:Pi-commute}, with probability larger than $1-e^{-cD^2}$,
	\[ \left\|\hat{\Pi}_i\hat{\Pi}_{i+1}-\hat{\Pi}_i\wedge\hat{\Pi}_{i+1}\right\|_\infty = \left\|\left(\Pi_{i,i+1}\otimes\Id_{i+2}\right)\left(\Id_i\otimes\Pi_{i+1,i+2}\right)-\Pi_{i,i+1,i+2}\right\|_\infty \leq \frac{C}{D^{\tau-5}} \, , \]
	and therefore $\hat{\Pi}_i\hat{\Pi}_{i+1}+\hat{\Pi}_{i+1}\hat{\Pi}_i\geq -C/D^{\tau-5}(\hat{\Pi}_i+\hat{\Pi}_{i+1})$. So, with probability larger than $1-e^{-cD^2}$,
	\[ \sum_{i=1}^{N-1} \left(\hat{\Pi}_i\hat{\Pi}_{i+1} + \hat{\Pi}_{i+1}\hat{\Pi}_i \right) \geq -\frac{2C}{D^{\tau-5}} \sum_{i=1}^{N}\hat{\Pi}_i = -\frac{2C}{D^{\tau-5}}H_{\chi} \, . \]
	Hence putting everything together, we eventually get that, with probability larger than $1-e^{-cD^2}$,
	\[ H_{\chi}^2 \geq \left(1-\frac{2C}{D^{\tau-5}}\right)H_{\chi} \, , \]
	which (suitably re-labelling $
	C$) is indeed what we wanted to show.
\end{proof}

\subsection{The case of PEPS} \hfill\par\smallskip
\label{sec:parent-hamiltonian-PEPS}

To study the case of PEPS we will follow step by step the strategy adopted in the case of MPS. We may thus skip several details.

We assume here that $d>D^4$, so that our random PEPS $\ket{\chi^N}$ is injective with probability $1$ (see Fact \ref{fact:irreducible}). We know from \cite{CPGVW} again that, in this case, there exists a canonical way of constructing a $2$-local frustration-free Hamiltonian on $(\C^d)^{\otimes N}$ whose unique ground-state is $\ket{\chi^N}$, with ground energy $0$. As before, we call it the parent Hamiltonian of $\ket{\chi^N}$, denote it by $H_{\chi}$, and want to show that it typically has a large (lower) spectral gap $\Delta(H_{\chi})$. 

This PEPS parent Hamiltonian is constructed similarly to the MPS parent Hamiltonian: Define $V_{\chi}\subset \C^d\otimes\C^d$ as  
\begin{equation} \label{eq:W-chi-PEPS} 
V_{\chi} := \mathrm{Span} \left\{ \ket{\chi_{\upsilon}} : \ket{\upsilon}\in(\C^D)^{\otimes 3}\otimes(\C^D)^{\otimes 3} \right\} \, ,
\end{equation}
where $\ket{\chi_{\upsilon}}\in\C^d\otimes\C^d$ is the $2$-site PEPS having $\ket{\chi}$ as $1$-site tensor and $\ket{\upsilon}$ as boundary condition. By construction we always have $\mathrm{dim}(V_{\chi})\leq D^6<d^2$, so that $\mathrm{dim}(V_{\chi}^{\perp})\geq d^2-D^6>0$. And in our case we actually have $\mathrm{dim}(V_{\chi})=D^6$ with probability $1$. Then, denote by $\Pi$ the projector on $V_{\chi}$. 
To streamline notation we use the following shorthand for `vertical' and `horizontal' $2$-site terms: for any $1\leq i,j \leq N$,
\[ \hat{\Pi}_{(i,j),(i+1,j)}^{v} := \Pi^{\perp}_{(i,j),(i+1,j)} \otimes\Id_{\{(k,l) \, : \, k\neq i,i+1, \, l\neq j\}} \ \ \text{and} \ \ \hat{\Pi}_{(i,j),(i,j+1)}^{h} := \Pi^{\perp}_{(i,j),(i,j+1)} \otimes\Id_{\{(k,l) \, : \, k\neq i, \, l\neq j,j+1\}} \, . \]
The parent Hamiltonian $H_{\chi}$ of $\ket{\chi^N}$ is then defined as
\begin{equation} \label{eq:H-parent-PEPS} 
H_{\chi} := \sum_{j=1}^N \left( \sum_{i=1}^N \hat{\Pi}_{(i,j),(i+1,j)}^v \right) + \sum_{i=1}^N \left( \sum_{j=1}^N \hat{\Pi}_{(i,j),(i,j+1)}^h \right)  \, . 
\end{equation}
So in conclusion, the PEPS parent Hamiltonian can actually be seen as a sum of terms which are of the form of an MPS parent Hamiltonian, just that the boundary dimensions are not $D$ but $D^3$. This means that, up to this replacement, we can use all the intermediate results proved in the MPS case. 

\subsubsection{Approximating the local ground space projectors} \hfill\par\smallskip

We now have to look at the following three operators $W,W',W''$:
\begin{center}
	\begin{tikzpicture} [scale=0.6]
	\begin{scope}[decoration={markings,mark=at position 0.5 with {\arrow{>}}}] 
	\draw[color=brown] (15,0) -- (15,1); \draw[postaction={decorate}, very thick, color=gray] (14,-0.5) to[in=-180,out=90] (15,0); \draw[postaction={decorate}, color=gray] (16,-0.5) to[in=-0,out=90] (15,0); \draw[postaction={decorate}, very thick, color=gray] (15,1) to[in=-90,out=180] (14,1.5); \draw[postaction={decorate}, color=gray] (15,1) to[in=-90,out=0] (16,1.5);
	\draw (15,-1) node {$W:\C^{D^3}\otimes\C^D\longrightarrow\C^{D^3}\otimes\C^D$};
	
	\draw[color=brown] (25,0) -- (25,1); \draw[postaction={decorate}, color=gray] (24,-0.5) to[in=-180,out=90] (25,0); \draw[postaction={decorate}, very thick, color=gray] (26,-0.5) to[in=-0,out=90] (25,0); \draw[postaction={decorate}, color=gray] (25,1) to[in=-90,out=180] (24,1.5); \draw[postaction={decorate}, very thick, color=gray] (25,1) to[in=-90,out=0] (26,1.5);
	\draw (25,-1) node {$W':\C^D\otimes\C^{D^3}\longrightarrow\C^D\otimes\C^{D^3}$};
	
	\draw[color=brown] (35,0) -- (35,1); \draw[postaction={decorate}, very thick, color=gray] (34,-0.5) to[in=-180,out=90] (35,0); \draw[postaction={decorate}, color=gray] (36,-0.5) to[in=-0,out=90] (35,0); \draw[postaction={decorate}, color=gray] (35,1) to[in=-90,out=180] (34,1.5); \draw[postaction={decorate}, very thick, color=gray] (35,1) to[in=-90,out=0] (36,1.5);
	\draw (35,-1) node {$W'':\C^{D^3}\otimes\C^D\longrightarrow\C^D\otimes\C^{D^3}$};
	\end{scope}
	\end{tikzpicture}
\end{center}

While the operators $Q,P,M$ that we now have to consider are:
\begin{center}
	\begin{tikzpicture} [scale=0.6]
	\begin{scope}[decoration={markings,mark=at position 0.5 with {\arrow{>}}}] 
	\draw[postaction={decorate}, color=brown] (5,0) -- (5,1); \draw[postaction={decorate}, color=brown] (6,0) -- (6,1); \draw[color=gray] (5,0) -- (6,0); \draw[postaction={decorate}, very thick, color=gray] (4,-0.5) to[in=-180,out=90] (5,0); \draw[postaction={decorate}, very thick, color=gray] (7,-0.5) to[in=-0,out=90] (6,0); 
	\draw (5.5,-1) node {$Q:(\C^{D^3})^{\otimes 2}\longrightarrow(\C^d)^{\otimes 2}$};
	
	\draw[postaction={decorate}, color=brown] (15,0) -- (15,1); \draw[postaction={decorate}, color=brown] (16,0) -- (16,1); \draw[postaction={decorate}, color=brown] (15,-2) -- (15,-1); \draw[postaction={decorate}, color=brown] (16,-2) -- (16,-1); \draw[color=gray] (15,0) -- (16,0); \draw[color=gray] (15,-1) -- (16,-1); \draw[very thick, color=gray] (14.2,-0.5) to[in=-180,out=90] (15,0); \draw[very thick, color=gray] (14.2,-0.5) to[in=180,out=-90] (15,-1); \draw[very thick, color=gray] (16.8,-0.5) to[in=-0,out=90] (16,0); \draw[very thick, color=gray] (16.8,-0.5) to[in=0,out=-90] (16,-1);  
	\draw (15.5,-2.5) node {$P=QQ^*:(\C^d)^{\otimes 2}\longrightarrow(\C^d)^{\otimes 2}$};
	
	\draw[color=brown] (25,-1) -- (25,0); \draw[color=brown] (26,-1) -- (26,0); \draw[color=gray] (25,-1) -- (26,-1); \draw[postaction={decorate}, very thick, color=gray] (24,-1.5) to[in=-180,out=90] (25,-1); \draw[postaction={decorate}, very thick, color=gray] (27,-1.5) to[in=-0,out=90] (26,-1); \draw[color=gray] (25,0) -- (26,0); \draw[postaction={decorate}, very thick, color=gray] (25,0) to[in=-90,out=180] (24,0.5); \draw[postaction={decorate}, very thick, color=gray] (26,0) to[in=-90,out=0] (27,0.5);
	\draw (25.5,-2) node {$M=Q^*Q:(\C^{D^3})^{\otimes 2}\longrightarrow(\C^{D^3})^{\otimes 2}$};
	\end{scope}
	\end{tikzpicture}
\end{center}

Note that all these operators can be viewed either as `vertical' or as `horizontal' operators. The analysis that follows is exactly the same in both cases, as the distribution of an operator does not depend on its orientation. So we do not make the distinction for now. It is only later, when looking at the commutation relations between operators of the form $\widetilde{P}_{12}\otimes\Id_3$ and $\Id_1\otimes\widetilde{P}_{23}$, that it will matter whether the operators $P_{12}$ and $P_{23}$ involved have the same or different orientations.

We can first show that, in the range $d>D^4$ that we are interested in, $W,W',W''$ (suitably renormalized) are generically close to the identity. Indeed, as immediate corollary of Theorem \ref{th:Wishart} (applied with $n=D^4$ and $s=d$) we have the result below, which is the analogue of Proposition \ref{prop:norm-W}.

\begin{proposition} \label{prop:norm-W-PEPS}
	Let $d\geq D^{4\tau}$, for some $\tau>1$. Then, there exist universal constants $c,C>0$ such that
	\[ \P\left( \left\|D^2W-\Id\right\|_{\infty} \leq \frac{C}{D^{2\tau-2}} \right) \geq 1-e^{-cD^4} \, . \]
	And the same holds for $W',W''$.
\end{proposition}

We can then show that, in the range $d>D^{14}$, also $M$ (suitably renormalized) is close to the identity, i.e.~an analogue of Proposition \ref{prop:norm-M}.

\begin{proposition} \label{prop:norm-M-PEPS}
	Let $d\geq D^{4\tau}$, for some $\tau>7/2$. Then, there exist universal constants $c,C>0$ such that
	\[ \P\left( \|D^3M-\Id\|_{\infty} \leq \frac{C}{D^{2\tau-7}} \right) \geq 1-e^{-cD^4} \, . \]
\end{proposition}

\begin{proof}
	We argue as in the proof of Proposition \ref{prop:norm-M}, observing that $M=\mathcal{R}\left(\mathcal{R}(W)\mathcal{R}(W')\right)$, while $\Id=D^3\mathcal{R}(\ketbra{\psi}{\psi})$ and $\Id=D^2\mathcal{R}(\ketbra{\psi}{\psi'})=D^2\mathcal{R}(\ketbra{\psi'}{\psi})$, where $\ket{\psi}\in\C^{D^3}\otimes\C^{D^3}$ and $\ket{\psi'}\in\C^D\otimes\C^D$ are the maximally entangled unit vectors. Thus,
	\begin{align*}
	\|D^3M-\Id\|_{\infty} & = D^3 \left\| \mathcal{R}\left(\mathcal{R}(W)\mathcal{R}(W')-\ketbra{\psi}{\psi}\right) \right\|_{\infty} \\
	& \leq D^6 \left\| \mathcal{R}(W)\mathcal{R}(W')-\ketbra{\psi}{\psi} \right\|_{\infty} \\
	& = D^6 \left\| \left(\mathcal{R}(W)-\ketbra{\psi}{\psi'}\right) \ketbra{\psi'}{\psi} + \mathcal{R}(W) \left(\mathcal{R}(W')-\ketbra{\psi'}{\psi}\right) \right\|_{\infty} \\
	& \leq D^6 \left( \left\|\ketbra{\psi'}{\psi} \right\|_{\infty} \left\| \mathcal{R}(W)-\ketbra{\psi}{\psi'} \right\|_{\infty} + \left\| \mathcal{R}(W) \right\|_{\infty} \left\| \mathcal{R}(W')-\ketbra{\psi'}{\psi} \right\|_{\infty} \right) \\
	& = D^4 \left( \left\| \mathcal{R}(D^2W-\Id) \right\|_{\infty} + \left\| \mathcal{R}(W) \right\|_{\infty} \left\| \mathcal{R}(D^2W'-\Id) \right\|_{\infty} \right)  \\
	& \leq D^5 \left( \left\| D^2W-\Id \right\|_{\infty} + D \left\| W \right\|_{\infty} \left\| D^2W'-\Id \right\|_{\infty} \right)  \, , 
	\end{align*}
	where the first and last inequalities are by Fact \ref{fact:realign}. Now, we know by Proposition \ref{prop:norm-W-PEPS} that
	\[ \P\left( \left\| D^2W-\Id \right\|_{\infty} > \frac{C}{D^{2\tau -2}} \right) \leq e^{-cD^4}\ \ \text{and}\ \ \P\left( \left\| D^2W'-\Id \right\|_{\infty} > \frac{C}{D^{2\tau -2}} \right) \leq e^{-cD^4} \, , \]
	and therefore also that
	\[ \P\left( \left\| W \right\|_{\infty} > \frac{2}{D^2} \right) \leq e^{-cD^4} \, . \]
	Hence putting everything together, we eventually get
	\[ \P\left( \|D^3M-\Id\|_{\infty} > \frac{2C}{D^{2\tau-7}} \right) \leq 3e^{-cD^4} \, , \]
	which (suitably re-labelling $c,C$) is exactly the announced result.
\end{proof}

\begin{comment}
As a direct consequence of Proposition \ref{prop:norm-M-PEPS} we get the following upper bound on the operator norm of $Q$, just as Corollary \ref{cor:norm-Q} is derived from Proposition \ref{prop:norm-M}.

\begin{corollary} \label{cor:norm-Q-PEPS}
	Let $d\geq D^{4\tau}$, for some $\tau>7/2$. Then, there exists a universal constant $c>0$ such that
	\[ \P\left( \|Q\|_{\infty} \leq \frac{2}{D^{3/2}} \right) \geq 1-e^{-cD^4} \, . \]
\end{corollary}
\end{comment}

We now define $\widetilde{P}:= D^3P$. With the above preliminary results at hand, we can derive the analogues of Corollary \ref{cor:P-Pi} and Proposition \ref{prop:P-commute}, following exactly the same proof strategies. We therefore only recall the main steps in the arguments.

\begin{corollary} \label{cor:P-Pi-PEPS}
	Let $d\geq D^{4\tau}$, for some $\tau>13/2$. Let $V_{\chi}\subset{\C^d\otimes\C^d}$ be defined as in equation \eqref{eq:W-chi-PEPS} and $\Pi$ be the projector on $V_{\chi}$. Then, there exist universal constants $c,C>0$ such that
	\[ \P\left( \left\| \widetilde{P}-\Pi\right\|_{\infty} \leq \frac{C}{D^{2\tau-13}} \right) \geq 1-e^{-cD^4} \, . \]	
\end{corollary}

\begin{proof}
	Corollary \ref{cor:P-Pi-PEPS} is a consequence of the two following intermediate results:
	\begin{align*}
	& \P\left( \forall\ \ket{\varphi}\in V_{\chi},\ \left\| \widetilde{P}\ket{\varphi}-\ket{\varphi} \right\| \leq \frac{C}{D^{2\tau-7}} \|\ket{\varphi}\| \right) \geq 1-e^{-cD^4} \, , \\
	& \P\left( \forall\ \ket{\varphi}\in V_{\chi}^{\perp},\ \left\| \widetilde{P}\ket{\varphi}\right\| \leq \frac{C}{D^{2\tau-13}} \|\ket{\varphi}\| \right) \geq 1-e^{-cD^4} \, .
	\end{align*}
	
	The first equation is due to the fact that, for any $\ket{\upsilon}\in(\C^D)^{\otimes 3}\otimes(\C^D)^{\otimes 3}$ such that $\|\ket{\chi_{\upsilon}}\|=1$,
	\[ \left\| \widetilde{P}\ket{\chi_{\upsilon}}-\ket{\chi_{\upsilon}} \right\| = \| Q(D^3M-\Id) \ket{\upsilon} \| \leq \|Q\|_{\infty} \|D^3M-\Id\|_{\infty} \|\ket{\upsilon}\| \leq \|Q\|_{\infty} \|D^3M-\Id\|_{\infty} \left(\frac{D^3}{1-\|D^3M-\Id\|_{\infty}}\right)^{1/2} \, , \]	
	and the latter quantity is, with probability larger than $1-e^{-cD^4}$, smaller than $C/D^{2\tau-7}$.
	
	The second equation is obtained by combining the first equation, which tells us that 
	\[ \P\left( \Tr\left(\widetilde{P}_{|V_{\chi}}\right) \geq D^6\left(1-\frac{C}{D^{2\tau-7}}\right) \right) \geq 1-e^{-cD^4} \, ,  \]
	with the observation that
	\[ \P\left( \Tr\left(\widetilde{P}\right) \leq D^6\left(1+\frac{C}{D^{2\tau-7}}\right) \right) \geq 1-e^{-cD^{2\tau+7}} \, . \]
	Indeed, we get from these that, with probability larger than $1-e^{-cD^4}$, $\|\widetilde{P}_{|V_{\chi}^{\perp}}\|_{\infty}\leq C/D^{2\tau-13}$.
\end{proof}

\begin{proposition} \label{prop:P-commute-PEPS}
	Let $d\geq D^{4\tau}$, for some $\tau>7/2$. Then, there exist universal constants $c,C>0$ such that, for the orientations `$o_1$' and `$o_2$' being either `$v$' or `$h$',
	\[ \P\left( \left\| \left(\widetilde{P}_{12}^{o_1}\otimes\Id_3\right) \left(\Id_1\otimes\widetilde{P}_{23}^{o_2}\right) - \widetilde{P}^{o_1,o_2}_{123} \right\|_{\infty} \leq \frac{C}{D^{2\tau-2}} \right) \geq 1-e^{-cD^4} \, . \]	
\end{proposition}

The proof of Proposition \ref{prop:P-commute-PEPS} goes exactly as the one of Proposition \ref{prop:P-commute}. In the case where $o_1=o_2=:o$, it consists in showing that $(\widetilde{P}_{12}^{o}\otimes\Id_3) (\Id_1\otimes\widetilde{P}_{23}^{o})$ is close to $\widetilde{P}_{123}^{o}:=D^4P_{123}^{o}$, where $P_{123}^{o}$ is defined as
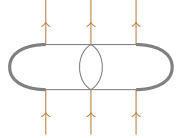
\begin{center}
	\begin{tikzpicture} [scale=0.6]
	\begin{scope}[decoration={markings,mark=at position 0.5 with {\arrow{>}}}] 
	\draw[postaction={decorate}, color=brown] (15,0) -- (15,1); \draw[postaction={decorate}, color=brown] (16,0) -- (16,1); \draw[postaction={decorate}, color=brown] (17,0) -- (17,1); \draw[postaction={decorate}, color=brown] (15,-2) -- (15,-1); \draw[postaction={decorate}, color=brown] (16,-2) -- (16,-1); \draw[postaction={decorate}, color=brown] (17,-2) -- (17,-1); \draw[color=gray] (15,0) -- (16,0); \draw[color=gray] (15,-1) -- (16,-1); \draw[color=gray] (16,0) -- (17,0); \draw[color=gray] (16,-1) -- (17,-1); \draw[very thick, color=gray] (14.2,-0.5) to[in=-180,out=90] (15,0); \draw[very thick, color=gray] (14.2,-0.5) to[in=180,out=-90] (15,-1); \draw[very thick, color=gray] (17.8,-0.5) to[in=-0,out=90] (17,0); \draw[very thick, color=gray] (17.8,-0.5) to[in=0,out=-90] (17,-1); \draw[color=gray] (16,0) to[in=150,out=-150] (16,-1); \draw[color=gray] (16,0) to[in=30,out=-30] (16,-1); 
	\draw (16,-2.5) node {$P_{123}:(\C^d)^{\otimes 3}\longrightarrow(\C^d)^{\otimes 3}$};
	\end{scope}
	\end{tikzpicture}
\end{center}
And in the case where $o_1\neq o_2$, it consists in showing that $(\widetilde{P}_{12}^{o_1}\otimes\Id_3) (\Id_1\otimes\widetilde{P}_{23}^{o_2})$ is close to $\widetilde{P}_{123}^{o_1,o_2}:=D^4P_{123}^{o_1,o_2}$, where $P_{123}^{o_1,o_2}$ is the same as $P_{123}^{o}$ but with one vertical and one horizontal indices swapped on site number $2$ (see the proof of Proposition \ref{prop:P-commute-PEPS} below for a precise definition).

For this we need, as before, to introduce two last operators $N,N'$ 
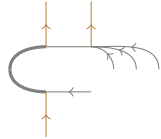
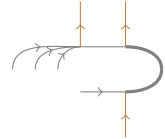
\begin{center}
	\begin{tikzpicture} [scale=0.6]
	\begin{scope}[decoration={markings,mark=at position 0.5 with {\arrow{>}}}] 
	\draw[postaction={decorate}, color=brown] (15,0) -- (15,1); \draw[postaction={decorate}, color=brown] (16,0) -- (16,1); \draw[postaction={decorate}, color=brown] (15,-2) -- (15,-1); \draw[color=gray] (15,0) -- (16,0); \draw[postaction={decorate}, color=gray] (16,-1) -- (15,-1); \draw[very thick, color=gray] (14.2,-0.5) to[in=-180,out=90] (15,0); \draw[very thick, color=gray] (14.2,-0.5) to[in=180,out=-90] (15,-1); \draw[postaction={decorate}, color=gray] (17,-0.5) to[in=-0,out=90] (16,0); \draw[postaction={decorate}, color=gray] (16.5,-0.5) to[in=-0,out=90] (16,0); \draw[postaction={decorate}, color=gray] (17.5,-0.5) to[in=-0,out=90] (16,0);
	\draw (15.5,-2.5) node {$N:\C^d\otimes(\C^D)^{\otimes 4}\longrightarrow(\C^d)^{\otimes 2}$};
	
	\draw[postaction={decorate}, color=brown] (25,0) -- (25,1); \draw[postaction={decorate}, color=brown] (26,0) -- (26,1); \draw[postaction={decorate}, color=brown] (26,-2) -- (26,-1); \draw[color=gray] (25,0) -- (26,0); \draw[postaction={decorate}, color=gray] (25,-1) -- (26,-1); \draw[postaction={decorate}, color=gray] (24,-0.5) to[in=-180,out=90] (25,0); \draw[postaction={decorate}, color=gray] (24.5,-0.5) to[in=-180,out=90] (25,0); \draw[postaction={decorate}, color=gray] (23.5,-0.5) to[in=-180,out=90] (25,0); \draw[very thick, color=gray] (26.8,-0.5) to[in=0,out=-90] (26,-1); \draw[very thick, color=gray] (26.8,-0.5) to[in=-0,out=90] (26,0); 
	\draw (25.5,-2.5) node {$N':(\C^D)^{\otimes 4}\otimes\C^d\longrightarrow(\C^d)^{\otimes 2}$};
	\end{scope}
	\end{tikzpicture}
\end{center}

\begin{proof}
	For the case where $o_1=o_2=:o$, we have the following chain of (in)equalities:
	\begin{align*} 
	\left\| \left(\widetilde{P}_{12}^{o}\otimes\Id_3\right)\left(\Id_1\otimes\widetilde{P}_{23}^{o}\right) - \widetilde{P}_{123}^{o} \right\|_{\infty} & = D^4 \left\| (N_{12}\otimes\Id_3)(\Id_1\otimes (D^2W''_2-\Id_2)\otimes\Id_3)(\Id_1\otimes N'^*_{23}) \right\|_{\infty} \\
	& \leq D^4 \|NN^*\|_{\infty} \|D^2W''-\Id\|_{\infty} \\
	& \leq D^5 \|Q\|_{\infty}^2 \|W\|_{\infty} \|D^2W''-\Id\|_{\infty} \, .
	\end{align*}
    Now, define the swap operator $S$ on $(\C^D)^{\otimes 4}$ as
    \[ S:\ket{v_1}\otimes\ket{v_2}\otimes\ket{h_1}\otimes\ket{h_2} \in (\C^D)^{\otimes 4} \mapsto \ket{v_1}\otimes\ket{h_1}\otimes\ket{v_2}\otimes\ket{h_2} \in (\C^D)^{\otimes 4}, \]
    and $\widetilde{P}_{123}^{o_1,o_2}$ the swapped version of $\widetilde{P}_{123}^{o}$ as
    \[ \widetilde{P}_{123}^{o_1,o_2} := (N_{12}\otimes\Id_3)(\Id_1\otimes S_2\otimes\Id_3)(\Id_1\otimes N'^*_{23}) \, . \]
    Then, for the case where $o_1\neq o_2$, we have the following chain of (in)equalities:
    \begin{align*} 
    	\left\| \left(\widetilde{P}_{12}^{o_1}\otimes\Id_3\right)\left(\Id_1\otimes\widetilde{P}_{23}^{o_2}\right) - \widetilde{P}_{123}^{o_1,o_2} \right\|_{\infty} & = D^4 \left\| (N_{12}\otimes\Id_3)(\Id_1\otimes (D^2W''_2-\Id_2)S_2\otimes\Id_3)(\Id_1\otimes N'^*_{23}) \right\|_{\infty} \\
    	& \leq D^4 \|NN^*\|_{\infty} \|D^2W''-\Id\|_{\infty}\|S\|_{\infty} \\
    	& \leq D^5 \|Q\|_{\infty}^2 \|W\|_{\infty} \|D^2W''-\Id\|_{\infty} \, ,
    \end{align*}
    where the next to last inequality is because $\|S\|_\infty=1$.

	In both cases, the last expression of the chain of (in)equalities is, with probability larger than $1-e^{-cD^4}$, smaller than $C/D^{2\tau -2}$, which concludes the proof.
\end{proof}

\subsubsection{Conclusions} \hfill\par\smallskip

Combining these two results we can then immediately deduce the analogue of Theorem \ref{th:Pi-commute}.
	
\begin{theorem} \label{th:Pi-commute-PEPS}
	Let $d\geq D^{4\tau}$, for some $\tau>13/2$. Then, there exist universal constants $c,C>0$ such that, for the orientations `$o_1$' and `$o_2$' being either `$v$' or `$h$',
	\[ \P\left( \left\| \left(\Pi_{12}^{o_1}\otimes\Id_3\right) \left(\Id_1\otimes\Pi_{23}^{o_2}\right) - \Pi^{o_1,o_2}_{123} \right\|_{\infty} \leq \frac{C}{D^{2\tau-13}} \right) \geq 1-e^{-cD^4} \, . \]	
\end{theorem}

And we can finally use Theorem \ref{th:Pi-commute-PEPS} above to derive Theorem \ref{th:gap-H-PEPS}, in the exact same way that Theorem \ref{th:gap-H} is derived from Theorem \ref{th:Pi-commute}.

\begin{theorem} \label{th:gap-H-PEPS}
	Let $d\geq D^{4\tau}$, for some $\tau>13/2$. Let $H_{\chi}$ be the parent Hamiltonian of $\ket{\chi^N}$, as defined by equation \eqref{eq:H-parent-PEPS}. Then, there exist universal constants $c,C>0$ such that
	\[ \P\left( \Delta(H_{\chi}) \geq 1-\frac{C}{D^{2\tau-13}} \right) \geq 1-e^{-cD^4} \, . \]
\end{theorem}

\section{Consequence: Typical correlation length in random MPS and PEPS}
\label{sec:decay-correlations-1}

In the previous section we showed that the parent Hamiltonians of our random MPS and PEPS are typically gapped, at least in a `super-injectivity' dimensional regime. In this section we derive from the latter result that our random MPS and PEPS typically exhibit exponential decay of correlations.  

Let us begin with explaining precisely what we mean when we talk about correlations in an MPS or a PEPS. Set $\tilde{N}:=N$ in the case of MPS and $\tilde{N}:=N^2$ in the case of PEPS. Let $\ket{\chi^{N}}\in(\C^d)^{\otimes \tilde{N}}$ be an $\tilde{N}$-site translation-invariant MPS or PEPS. Let $R,R'\subset[\tilde{N}]$ be such that $R\cap R'=\emptyset$ and let $A,A'$ be Hermitian operators on $(\C^d)^{\otimes |R|},(\C^d)^{\otimes |R'|}$. We would like to compare the value of the observable $A_{R}\otimes A'_{R'}\otimes\Id_{[\tilde{N}]\setminus R\cup R'}$ on $\ket{\chi^{N}}$ to the product of the values of $A_{R}\otimes\Id_{[\tilde{N}]\setminus R}$ and $A'_{R'}\otimes\Id_{[\tilde{N}]\setminus R'}$ on $\ket{\chi^{N}}$. So we define the correlation function
\begin{equation} \label{eq:def-C}
\gamma_{\chi}(A,A',R,R') := \left| \frac{\bra{\chi^{N}} A_{R}\otimes A'_{R'}\otimes\Id_{[\tilde{N}]\setminus R\cup R'} \ket{\chi^{N}}}{\braket{\chi^{N}}{\chi^{N}}} - \frac{\bra{\chi^{N}} A_{R}\otimes\Id_{[\tilde{N}]\setminus R} \ket{\chi^{N}} \bra{\chi^{N}} A'_{R'}\otimes\Id_{[\tilde{N}]\setminus R'} \ket{\chi^{N}}}{\braket{\chi^{N}}{\chi^{N}}^2} \right| \, , 
\end{equation}
and we ask whether it is small for $R,R'$ far way from each other. This is represented in the case of our random MPS in Figure \ref{fig:correlations}.

In what follows, given $R,R'\subset[\tilde{N}]$, we denote by $d(R,R')$ the graph distance between $R$ and $R'$, i.e.~the smallest number of edges separating a vertex in $R$ from a vertex in $R'$. And we will show that $\gamma_{\chi}(A,A',R,R')$ typically decays exponentially fast with $d(R,R')$, i.e.~$\gamma_{\chi}(A,A',R,R')\leq \Theta e^{-\tau d(R,R')}$ for some $\Theta,\tau>0$. We call the inverse of the rate $\tau$, $\xi:=1/\tau$, the correlation length of the MPS or PEPS.

\begin{figure}[h]
	\caption{Correlations in our random MPS}
	\label{fig:correlations}
\begin{center}
	\begin{tikzpicture} [scale=0.5]
	\draw[color=brown] (0,0) -- (0,0.5); \draw[color=brown] (0,1.5) -- (0,2); \draw[color=brown] (1,0) -- (1,2); \draw[color=brown] (2,0) -- (2,2); \draw[color=brown] (3,0) -- (3,2); \draw[color=brown] (4,0) -- (4,0.5); \draw[color=brown] (4,1.5) -- (4,2); \draw[color=brown] (5,0) -- (5,2); \draw[color=brown] (6,0) -- (6,2); \draw[color=brown] (7,0) -- (7,2); \draw[color=brown] (8,0) -- (8,2);
	\draw[color=gray] (0,0) -- (8,0); \draw[color=gray] (0,-1) -- (8,-1);
	\draw[color=gray] (0,0) to[in=90,out=-180] (-0.5,-0.5); \draw[color=gray] (0,-1) to[in=-90,out=180] (-0.5,-0.5); \draw[color=gray] (8,0) to[in=90,out=-0] (8.5,-0.5); \draw[color=gray] (8,-1) to[in=-90,out=0] (8.5,-0.5);
	\draw[color=gray] (0,2) -- (8,2); \draw[color=gray] (0,3) -- (8,3);
	\draw[color=gray] (0,2) to[in=-90,out=180] (-0.5,2.5); \draw[color=gray] (0,3) to[in=90,out=-180] (-0.5,2.5); \draw[color=gray] (8,2) to[in=-90,out=0] (8.5,2.5); \draw[color=gray] (8,3) to[in=90,out=-0] (8.5,2.5);
	\draw (0,1) node {{\small A}}; \draw (4,1) node {{\small A'}}; \draw (-0.5,0.5) rectangle (0.5,1.5); \draw (3.5,0.5) rectangle (4.5,1.5);
	\draw[decoration={brace,raise=7pt},decorate] (-0.2,2.7) -- node[above=9pt] {$N$} (8.2,2.7);
	\draw[decoration={brace,raise=7pt,mirror},decorate] (0.8,-0.7) -- node[below=9pt] {$d(R,R')$} (3.2,-0.7);
	
	\begin{scope}[xshift=15.5cm]
	\draw[color=brown] (0,0) -- (0,0.5); \draw[color=brown] (0,1.5) -- (0,2); \draw[color=brown] (1,0) -- (1,2); \draw[color=brown] (2,0) -- (2,2); \draw[color=brown] (3,0) -- (3,2); \draw[color=brown] (4,0) -- (4,2); \draw[color=brown] (5,0) -- (5,2); \draw[color=brown] (6,0) -- (6,2); \draw[color=brown] (7,0) -- (7,2); \draw[color=brown] (8,0) -- (8,2);
	\draw[color=gray] (0,0) -- (8,0); \draw[color=gray] (0,-1) -- (8,-1);
	\draw[color=gray] (0,0) to[in=90,out=-180] (-0.5,-0.5); \draw[color=gray] (0,-1) to[in=-90,out=180] (-0.5,-0.5); \draw[color=gray] (8,0) to[in=90,out=-0] (8.5,-0.5); \draw[color=gray] (8,-1) to[in=-90,out=0] (8.5,-0.5);
	\draw[color=gray] (0,2) -- (8,2); \draw[color=gray] (0,3) -- (8,3);
	\draw[color=gray] (0,2) to[in=-90,out=180] (-0.5,2.5); \draw[color=gray] (0,3) to[in=90,out=-180] (-0.5,2.5); \draw[color=gray] (8,2) to[in=-90,out=0] (8.5,2.5); \draw[color=gray] (8,3) to[in=90,out=-0] (8.5,2.5);
	\draw (0,1) node {{\small A}}; \draw (-0.5,0.5) rectangle (0.5,1.5); 
	\end{scope}
	
	\begin{scope}[xshift=25.5cm]
	\draw[color=brown] (0,0) -- (0,2); \draw[color=brown] (1,0) -- (1,2); \draw[color=brown] (2,0) -- (2,2); \draw[color=brown] (3,0) -- (3,2); \draw[color=brown] (4,0) -- (4,0.5); \draw[color=brown] (4,1.5) -- (4,2); \draw[color=brown] (5,0) -- (5,2); \draw[color=brown] (6,0) -- (6,2); \draw[color=brown] (7,0) -- (7,2); \draw[color=brown] (8,0) -- (8,2);
	\draw[color=gray] (0,0) -- (8,0); \draw[color=gray] (0,-1) -- (8,-1);
	\draw[color=gray] (0,0) to[in=90,out=-180] (-0.5,-0.5); \draw[color=gray] (0,-1) to[in=-90,out=180] (-0.5,-0.5); \draw[color=gray] (8,0) to[in=90,out=-0] (8.5,-0.5); \draw[color=gray] (8,-1) to[in=-90,out=0] (8.5,-0.5);
	\draw[color=gray] (0,2) -- (8,2); \draw[color=gray] (0,3) -- (8,3);
	\draw[color=gray] (0,2) to[in=-90,out=180] (-0.5,2.5); \draw[color=gray] (0,3) to[in=90,out=-180] (-0.5,2.5); \draw[color=gray] (8,2) to[in=-90,out=0] (8.5,2.5); \draw[color=gray] (8,3) to[in=90,out=-0] (8.5,2.5);
	\draw (4,1) node {{\small A'}}; \draw (3.5,0.5) rectangle (4.5,1.5);
	\end{scope}
	
	\draw (11.75,1) node {$\simeq$}; \draw (11.75,1.7) node {$?$}; \draw (11.75,0.2) node {$N\gg d(R,R')\gg 1$};
	\draw (24.5,1) node {$\times$};
	\end{tikzpicture}
\end{center}
\end{figure}

It was seminally observed in \cite{HK} that an MPS or a PEPS exhibiting exponential decay of correlations can be derived from its parent Hamiltonian being gapped. Here we first show how a more basic approach already gives such kind of statement, even though with a non-optimal scaling. We then proceed to improving this result by following a route more similar to that of \cite{HK}.

\subsection{Rough upper bound on the typical correlation length via the detectability lemma} \hfill\par\smallskip

Our first strategy to prove typical exponential decay of correlations in our random MPS and PEPS, from the statements of Section \ref{sec:parent-hamiltonian} on the typical spectral gap of their parent Hamiltonian, is to make use of a result proved in \cite{GH}. The latter relies on the detectability lemma, first introduced in \cite{AALV} and later improved and simplified in \cite{AAV}. The reasoning is in fact entirely the same for MPS and PEPS. The only thing that changes is the range of physical and bond dimensions for which we are able to say something, the constraints being exactly those of either Theorem \ref{th:gap-H} or Theorem \ref{th:gap-H-PEPS}.

Let us start with the case of MPS.

\begin{theorem} \label{th:dc1-MPS}
	Let $d\geq D^{2\tau}$, for some $\tau>5$, and $D\geq D_0$, where $D_0>0$ is a universal constant. Let $\ket{\chi^{N}}\in(\C^d)^{\otimes N}$ be the random $N$-site translation-invariant MPS whose random $1$-site tensor $\ket{\chi}\in\C^d\otimes(\C^D)^{\otimes 2}$ is defined as in equation \eqref{eq:MPS}. Then, with probability larger than $1-e^{-cD^2}$, for any $R,R'\subset\{1,\ldots,N\}$ such that $R\cap R'=\emptyset$ and any Hermitian operators $A,A'$ on $(\C^d)^{\otimes |R|},(\C^d)^{\otimes |R'|}$,
	\[ \gamma_{\chi}(A,A',R,R') \leq e^{-c'd(R,R')}\|A\|_{\infty}\|A'\|_{\infty} \, , \]
	where $c,c'>0$ are universal constants.
\end{theorem}

\begin{proof}
Since $H_{\chi}$ is a frustration-free local Hamiltonian, we know by \cite[Theorem 1]{GH} that, if it has a spectral gap $\Delta$, then
\[ \gamma_{\chi}(A,A',R,R') \leq e^{-c_0d(R,R')\sqrt{\Delta}}\|A\|_{\infty}\|A'\|_{\infty} \, . \]
Now, we also know by Theorem \ref{th:gap-H} that, with probability larger than $1-e^{-cD^2}$, $\Delta(H_{\chi})\geq 1-C/D^{\tau-5}$, which is larger than (say) $1/2$ for $D$ large enough. And the proof is thus complete (re-labelling $c_0/\sqrt{2}$ into $c'$).
\end{proof}

Let us now turn to the case of PEPS, which is treated in the exact same way as the case of MPS.

\begin{theorem} \label{th:dc1-PEPS}
	Let $d\geq D^{4\tau}$, for some $\tau>13/2$, and $D\geq D_0$, where $D_0>0$ is a universal constant. Let $\ket{\chi^N}\in(\C^d)^{\otimes N^2}$ be the random $N^2$-site translation-invariant PEPS whose random $1$-site tensor $\ket{\chi}\in\C^d\otimes(\C^D)^{\otimes 4}$ is defined as in equation \eqref{eq:PEPS}. Then, with probability larger than $1-e^{-cD^4}$, for any $R,R'\subset\{1,\ldots,N^2\}$ such that $R\cap R'=\emptyset$ and any Hermitian operators $A,A'$ on $(\C^d)^{\otimes |R|},(\C^d)^{\otimes |R'|}$,
	\[ \gamma_{\chi}(A,A',R,R') \leq e^{-c'd(R,R')}\|A\|_{\infty}\|A'\|_{\infty} \, , \]
	where $c,c'>0$ are universal constants.
\end{theorem}

\begin{proof}
Since $H_{\chi}$ is a frustration-free local Hamiltonian, we know by \cite[Theorem 1]{GH} that, if it has a spectral gap $\Delta$, then
\[ \gamma_{\chi}(A,A',R,R') \leq e^{-c'd(R,R')\sqrt{\Delta}}\|A\|_{\infty}\|A'\|_{\infty} \, . \]
Now, we also know by Theorem \ref{th:gap-H-PEPS} that, with probability larger than $1-e^{-cD^4}$, $\Delta(H_{\chi})\geq 1-C/D^{2\tau-13}$, which is larger than (say) $1/2$ for $D$ large enough. And the proof is thus complete (re-labelling $c_0/\sqrt{2}$ into $c'$).
\end{proof}

To summarize, we have shown that, as a straightforward consequence of Theorem \ref{th:gap-H} and Theorem \ref{th:gap-H-PEPS}, our random MPS and PEPS typically exhibit exponential decay of correlation at a rate which is at least a constant independent of any other parameter (physical dimension $d$, bond dimension $D$, number of particles $\tilde{N}$).

\subsection{Tighter upper bound on the typical correlation length via a refined Lieb--Robinson bound} \hfill\par\smallskip

It is actually possible to improve the previous result, namely a typical upper bound on the correlation length of our random MPS and PEPS of order $1$, to a typical upper bound of order $1/\log D$. To achieve this we first use a Lieb--Robinson bound, recently proved in \cite{HHKL}, which is suited to the case where the local terms composing the Hamiltonian have small commutators. From there we derive exponential decay of correlations by following the same reasoning as the one detailed, for instance, in \cite{Has3}.

Set again $\tilde{N}:=N$ in the case of MPS and $\tilde{N}:=N^2$ in the case of PEPS, and let $\ket{\chi^{N}}\in(\C^d)^{\otimes \tilde{N}}$ be an $\tilde{N}$-site translation-invariant MPS or PEPS with parent Hamiltonian $H_{\chi}$. Then, for any Hermitian operator $A$ on $(\C^d)^{\otimes \tilde{N}}$ and any $t\in\R$, define
\[ A(t) := \exp(itH_{\chi})A\exp(-itH_{\chi}) \, . \]
In the sequel, for any $R\subset\{1,\ldots,\tilde{N}\}$ and any Hermitian operator $A$ on $(\C^d)^{\otimes |R|}$, we will use the short-hand notation $A_R$ for $A_R\otimes \Id_{[\tilde{N}]\setminus R}$.

Let us start with the case of MPS.

\begin{lemma} \label{lem:LR-MPS}
	Let $d\geq D^{2\tau}$, for some $\tau>5$. Let $\ket{\chi^{N}}\in(\C^d)^{\otimes N}$ be the random $N$-site translation-invariant MPS whose random $1$-site tensor $\ket{\chi}\in\C^d\otimes(\C^D)^{\otimes 2}$ is defined as in equation \eqref{eq:MPS}. Then, with probability larger than $1-e^{-cD^2}$, for any $R,R'\subset\{1,\ldots,N\}$ such that $R\cap R'=\emptyset$ and any Hermitian operators $A,A'$ on $(\C^d)^{\otimes |R|},(\C^d)^{\otimes |R'|}$, for any $t\in\R,\mu\in\R^+$,
	\[ \left\| \left[ A_R(t),A'_{R'} \right] \right\|_{\infty} \leq CD^{(\tau-5)/2} \left( \exp\left(C'e^{2\mu}|t|/D^{(\tau-5)/2}\right) -1 \right) |R|e^{-\mu d(R,R')} \|A_R\|_{\infty}\|A'_{R'}\|_{\infty} \, , \]
	where $c,C,C'>0$ are universal constants.
\end{lemma}

\begin{proof}
The result immediately follows from \cite[Lemma 12]{HHKL}. We just have to see how the commutator and exponential decay conditions (equations (14) and (15) of \cite{HHKL}) read in our case.
Recall that, using the same notation as in Section \ref{sec:parent-hamiltonian-MPS}, $H_{\chi} = \sum_{i=1}^N \hat{\Pi}_{i,i+1}$. Let us start with the bound on the commutator: If $|i-j|>1$, then 
\[  \left\| \left[ \hat{\Pi}_{i,i+1},\hat{\Pi}_{j,j+1} \right] \right\|_{\infty}=0 \, . \] 
While we know by Theorem \ref{th:Pi-commute} that, with probability larger than $1-e^{-cD^2}$, 
\[  \left\| \left[ \hat{\Pi}_{i-1,i},\hat{\Pi}_{i,i+1} \right] \right\|_{\infty} \leq \frac{C}{D^{\tau-5}} \, . \] 
Let us now turn to the bound on the exponential decay: Setting $I^-=\{i-1,i\}$, $I^+=\{i,i+1\}$, we have
\[ |I^-|^2\|\hat{\Pi}_{I^-}\|_{\infty}e^{\mu\mathrm{diam}(I^-)} + |I^+|^2\|\hat{\Pi}_{I^+}\|_{\infty}e^{\mu\mathrm{diam}(I^+)} = 8e^{2\mu} \, . \]
Plugging these values into equation (16) of \cite{HHKL} gives exactly the announced result.	
\end{proof}

\begin{corollary} \label{cor:LR-MPS}
Let $d\geq D^{2\tau}$, for some $\tau>5$. Let $\ket{\chi^{N}}\in(\C^d)^{\otimes N}$ be the random $N$-site translation-invariant MPS whose random $1$-site tensor $\ket{\chi}\in\C^d\otimes(\C^D)^{\otimes 2}$ is defined as in equation \eqref{eq:MPS}. Then, with probability larger than $1-e^{-cD^2}$, for any $R,R'\subset\{1,\ldots,N\}$ such that $R\cap R'=\emptyset$ and any Hermitian operators $A,A'$ on $(\C^d)^{\otimes |R|},(\C^d)^{\otimes |R'|}$, for any $t\in\R$ such that $|t|\leq c'\log(D^{\tau-5})d(R,R')$,
\[ \left\| \left[ A_R(t),A'_{R'} \right] \right\|_{\infty} \leq  |R|e^{-c''\log(D^{\tau-5}) d(R,R')} \|A_R\|_{\infty}\|A'_{R'}\|_{\infty} \, , \]
where $c,c',c''>0$ are universal constants.
\end{corollary}

\begin{proof}
Taking $\mu=\log(D^{\tau-5})/4$ in Lemma \ref{lem:LR-MPS}, we get	
\[ \left\| \left[ A_R(t),A'_{R'} \right] \right\|_{\infty} \leq CD^{(\tau-5)/2} \left( e^{C'|t|} -1 \right) |R|e^{-\log(D^{\tau-5}) d(R,R')/4} \|A_R\|_{\infty}\|A'_{R'}\|_{\infty} \, . \]
And as soon as (say) $|t|\leq \log(D^{\tau-5}) d(R,R')/8C'$, we have
\[ CD^{(\tau-5)/2} \left( e^{C'|t|} -1 \right)e^{-\log(D^{\tau-5}) d(R,R')/4} \leq e^{-c'\log(D^{\tau-5}) d(R,R')} \, , \]
which, up to re-labelling the constants, is exactly the claimed result.
\end{proof}

\begin{theorem} \label{th:dc1'-MPS}
	Let $d\geq D^{2\tau}$, for some $\tau>5$. Let $\ket{\chi^{N}}\in(\C^d)^{\otimes N}$ be the random $N$-site translation-invariant MPS whose random $1$-site tensor $\ket{\chi}\in\C^d\otimes(\C^D)^{\otimes 2}$ is defined as in equation \eqref{eq:MPS}. Then, with probability larger than $1-e^{-cD^2}$, for any $R,R'\subset\{1,\ldots,N\}$ such that $R\cap R'=\emptyset$ and any Hermitian operators $A,A'$ on $(\C^d)^{\otimes |R|},(\C^d)^{\otimes |R'|}$, 
	\[ \gamma_{\chi}(A,A',R,R') \leq \min(|R|,|R'|)e^{-c'\log(D^{\tau-5}) d(R,R')} \|A_R\|_{\infty}\|A'_{R'}\|_{\infty} \, , \]
	where $c,c'>0$ are universal constants.
\end{theorem}

\begin{proof}
We proceed in the exact same way as how, in \cite{Has3}, Theorem 2 is proved from Theorem 1. We thus get from Corollary \ref{cor:LR-MPS} that, if $H_{\chi}$ has a spectral gap $\Delta$, then with probability larger than $1-e^{-cD^2}$,
\[ \gamma_{\chi}(A,A',R,R') \leq C\left( e^{-c'\Delta\log(D^{\tau-5})d(R,R')} + \min(|R|,|R'|)e^{-c''\log(D^{\tau-5}) d(R,R')} \right) \|A_R\|_{\infty}\|A'_{R'}\|_{\infty} \, . \]
Now, we also know by Theorem \ref{th:gap-H} that, with probability larger than $1-e^{-cD^2}$, $\Delta(H_{\chi})\geq 1-C/D^{\tau-5}$, which is larger than (say) $1/2$ for $D$ large enough. And therefore, with probability larger than $1-2e^{-cD^2}$,
\[ \gamma_{\chi}(A,A',R,R') \leq \min(|R|,|R'|)e^{-\hat{c}\log(D^{\tau-5}) d(R,R')} \|A_R\|_{\infty}\|A'_{R'}\|_{\infty} \, , \]
which, up to re-labelling the constants, is precisely the announced result.
\end{proof}

Let us now turn to the case of PEPS, which here again can be analysed just as the case of MPS.

\begin{lemma} \label{lem:LR-PEPS}
	Let $d\geq D^{4\tau}$, for some $\tau>13/2$. Let $\ket{\chi^N}\in(\C^d)^{\otimes N^2}$ be the random $N^2$-site translation-invariant PEPS whose random $1$-site tensor $\ket{\chi}\in\C^d\otimes(\C^D)^{\otimes 4}$ is defined as in equation \eqref{eq:PEPS}. Then, with probability larger than $1-e^{-cD^4}$, for any $R,R'\subset\{1,\ldots,N^2\}$ such that $R\cap R'=\emptyset$ and any Hermitian operators $A,A'$ on $(\C^d)^{\otimes |R|},(\C^d)^{\otimes |R'|}$, for any $t\in\R,\mu\in\R^+$,
	\[ \left\| \left[ A_R(t),A'_{R'} \right] \right\|_{\infty} \leq CD^{\tau-13/2} \left( \exp\left(C'e^{2\mu}|t|/D^{\tau-13/2}\right) -1 \right) |R|e^{-\mu d(R,R')} \|A_R\|_{\infty}\|A'_{R'}\|_{\infty} \, , \]
	where $c,C,C'>0$ are universal constants.
\end{lemma}

\begin{proof}
	The result immediately follows from \cite[Lemma 12]{HHKL}. We just have to see how the commutator and exponential decay conditions (equations (14) and (15) of \cite{HHKL}) read in our case.
	Recall that, using the same notation as in Section \ref{sec:parent-hamiltonian-PEPS}, $H_{\chi} = \sum_{j=1}^N \left( \sum_{i=1}^N \hat{\Pi}^v_{(i,j),(i+1,j)} \right) + \sum_{i=1}^N \left( \sum_{j=1}^N \hat{\Pi}^h_{(i,j),(i,j+1)} \right)$. Let us start with the bound on the commutator: If $|i-i'|>1$ or $j\neq j'$, then 
	\[  \left\| \left[ \hat{\Pi}^v_{(i,j),(i+1,j)}, \hat{\Pi}^v_{(i',j'),(i'+1,j')} \right] \right\|_{\infty}=0 \, , \] 
	if $|j-j'|>1$ or $i\neq i'$, then 
	\[  \left\| \left[ \hat{\Pi}^h_{(i,j),(i,j+1)}, \hat{\Pi}^h_{(i',j'),(i',j'+1)} \right] \right\|_{\infty}=0 \, , \] 
	and if $i'\notin\{i,i+1\}$ or $j'\notin\{j,j-1\}$, then
	\[  \left\| \left[ \hat{\Pi}^v_{(i,j),(i+1,j)}, \hat{\Pi}^h_{(i',j'),(i',j'+1)} \right] \right\|_{\infty}=0 \, . \] 
	While we know by Theorem \ref{th:Pi-commute} that, with probability larger than $1-e^{-cD^4}$,
	\begin{align*}
	& \left\| \left[ \hat{\Pi}^v_{(i-1,j),(i,j)}, \hat{\Pi}^v_{(i,j),(i+1,j)} \right] \right\|_{\infty} \leq \frac{C}{D^{2\tau-13}} \, , \\
	& \left\| \left[ \hat{\Pi}^h_{(i,j-1),(i,j)}, \hat{\Pi}^h_{(i,j),(i,j+1)} \right] \right\|_{\infty} \leq \frac{C}{D^{2\tau-13}} \, , \\
	& \left\| \left[ \hat{\Pi}^v_{(i,j),(i+1,j)}, \hat{\Pi}^h_{(i',j'),(i',j'+1)} \right] \right\|_{\infty} \leq \frac{C}{D^{2\tau-13}} \text{ for } i'\in\{i,i+1\},j'\in\{j,j-1\} \, .
	\end{align*} 
	Let us now turn to the bound on the exponential decay: Setting $I_v^-=\{(i-1,j),(i,j)\}$, $I_v^+=\{(i,j),(i+1,j)\}$, $I_h^-=\{(i,j-1),(i,j)\}$, $I_h^+=\{(i,j),(i,j+1)\}$, we have
	\[ |I_v^-|^2\|\hat{\Pi}^v_{I_v^-}\|_{\infty}e^{\mu\mathrm{diam}(I_v^-)} + |I_v^+|^2\|\hat{\Pi}^v_{I_v^+}\|_{\infty}e^{\mu\mathrm{diam}(I_v^+)} + |I_h^-|^2\|\hat{\Pi}^h_{I_h^-}\|_{\infty}e^{\mu\mathrm{diam}(I_h^-)} + |I_h^+|^2\|\hat{\Pi}_{I_h^+}\|_{\infty}e^{\mu\mathrm{diam}(I_h^+)} = 16e^{2\mu} \, . \]
	Plugging these values into equation (16) of \cite{HHKL} gives exactly the announced result.	
\end{proof}

\begin{corollary} \label{cor:LR-PEPS}
	Let $d\geq D^{4\tau}$, for some $\tau>13/2$. Let $\ket{\chi^N}\in(\C^d)^{\otimes N^2}$ be the random $N^2$-site translation-invariant PEPS whose random $1$-site tensor $\ket{\chi}\in\C^d\otimes(\C^D)^{\otimes 4}$ is defined as in equation \eqref{eq:PEPS}. Then, with probability larger than $1-e^{-cD^4}$, for any $R,R'\subset\{1,\ldots,N^2\}$ such that $R\cap R'=\emptyset$ and any Hermitian operators $A,A'$ on $(\C^d)^{\otimes |R|},(\C^d)^{\otimes |R'|}$, for any $t\in\R$ such that $|t|\leq c'\log(D^{2\tau-13})d(R,R')$,
	\[ \left\| \left[ A_R(t),A'_{R'} \right] \right\|_{\infty} \leq |R|e^{-c''\log(D^{2\tau-13}) d(R,R')} \|A_R\|_{\infty}\|A'_{R'}\|_{\infty} \, , \]
	where $c,c',c''>0$ are universal constants.
\end{corollary}

\begin{proof}
Taking $\mu=\log(D^{2\tau-13})/4$ in Lemma \ref{lem:LR-PEPS}, we get	
\[ \left\| \left[ A_R(t),A'_{R'} \right] \right\|_{\infty} \leq CD^{\tau-13/2} \left( e^{C'|t|} -1 \right) |R|e^{-\log(D^{2\tau-13}) d(R,R')/4} \|A_R\|_{\infty}\|A'_{R'}\|_{\infty} \, . \]
And as soon as (say) $|t|\leq \log(D^{2\tau-13}) d(R,R') /8C'$, we have
\[ CD^{\tau-13/2} \left( e^{C'|t|} -1 \right)e^{-\log(D^{2\tau-13}) d(R,R')/4} \leq e^{-c'\log(D^{2\tau-13}) d(R,R')} \, , \]
which, up to re-labelling the constants, is exactly the claimed result.
\end{proof}

\begin{theorem} \label{th:dc1'-PEPS}
	Let $d\geq D^{4\tau}$, for some $\tau>13/2$. Let $\ket{\chi^N}\in(\C^d)^{\otimes N^2}$ be the random $N^2$-site translation-invariant PEPS whose random $1$-site tensor $\ket{\chi}\in\C^d\otimes(\C^D)^{\otimes 4}$ is defined as in equation \eqref{eq:PEPS}. Then, with probability larger than $1-e^{-cD^4}$, for any $R,R'\subset\{1,\ldots,N^2\}$ such that $R\cap R'=\emptyset$ and any Hermitian operators $A,A'$ on $(\C^d)^{\otimes |R|},(\C^d)^{\otimes |R'|}$,
	\[ \gamma_{\chi}(A,A',R,R') \leq \min(|R|,|R'|)e^{-c'\log(D^{2\tau-13}) d(R,R')} \|A_R\|_{\infty}\|A'_{R'}\|_{\infty} \, , \]
	where $c,c'>0$ are universal constants.
\end{theorem}

\begin{proof}
	We proceed in the exact same way as how, in \cite{Has3}, Theorem 2 is proved from Theorem 1. We thus get from Corollary \ref{cor:LR-PEPS} that, if $H_{\chi}$ has a spectral gap $\Delta$, then with probability larger than $1-e^{-cD^4}$,
	\[ \gamma_{\chi}(A,A',R,R') \leq C\left( e^{-c'\Delta\log(D^{2\tau-13})d(R,R')} + \min(|R|,|R'|)e^{-c''\log(D^{2\tau-13}) d(R,R')} \right) \|A_R\|_{\infty}\|A'_{R'}\|_{\infty} \, . \]
	Now, we also know by Theorem \ref{th:gap-H-PEPS} that, with probability larger than $1-e^{-cD^4}$, $\Delta(H_{\chi})\geq 1-C/D^{2\tau-13}$, which is larger than (say) $1/2$ for $D$ large enough. And therefore, with probability larger than $1-2e^{-cD^4}$,
	\[ \gamma_{\chi}(A,A',R,R') \leq \min(|R|,|R'|)e^{-\hat{c}\log(D^{2\tau-13}) d(R,R')} \|A_R\|_{\infty}\|A'_{R'}\|_{\infty} \, , \]
	which, up to re-labelling the constants, is precisely the announced result.
\end{proof}

To summarize, we have derived from Theorems \ref{th:Pi-commute}, \ref{th:gap-H} and Theorems \ref{th:Pi-commute-PEPS}, \ref{th:gap-H-PEPS} that our random MPS and PEPS typically exhibit exponential decay of correlation at a rate which is at least of order $\log D$.

It is quite instructive to look at how the results of Theorem \ref{th:dc1'-PEPS} get modified under blocking. There are two ways in which one can do such grouping of sites: either before or after sampling the random $1$-site tensor. In both cases, the blocking procedure goes as follows: We start from a square lattice with $\underline{N}\times\underline{N}$ sites, where $\underline{N}:=N\sqrt{\log N}$, each having physical dimension $\underline{d}$ and bond dimension $\underline{D}$. We then redefine $1$ site as being a square of $\sqrt{\log N}\times\sqrt{\log N}$ sites. We thus obtain a square lattice with $N\times N$ sites, each having physical dimension $d:=\underline{d}^{\log N}$ and bond dimension $D:=\underline{D}^{\sqrt{\log N}}$. 

Let us first look at the simplest situation to analyse, i.e.~the one where we do the redefinition of sites before the sampling. In this case, we only have to plug the scaling for $d,D$ in Theorem \ref{th:dc1'-PEPS}, to obtain the result below. 

\begin{theorem} \label{th:dc1'-PEPS-block1}
	Fix $\underline{d},\underline{D}\in\N$ and let $d=\underline{d}^{\log N},D=\underline{D}^{\sqrt{\log N}}$. Let $\ket{\chi^N}\in(\C^d)^{\otimes N^2}$ be the random $N^2$-site translation-invariant PEPS whose random $1$-site tensor $\ket{\chi}\in\C^d\otimes(\C^D)^{\otimes 4}$ is defined as in equation \eqref{eq:PEPS}. Then, with probability larger than $1-e^{-c\underline{D}^{4\sqrt{\log N}}}$, for any $R,R'\subset\{1,\ldots,N^2\}$ such that $R\cap R'=\emptyset$ and any Hermitian operators $A,A'$ on $(\C^d)^{\otimes |R|},(\C^d)^{\otimes |R'|}$,
	\[ \gamma_{\chi}(A,A',R,R') \leq \min(|R|,|R'|)e^{-c'(\log\underline{d})(\log N)d(R,R')} \|A_R\|_{\infty}\|A'_{R'}\|_{\infty} \, , \]
	where $c,c'>0$ are universal constants.
\end{theorem}

The situation where the redefinition of sites is done after the sampling is only slightly more subtle to deal with. In this case, it is the scaling for $|R|,|R'|$ and $d(R,R')$ that we have to plug in Theorem \ref{th:dc1'-PEPS}. Indeed, site $(i,j)$ in the new lattice actually corresponds to the square of sites $\{(i-1)\sqrt{\log N},\ldots,i\sqrt{\log N}-1\} \times \{(j-1)\sqrt{\log N},\ldots,j\sqrt{\log N}-1\}$ in the old lattice. Hence, regions $R,R'$ in the new lattice correspond to regions $\underline{R},\underline{R}'$ in the old lattice which are such that $|\underline{R}|=(\log N)|R|,|\underline{R}'|=(\log N)|R'|$ and $d(\underline{R},\underline{R}')=\sqrt{\log N}d(R,R')$. We thus get the result below.

\begin{theorem} \label{th:dc1'-PEPS-block2}
	Fix $\underline{d},\underline{D}\in\N$ with $\underline{d}>\underline{D}^{26}$ and let $d=\underline{d}^{\log N}$. Set also $\underline{N}= N\sqrt{\log N}$. Let $\ket{\chi^{\underline{N}}}\in(\C^{\hat{d}})^{\otimes \underline{N}^2}$ be the random $\underline{N}^2$-site translation-invariant PEPS whose random $1$-site tensor $\ket{\chi}\in\C^{\underline{d}}\otimes(\C^{\underline{D}})^{\otimes 4}$ is defined as in equation \eqref{eq:PEPS}. Then, with probability larger than $1-e^{-c\underline{D}^4}$, for any $R,R'\subset\{1,\ldots,N^2\}$ such that $R\cap R'=\emptyset$ and any Hermitian operators $A,A'$ on $(\C^d)^{\otimes |R|},(\C^d)^{\otimes |R'|}$,
	\[ \gamma_{\chi}(A,A',R,R') \leq (\log N)\min(|R|,|R'|)e^{-c'(\log\underline{d})\sqrt{\log N}d(R,R')} \|A_R\|_{\infty}\|A'_{R'}\|_{\infty} \, , \]
	where $c,c'>0$ are universal constants.
\end{theorem}

Let us just make one last comment about the parent Hamiltonian $H_{\chi}$ of $\ket{\chi^N}$ in this latter case. It takes the same form as before, i.e.
\[ H_{\chi} := \sum_{j=1}^N \left( \sum_{i=1}^N \hat{\Pi}^v_{(i,j),(i+1,j)} \right) + \sum_{i=1}^N \left( \sum_{j=1}^N \hat{\Pi}^h_{(i,j),(i,j+1)} \right)  \, , \]
except that, now, site $(i,j)$ is the square of sites $\{(i-1)\sqrt{\log N},\ldots,i\sqrt{\log N}-1\} \times \{(j-1)\sqrt{\log N},\ldots,j\sqrt{\log N}-1\}$. The interesting fact to point out is that this grouping makes the new $2$-site projectors $\hat{\Pi}$ commute more than the old ones $\underline{\hat{\Pi}}$. This intuitive statement is made quantitative by \cite[Theorem 3]{KL}, which upper bounds the commutator of the new projectors $\hat{\Pi}$ in terms of the spectral gap of the old parent Hamiltonian $\underline{H}_{\chi}$. More precisely, we thus get: If $|i-i'|>1$ or $j\neq j'$, then 
\[  \left\| \left[ \hat{\Pi}^v_{(i,j),(i+1,j)}, \hat{\Pi}^v_{(i',j'),(i'+1,j')} \right] \right\|_{\infty}=0 \, , \] 
if $|j-j'|>1$ or $i\neq i'$, then 
\[  \left\| \left[ \hat{\Pi}^h_{(i,j),(i,j+1)}, \hat{\Pi}^h_{(i',j'),(i',j'+1)} \right] \right\|_{\infty}=0 \, , \] 
and if $i'\notin\{i,i+1\}$ or $j'\notin\{j,j-1\}$, then
\[  \left\| \left[ \hat{\Pi}^v_{(i,j),(i+1,j)}, \hat{\Pi}^h_{(i',j'),(i',j'+1)} \right] \right\|_{\infty}=0 \, . \] 
While we know by Theorem \ref{th:gap-H-PEPS} that, with probability larger than $1-e^{-c\underline{D}^4}$, $\Delta(\underline{H}_{\chi})\geq 1-C/\underline{D}^{2\tau-13}$, so that by \cite[Theorem 3]{KL}
\begin{align*}  
& \left\| \left[ \hat{\Pi}^v_{(i-1,j),(i,j)}, \hat{\Pi}^v_{(i,j),(i+1,j)} \right] \right\|_{\infty} \leq 2\left(\frac{1}{1+\kappa\big(1-C/\underline{D}^{2\tau-13}\big)}\right)^{(\log N)/2}\leq 2\left(\frac{1}{1+\kappa'}\right)^{\log N} \, , \\
& \left\| \left[ \hat{\Pi}^h_{(i,j-1),(i,j)}, \hat{\Pi}^h_{(i,j),(i,j+1)} \right] \right\|_{\infty} \leq 2\left(\frac{1}{1+\kappa\big(1-C/\underline{D}^{2\tau-13}\big)}\right)^{(\log N)/2} \leq 2\left(\frac{1}{1+\kappa'}\right)^{\log N} \, , \\
& \left\| \left[ \hat{\Pi}^v_{(i,j),(i+1,j)}, \hat{\Pi}^h_{(i',j'),(i',j'+1)} \right] \right\|_{\infty} \leq 2\left(\frac{1}{1+\kappa\left(1-C/\underline{D}^{2\tau-13}\right)}\right)^{(\log N)/2} \leq 2\big(\frac{1}{1+\kappa'}\big)^{\log N} \text{ for } \begin{cases} i'\in\{i,i+1\} \\ j'\in\{j,j-1\} \end{cases} .
\end{align*} 

Comparing Theorems \ref{th:dc1'-PEPS-block1} and \ref{th:dc1'-PEPS-block2}, we see that they yield a typical correlation length of order $1/\log N$ for the former and $1/\sqrt{\log N}$ for the latter. The result of Theorem \ref{th:dc1'-PEPS-block1} is absolutely not surprising: blocking before sampling the random tensor simply means that the physical and bond dimensions of $1$ site have been scaled up, so that the correlation decay rate is expected to scale up accordingly. In contrast, the result of Theorem \ref{th:dc1'-PEPS-block2} is slightly more subtle. Also, since in this second case the random tensor is sampled on a site having physical and bond dimensions $\underline{d}$ and $\underline{D}$, these need to be large for the result to actually hold with probability close to $1$. While in the first case the random tensor is sampled on a site having physical and bond dimensions $\underline{d}^{\log N}$ and $\underline{D}^{\sqrt{\log N}}$, which are automatically large as $N$ grows.

\section{Typical spectral gap of the transfer operator of random MPS and PEPS}
\label{sec:transfer-operator}

In this section, in contrast to the two previous ones, we do not constrain our random MPS and PEPS to be injective. What we want to show here is that their associated transfer operators are typically gapped. In the MPS case, treated in Section \ref{sec:transfer-operator-MPS}, we can prove this, in a quantitative way, for any physical and bond dimensions (see Theorem \ref{th:gap-MPS}). On the contrary, in the PEPS case, treated in Section \ref{sec:transfer-operator-PEPS}, we need to impose that the physical and bond dimensions grow polynomially with the number of particles and scale in a specific way with respect to one another (see Theorem \ref{th:gap-PEPS}). 

\subsection{Toolbox and strategy} \hfill\par\smallskip

Our goal here will be to show that the random transfer operators $T$ and $T_N$, as defined by equations \eqref{eq:transfer-MPS} and \eqref{eq:transfer-PEPS}, typically have a large (upper) spectral gap. For this we will make use of two technical results, providing variational formulas for the singular values of a matrix (see e.g.~\cite[Problem III.6.1]{Bha}) and a majorization result between the eigenvalues and the singular values of a matrix (see e.g.~\cite[Theorem II.3.6]{Bha}). 

Before stating them, let us fix some notation. Given an $n\times n$ complex matrix $M$ we denote by $\lambda_1(M),\ldots,\lambda_n(M)$ its eigenvalues, ordered so that $|\lambda_1(M)| \geq\cdots\geq |\lambda_n(M)|$, and by $s_1(M) \geq\cdots\geq s_n(M)\geq 0$ its singular values. We furthermore define its upper spectral gap as $\Delta(M):=|\lambda_1(M)|-|\lambda_2(M)|$. This is the same notation as the one we were using in Sections \ref{sec:parent-hamiltonian} and \ref{sec:decay-correlations-1} for the lower spectral gap, but there should be no possible confusion.

Also, given an $n^2\times n^2$ complex matrix $M$, we will denote by $\mathcal M$ its corresponding map on $n\times n$ complex matrices. Formally, we identify 
\[ M=\sum_{x=1}^r K_x\otimes \bar{L}_x \ \ \text{and} \ \ \mathcal M:X\mapsto\sum_{x=1}^r K_x X L_x^* \, . \]
What is important for us is that this identification preserves the spectrum.

\begin{theorem}[Minimax principle for singular values \cite{Bha}] \label{th:Bha1}
	Let $M$ be an $n\times n$ complex matrix. Then, for any $1\leq i \leq n$,
	\[ s_i(M) = \min_{P\in\mathcal{P}_i} \|MP\|_{\infty} \, , \]
	where $\mathcal{P}_i$ denotes the set of rank $n-i+1$ projectors on $\C^n$.
\end{theorem}

As an immediate consequence of Theorem \ref{th:Bha1} we see that, for any unit vector $\ket{\varphi}\in\C^n$, 
\[ s_2(M) \leq \left\|M\left(\Id-\ketbra{\varphi}{\varphi}\right)\right\|_{\infty} \, . \]
Indeed, the minimax principle applied to the case $i=2$ tells us that $s_2(M)$ is equal to the infimum over unit vectors $\ket{\phi}\in\C^n$ of $\left\|M\left(\Id-\ketbra{\phi}{\phi}\right)\right\|_{\infty}$.

\begin{theorem}[Weyl's majorant theorem \cite{Bha}] \label{th:Bha2}
	Let $M$ be an $n\times n$ complex matrix. Then, for any $1\leq k \leq n$,
	\[ \sum_{i=1}^k |\lambda_i(M)| \leq \sum_{i=1}^k s_i(M) \, . \]
\end{theorem}

In particular, applying Theorem \ref{th:Bha2} to the case $k=2$, we get
\[ |\lambda_1(M)| + |\lambda_2(M)| \leq s_1(M) + s_2(M) \, . \]
	
\begin{theorem}[Perron-Froebenius theorem for irreducible positive maps \cite{EHK}] \label{th:PF}
Let $\mathcal M$ be an irreducible positive map on $n\times n$ matrices. Then,
\[ |\lambda_1(\mathcal M)| = \underset{X\geq 0}{\sup} \sup \{ \lambda\in\R \st \mathcal M(X) \geq \lambda X \} \, .\]		
\end{theorem}

As a consequence of Theorem \ref{th:PF} we have that, if there exists $X\geq 0$ such that $\mathcal{M}(X) \geq \lambda X$ for some $\lambda\in\R$, then $|\lambda_1(\mathcal M)|\geq \lambda$.

\begin{lemma} \label{lem:max-sv}
	Let $M$ be an $n\times n$ complex matrix satisfying the following: there exists a unit vector $\ket{\varphi}\in\C^n$ such that $|\bra{\varphi}M\ket{\varphi}| \leq \lambda$ and $\left\|M\left(\Id-\ketbra{\varphi}{\varphi}\right)\right\|_{\infty},\left\|\left(\Id-\ketbra{\varphi}{\varphi}\right)M\right\|_{\infty}\leq \lambda'$ for some $\lambda>\lambda'>0$. Then, $s_1(M)\leq \lambda+\lambda'$.
\end{lemma}

\begin{proof}
	Let $\ket{\phi_1},\ket{\phi_2}\in\C^n$ be unit vectors, which we write as $\ket{\phi_i}=\alpha_i\ket{\varphi}+\beta_i\ket{\varphi'}$, where $\ket{\varphi'}\in\C^n$ is a unit vector orthogonal to $\ket{\varphi}$ and $\alpha_i,\beta_i\in\C$ are such that $|\alpha_i|^2+|\beta_i|^2=1$, for $i=1,2$. 
	We then have
	\[ \bra{\phi_1}M\ket{\phi_2} = \bar{\alpha}_1\alpha_2\bra{\varphi}M\ket{\varphi} + \bar{\alpha}_1\beta_2\bra{\varphi}M\ket{\varphi'} + \bar{\beta}_1\alpha_2\bra{\varphi'}M\ket{\varphi} + \bar{\beta}_1\beta_2\bra{\varphi'}M\ket{\varphi'} \, . \]
	Now, by assumption, $|\bra{\varphi}M\ket{\varphi}| \leq \lambda$, while $\left| \bra{\varphi'}M\ket{\varphi'} \right|, \left| \bra{\varphi}M\ket{\varphi'} \right|, \left| \bra{\varphi'}M\ket{\varphi} \right| \leq \lambda'$. Also, 
	\[ \left|\bar{\alpha}_1\alpha_2\right| + \left|\bar{\alpha}_1\beta_2\right| + \left|\bar{\beta}_1\alpha_2\right| + \left|\bar{\beta}_1\beta_2\right| = \left( \left|\alpha_1\right| + \left|\beta_1\right| \right) \left( \left|\alpha_2\right| + \left|\beta_2\right| \right) \leq 2 \left( \left|\alpha_1\right|^2 + \left|\beta_1\right|^2 \right)^{1/2} \left( \left|\alpha_2\right|^2 + \left|\beta_2\right|^2 \right)^{1/2} = 2 \, . \]
	Hence, $\left|\bar{\alpha}_1\beta_2\right| + \left|\bar{\beta}_1\alpha_2\right| + \left|\bar{\beta}_1\beta_2\right| \leq 2-\left|\bar{\alpha}_1\alpha_2\right|$. And we therefore get, by the triangle inequality and the fact that $\left|\bar{\alpha}_1\alpha_2\right|\leq 1$,
	\[ \left| \bra{\phi_1}M\ket{\phi_2} \right| \leq \left|\bar{\alpha}_1\alpha_2\right|\lambda+ \left( \left|\bar{\alpha}_1\beta_2\right| + \left|\bar{\beta}_1\alpha_2\right| + \left|\bar{\beta}_1\beta_2\right| \right)\lambda' \leq \left|\bar{\alpha}_1\alpha_2\right|(\lambda-\lambda')+2\lambda' \leq \lambda+\lambda' \, . \]
	Since the latter upper bound holds for any unit vectors $\ket{\phi_1},\ket{\phi_2}\in\C^n$, it indeed proves that $s_1(M)=\|M\|_{\infty}\leq \lambda+\lambda'$.
\end{proof}

\begin{proposition} \label{prop:spectral-gap}
	Let $M$ be an $n^2\times n^2$ complex matrix and $\mathcal M$ be its corresponding map on $n\times n$ complex matrices. Assume that the following holds: (i) $\mathcal M$ is positive irreducible and there exists a positive semidefinite matrix $X$ on $\C^n$ such that $\mathcal M(X)\geq (1-\delta) X$, (ii) there exists a unit vector $\ket{\varphi}$ in $\C^{n^2}$ such that $|\bra{\varphi}M\ket{\varphi}| \leq 1+\epsilon$ and $\left\|M\left(\Id-\ketbra{\varphi}{\varphi}\right)\right\|_{\infty},\left\|\left(\Id-\ketbra{\varphi}{\varphi}\right)M\right\|_{\infty}\leq \eta$, where $0<\delta,\epsilon,\eta<1/5$. Then,
	\[ |\lambda_1(M)|\geq 1-\delta\ \ \text{and}\ \ |\lambda_2(M)|\leq \delta+ \epsilon+2\eta \, , \]	
	so that in particular
	\[ \Delta(M)\geq 1-2\delta-\epsilon-2\eta\, . \]
\end{proposition}

\begin{proof}
	To begin with, condition (i) implies, by Theorem \ref{th:PF}, that
	\[ |\lambda_1(M)| = |\lambda_1(\mathcal M)| \geq 1 -\delta \, . \] 
	Next, Theorem \ref{th:Bha2} tells us that
	\[ |\lambda_2(M)| \leq s_1(M) + s_2(M) - |\lambda_1(M)| \, .  \]
	Now, condition (ii) implies, first of all by Lemma \ref{lem:max-sv} that $s_1(M)\leq 1+\epsilon+\eta$, and second of all by Theorem \ref{th:Bha1} $s_2(M)\leq \left\|M\left(\Id-\ketbra{\varphi}{\varphi}\right)\right\|_{\infty}\leq \eta$. Hence,
	\[ |\lambda_2(M)| \leq (1+\epsilon+\eta)+\eta-(1-\delta) =\delta+ \epsilon+2\eta \, .\]
	And the proof is thus complete.
\end{proof}

With the result of Proposition \ref{prop:spectral-gap} in mind, we can now explain what is our strategy in order to show that the random MPS transfer operator $T$, as defined by equation \eqref{eq:transfer-MPS}, typically has a spectral gap $\Delta(T)>0$. We know by Fact \ref{fact:irreducible} that, with probability $1$, the positive map $\mathcal T$ corresponding to $T$ is irreducible. Our goal is thus to find a positive semidefinite matrix $X$ on $\C^D$ and a unit vector $\ket{\varphi}$ in  $\C^D\otimes\C^D$ such that, with high probability 
\begin{equation} \label{eq:ineq-Delta-1} 
	\mathcal T(X) \geq (1-\delta)X \, ,
\end{equation} 
\begin{equation} \label{eq:ineq-Delta-2}
|\bra{\varphi}T\ket{\varphi}| \leq 1+\epsilon \ \ \text{and} \ \  \|T(\Id-\ketbra{\varphi}{\varphi})\|_{\infty},\|(\Id-\ketbra{\varphi}{\varphi})T\|_{\infty} \leq \eta \, , 
\end{equation} 
for some $0<\delta,\epsilon,\eta<1/5$. Indeed, we know by Proposition \ref{prop:spectral-gap} that if equations \eqref{eq:ineq-Delta-1} and \eqref{eq:ineq-Delta-2} hold it guarantees that, with high probability
\[ \Delta(T) = |\lambda_1(T)| - |\lambda_2(T)| \geq 1-2\delta-\epsilon-2\eta >0 \, . \]
We proceed similarly for the random PEPS transfer operator $T_N$, as defined by equation \eqref{eq:transfer-PEPS}, with a positive semidefinite matrix $X_N$ on $(\C^D)^{\otimes N}$ and a unit vector $\ket{\varphi_N}$ in $(\C^D\otimes\C^D)^{\otimes N}$.

Before proceeding, we shall make one last simple observation on the spectrum of the random MPS transfer operator $T$, which straightforwardly follows from noticing that $T$ and $\bar{T}$ have the same spectrum. The latter claim is in turn a consequence of the fact that $\bar{T}=FTF^*$, where $F$ denotes the flip unitary on $\C^D\otimes\C^D$ (which is defined by $F\ket{\alpha\beta}=\ket{\beta\alpha}$, for any $1\leq \alpha,\beta \leq D$).

\begin{fact} \label{fact:real-trace}
	Let $T$ be defined as in equation \eqref{eq:transfer-MPS}. If $\lambda\in\mathrm{spec}(T)$ then $\bar{\lambda}\in\mathrm{spec}(T)$. And therefore, for any $n\in\N$, $\Tr(T^n)\in\R$.
\end{fact}

\subsection{The case of MPS} \hfill\par\smallskip
\label{sec:transfer-operator-MPS}

Our candidate semidefinite positive matrix on $\C^D$ satisfying with high probability equation \eqref{eq:ineq-Delta-1} will be the identity matrix $\Id$, while our candidate unit vector in $\C^D\otimes\C^D$ satisfying with high probability equation \eqref{eq:ineq-Delta-2} will be the maximally entangled unit vector $\ket{\psi}$. Before launching into proofs, let us briefly explain what is the intuition behind such choice. First, it is easy to check (cf.~subsequent computations) that $\E T=\ketbra{\psi}{\psi}$. It is thus natural to expect that the largest eigenvalue of $T$ should be close to $1$ and that the corresponding eigenvector should be close to $\ket{\psi}$. Second, we know from observations in Section \ref{sec:parent-hamiltonian-MPS} that $T=\mathcal{R}(W)/Dd$ for $W$ a $D^2\times D^2$ Wishart matrix with parameter $d$. And it was proved in \cite{AN} that the singular value distribution of $\sqrt{d}(\mathcal{R}(W)/Dd-\ketbra{\psi}{\psi})$, i.e.~of $\sqrt{d}(T-\ketbra{\psi}{\psi})$, converges in moments to the quarter-circle distribution. This means that the singular values of $T-\ketbra{\psi}{\psi}$ are at least almost all of order at most $1/\sqrt{d}$. However, this result does not tell us anything about potential isolated singular values (or in fact eigenvalues), which is what would truly matter for our purposes. What is more, even the statement about so-called weak convergence of the singular value distribution of $\sqrt{d}(T-\ketbra{\psi}{\psi})$ was proved only in the regime where $d$ is of order $D^2$. While, as we will later see, our results are valid for any respective scaling of $d$ and $D$.

\subsubsection{Computing the typical value of the transfer CP map on the identity} \hfill\par\smallskip

\begin{proposition} \label{prop:MPS-CP}
Let $T$ be defined as in equation \eqref{eq:transfer-MPS} and let $\mathcal T$ be its corresponding CP map. Then,
\[ \P \left( \left\| \mathcal T(\Id) - \Id \right\|_{\infty} \leq \frac{6}{\sqrt{d}} \right) \geq 1-2e^{-D/4} \, , \]
which implies that
\[ \P \left( \mathcal T(\Id) \geq \left(1-\frac{6}{\sqrt{d}}\right)\Id \right) \geq 1-2e^{-D/4} \, . \]
\end{proposition}

\begin{proof}
Recall that
\[ \mathcal T(\Id) = \frac{1}{d} \sum_{x=1}^d G_xG_x^* \, , \]
where the $G_x$'s are independent $D\times D$ matrices whose entries are independent complex Gaussians with mean $0$ and variance $1/D$. This means that $\mathcal T(\Id)$ is distributed as $GG^*/dD$, for $G$ a $D\times dD$ matrix whose entries are independent complex Gaussians with mean $0$ and variance $1$, i.e.~for $GG^*$ a $D\times D$ Wishart matrix of parameter $dD$. Hence, by Theorem \ref{th:Wishart} (applied with $n=D$ and $s=dD$), we know that
\[ \P \left( \left\| \mathcal T(\Id) - \Id \right\|_{\infty} > \frac{6}{\sqrt{d}} \right) \leq 2e^{-D/4} \, , \]
as claimed.
\end{proof}

\subsubsection{Computing the expected overlap of the transfer operator with the maximally entangled state} \hfill\par\smallskip

\begin{proposition} \label{prop:MPS-1}
Let $T$ be defined as in equation \eqref{eq:transfer-MPS}. Then, $\bra{\psi}T\ket{\psi}\in\R$ and
\[ \E\bra{\psi}T\ket{\psi} = 1 \, . \]
\end{proposition}

\begin{proof}
The claimed result easily follows from a direct computation. Indeed,
\begin{align*} 
\bra{\psi}T\ket{\psi}  & = \frac{1}{dD} \sum_{x=1}^d \sum_{\alpha,\beta=1}^{D} \bra{\alpha}G_x\ket{\beta} \bra{\alpha}\bar{G}_x\ket{\beta} \\
& = \frac{1}{dD} \sum_{x=1}^d \sum_{\alpha,\beta=1}^{D} |\bra{\alpha}G_x\ket{\beta}|^2  \\
& = \frac{1}{dD} \sum_{x=1}^d \Tr(G_xG_x^*) \, .
\end{align*}
So first it is clear that $\bra{\psi}T\ket{\psi}\in\R$. And second, 
\[ \E \bra{\psi}T\ket{\psi} = \frac{1}{dD} \sum_{x=1}^d \E \Tr(G_xG_x^*) = 1 \, , \]
where the last equality is because, for each $1\leq x\leq d$, $\E \Tr(G_xG_x^*)=D$.
\end{proof}

\subsubsection{Upper bounding the expected norm of the projection of the transfer operator on the orthogonal of the maximally entangled state} \hfill\par\smallskip

\begin{lemma} \label{lem:moments-G}
Let $G$ be a $D\times D$ matrix whose entries are independent complex Gaussians with mean $0$ and variance $1/D$. Then, for any even $p\in\N$ such that $(2D^2)^{1/5}\leq p/2\leq D^{2/3}$, we have
\[ \E\Tr|G|^p \leq 2^p\times\frac{p^5}{128D} \, . \]
\end{lemma}

\begin{proof}
Given $q\in\N$, we denote by $\mathcal{S}_q$ the set of permutations of $\{1,\ldots,q\}$, by $\gamma\in\mathcal{S}_q$ the full cycle $(1\cdots q)$ and, for any $\pi\in\mathcal{S}_q$, by $\sharp(\pi)$ the number of cycles in the cycle decomposition of $\pi$. Then, it is well-known that we can write
\[ \E\Tr|G|^{2q} = D\sum_{\delta=0}^{\lfloor q/2 \rfloor} S(q,\delta)D^{-2\delta} \, , \]
where $S(q,\delta)=|\{ \pi\in\mathcal{S}_q : \sharp(\gamma\pi^{-1})+\sharp(\pi)=q+1-2\delta \}|$ (see e.g.~\cite[Appendix B.2]{Lan} for further details). Now, we know from \cite[Lemma 12]{Mon} that $S(q,0)\leq 4^{q-1}$ and, for each $1\leq \delta\leq \lfloor q/2 \rfloor$, $S(q,\delta)\leq 4^{q-1}q^{3\delta+1}$. Hence,
\[ \E\Tr|G|^{2q} \leq 4^{q-1}D\left(1+q\sum_{\delta=1}^{\lfloor q/2 \rfloor} \left(\frac{q^3}{D^2}\right)^{\delta}\right) \, . \]
Consequently, if $(2D^2)^{1/5}\leq q\leq D^{2/3}$, we have
\[ \E\Tr|G|^{2q} \leq 4^{q-1}D\left(1+q\times\frac{q}{2}\times\frac{q^3}{D^2}\right) \leq 4^q\times\frac{q^5}{4D} \, , \]
where the first inequality is because $q^3/D^2\leq 1$ and the second inequality is because $1\leq q^5/2D^2$. And the advertised result follows, simply replacing $q$ by $p/2$.
\end{proof}

\begin{lemma} \label{lem:norm-sym}
Let $G_1,\ldots,G_d,H_1,\ldots,H_d$ be independent $D\times D$ matrices whose entries are independent complex Gaussians with mean $0$ and variance $1/D$. Then,
\[ \E \left\| \sum_{x=1}^d \left( G_x\otimes\bar{G}_x - H_x\otimes\bar{H}_x \right) \right\|_{\infty} \leq 20\sqrt{d} \, .  \]
\end{lemma}

\begin{proof}
The reasoning is directly inspired from the one in the proofs of \cite[Lemma 4.1 and Theorem 4.2]{Pis2} and \cite[Theorem 16.6]{Pis1}.
First observe that
\begin{align*}
\E \left\| \sum_{x=1}^d \left( G_{x}\otimes\bar{G}_{x} - H_{x}\otimes\bar{H}_{x} \right) \right\|_{\infty} & = \E \left\| \sum_{x=1}^d \left( \frac{G_x+H_x}{\sqrt{2}}\otimes \frac{\bar{G}_x+\bar{H}_x}{\sqrt{2}} - \frac{G_x-H_x}{\sqrt{2}}\otimes \frac{\bar{G}_x-\bar{H}_x}{\sqrt{2}} \right) \right\|_{\infty} \\
& = \E \left\| \sum_{x=1}^d \left( G_{x}\otimes\bar{H}_{x} + H_{x}\otimes\bar{G}_{x} \right) \right\|_{\infty} \\
& \leq 2 \E \left\| \sum_{x=1}^d G_{x}\otimes\bar{H}_{x} \right\|_{\infty}
\end{align*}
Next, for any $p\in\N$, we know that $\|\cdot\|_{\infty} \leq \|\cdot\|_p$, so that by Jensen inequality
\[ \E \left\| \sum_{x=1}^d G_{x}\otimes\bar{H}_{x} \right\|_{\infty} \leq \left(\E \left\| \sum_{x=1}^d G_{x}\otimes\bar{H}_{x} \right\|^p_p \right)^{1/p} \, . \]
Now, for any even $p\in\N$, writing $p=2q$, we have
\begin{align*}
\E \left\| \sum_{x=1}^d G_{x}\otimes\bar{H}_{x} \right\|^p_p & = \E \Tr \left| \left(\sum_{x=1}^d G_{x}\otimes\bar{H}_{x}\right) \left(\sum_{x=1}^d G_{x}^*\otimes\bar{H}_{x}^*\right) \right|^q \\
& = \E \left( \sum_{x_1,\ldots,x_q,y_1,\ldots,y_q=1}^d \Tr \left(G_{x_1}G^*_{y_1}\cdots G_{x_q}G^*_{y_q}\right) \Tr\left(\bar{H}_{x_1}\bar{H}^*_{y_1}\cdots \bar{H}_{x_q}\bar{H}^*_{y_q}\right) \right) \\
& = \sum_{x_1,\ldots,x_q,y_1,\ldots,y_q=1}^d \left( \E \Tr\left(G_{x_1}G^*_{y_1}\cdots G_{x_q}G^*_{y_q}\right) \right)^2 \, . \\
\end{align*}
Yet, for each $1\leq x_1,\ldots,x_q,y_1,\ldots,y_q\leq d$, we know by H\"{o}lder inequality that
\[ \left| \Tr\left(G_{x_1}G^*_{y_1}\cdots G_{x_q}G^*_{y_q}\right) \right| \leq \left(\Tr\left|G_{x_1}\right|^p\right)^{1/p}\left(\Tr\left|G_{y_1}\right|^p\right)^{1/p} \cdots \left(\Tr\left|G_{x_q}\right|^p\right)^{1/p}\left(\Tr\left|G_{y_q}\right|^p\right)^{1/p} \, , \] 
which implies, since $\E|X_1\cdots X_p| \leq \E|X|^p$ for identically distributed random variables $X,X_1,\ldots,X_p$, that
\begin{align*}
\left| \E \Tr\left(G_{x_1}G^*_{y_1}\cdots G_{x_q}G^*_{y_q}\right) \right| & \leq \E \left| \Tr\left(G_{x_1}G^*_{y_1}\cdots G_{x_q}G^*_{y_q}\right) \right| \\
& \leq \E \left( \left(\Tr\left|G_{x_1}\right|^p\right)^{1/p}\left(\Tr\left|G_{y_1}\right|^p\right)^{1/p} \cdots \left(\Tr\left|G_{x_q}\right|^p\right)^{1/p}\left(\Tr\left|G_{y_q}\right|^p\right)^{1/p} \right) \\
& \leq \E \Tr|G|^p \, .
\end{align*}
We thus have shown that
\begin{align*} 
\E \left\| \sum_{x=1}^d G_{x}\otimes\bar{H}_{x} \right\|^p_p & \leq \E \Tr|G|^p \sum_{x_1,\ldots,x_q,y_1,\ldots,y_q=1}^d \left| \E \Tr\left(G_{x_1}G^*_{y_1}\cdots G_{x_q}G^*_{y_q}\right) \right| \\ 
& = \E \Tr|G|^p \sum_{x_1,\ldots,x_q,y_1,\ldots,y_q=1}^d \E \Tr\left(G_{x_1}G^*_{y_1}\cdots G_{x_q}G^*_{y_q}\right) \\
& = \E \Tr|G|^p \E \Tr\left| \sum_{x=1}^d G_x \right|^p \\
& = (\E \Tr|G|^p)^2\sqrt{d}^p \, , 
\end{align*}
where the first equality is because the only non vanishing terms in the sum are non negative and the last equality is because $\sum_{x=1}^d G_x \sim \sqrt{d}G$, where $G$ is a $D\times D$ matrix whose entries are independent complex Gaussians with mean $0$ and variance $1/D$. Now, we know from Lemma \ref{lem:moments-G} that, for any even $p\in\N$ such that $(2D^2)^{1/5}\leq p/2 \leq D^{2/3}$,
\[ \E \Tr|G|^p \leq 2^p\times\frac{p^5}{128D} \, . \]
Hence putting everything together, we finally get that, for such $p\in\N$,
\begin{align*}
\E \left\| \sum_{x=1}^d \left( G_{x}\otimes\bar{G}_{x} - H_{x}\otimes\bar{H}_{x} \right) \right\|_{\infty} & \leq 2 \E \left\| \sum_{x=1}^d G_{x}\otimes\bar{H}_{x} \right\|_{\infty} \\
& \leq 2\left(\E \left\| \sum_{x=1}^d G_{x}\otimes\bar{H}_{x} \right\|^p_p \right)^{1/p} \\
& \leq 2 (\E \Tr|G|^p)^{2/p}\sqrt{d} \\
& \leq 8 \left(\frac{p^5}{128D}\right)^{2/p} \sqrt{d} \, .
\end{align*}
Choosing $p=2\lfloor D^{2/3} \rfloor$ in the above inequality, we see that
\[ \E \left\| \sum_{x=1}^d \left( G_{x}\otimes\bar{G}_{x} - H_{x}\otimes\bar{H}_{x} \right) \right\|_{\infty} \leq 8 \left(\frac{D^{7/3}}{4}\right)^{1/D^{2/3}} \sqrt{d} \, . \]
Since $\left(D^{7/3}/4\right)^{1/D^{2/3}}\leq 5/2$ for all $D\in\N$, the claimed result follows. 
\end{proof}

\begin{proposition} \label{prop:MPS-2}
Let $T$ be defined as in equation \eqref{eq:transfer-MPS}. Then,
\[ \E \left\|T(\Id-\ketbra{\psi}{\psi})\right\|_{\infty} \leq \frac{40}{\sqrt{d}} \, . \]
\end{proposition}

\begin{proof}
To begin with, note that $\left\|T(\Id-\ketbra{\psi}{\psi})\right\|_{\infty} \leq 2 \|T-\ketbra{\psi}{\psi}\|_{\infty}$. Indeed, by the triangle inequality
\[ \left\|T(\Id-\ketbra{\psi}{\psi})\right\|_{\infty} = \|T-T\ketbra{\psi}{\psi}\|_{\infty} \leq \|T-\ketbra{\psi}{\psi}\|_{\infty} + \|T\ketbra{\psi}{\psi}-\ketbra{\psi}{\psi}\|_{\infty} \, , \] 
and by H\"{o}lder inequality
\[ \|T\ketbra{\psi}{\psi}-\ketbra{\psi}{\psi}\|_{\infty} = \|(T-\ketbra{\psi}{\psi})\ketbra{\psi}{\psi}\|_{\infty} \leq \|T-\ketbra{\psi}{\psi}\|_{\infty}\|\ketbra{\psi}{\psi}\|_{\infty} = \|T-\ketbra{\psi}{\psi}\|_{\infty} \, . \]

Next, observe that $\ketbra{\psi}{\psi}=\E T$, so that we can re-write
\[ \left\| T-\ketbra{\psi}{\psi} \right\|_{\infty} = \left\| \frac{1}{d} \sum_{x=1}^d \left(G_x \otimes \bar{G}_x - \E H_x \otimes \bar{H}_x \right) \right\|_{\infty} \, , \]
where the $H_x$'s are independent copies of the $G_x$'s. Now, by Jensen inequality
\[  \E \left\| \frac{1}{d} \sum_{x=1}^d \left(G_x \otimes \bar{G}_x - \E H_x \otimes \bar{H}_x \right) \right\|_{\infty} \leq \E \left\| \frac{1}{d} \sum_{x=1}^d (G_x \otimes \bar{G}_x - H_x \otimes \bar{H}_x) \right\|_{\infty} \, . \]
Yet, we know from Lemma \ref{lem:norm-sym} that
\[ \E \left\| \frac{1}{d} \sum_{x=1}^d (G_x \otimes \bar{G}_x - H_x \otimes \bar{H}_x) \right\|_{\infty} \leq \frac{20}{\sqrt{d}} \, , \]

Putting everything together, we thus get
\[ \E \left\|T(\Id-\ketbra{\psi}{\psi})\right\|_{\infty} \leq 2\times\frac{20}{\sqrt{d}}= \frac{40}{\sqrt{d}} \, , \]
which is exactly the announced result.
\end{proof}

\subsubsection{Typical spectral gap of the transfer operator} \hfill\par\smallskip

\begin{lemma} \label{lem:dev-ineq}
	Let $P$ be a projector on $\C^D\otimes\C^D$, and define the function $\hat{f}$, of $d$-uples of $D\times D$ matrices, as
	\[ \hat{f}:(A_1,\ldots,A_d) \mapsto \left\|\left(\frac{1}{d}\sum_{x=1}^dA_x\otimes\bar{A}_x\right)P\right\|_{\infty} \, . \]
	Then, for $G_1,\ldots,G_d$ independent $D\times D$ matrices whose entries are independent complex Gaussians with mean $0$ and variance $1/D$, we have
	\[ \forall\ \epsilon>0,\ \P\left( \hat{f}(G_1,\ldots,G_d) \gtrless \E \hat{f} \pm \epsilon \right) \leq e^{-Dd\epsilon^2/72}+e^{-Dd} \, . \]
\end{lemma}

\begin{proof}
	Define the following subset of the set of $d$-uples of $D\times D$ matrices:
	\[ \hat{\Omega} := \left\{ (A_1,\ldots,A_d) : \left(\sum_{x=1}^d \|A_x\|_{\infty}^2 \right)^{1/2} \leq 3\sqrt{d} \right\} \, . \]
	We will first show that, for $G_1,\ldots,G_d$ independent $D\times D$ matrices whose entries are independent complex Gaussians with mean $0$ and variance $1/D$, we have	
	\[ \P\left( (G_1,\ldots,G_d)\notin\hat{\Omega} \right) \leq e^{-Dd}\, . \]
	For this, we will use the Gaussian concentration inequality, recalled in Theorem \ref{th:g-global}. Let us start with showing that the average of the function we are interested in is upper bounded by $2\sqrt{d}$. Indeed, by Jensen inequality
	\[ \E \left(\sum_{x=1}^d \|G_x\|_{\infty}^2 \right)^{1/2} \leq \left( \sum_{x=1}^d \E\|G_x\|_{\infty}^2 \right)^{1/2} \, . \]
	Yet, for $G$ a $D\times D$ matrix whose entries are independent complex Gaussians with mean $0$ and variance $1/D$, it is well known that $\E\|G\|_{\infty}^2\leq 4$. And therefore,
	\[ \E \left(\sum_{x=1}^d \|G_x\|_{\infty}^2 \right)^{1/2} \leq 2\sqrt{d} \, .\]
	Let us now turn to showing that the Lipschitz constant of the function we are interested in is upper bounded by $1$. Indeed, by the triangle inequality (twice)
	\begin{align*}
	\left| \left(\sum_{x=1}^d \|G_x\|_{\infty}^2 \right)^{1/2} - \left(\sum_{x=1}^d \|G_x'\|_{\infty}^2 \right)^{1/2} \right| & \leq \left( \sum_{x=1}^d \left(\|G_x\|_{\infty}- \|G_x'\|_{\infty}\right)^2 \right)^{1/2} \\
	& \leq \left( \sum_{x=1}^d \|G_x-G_x'\|_{\infty}^2 \right)^{1/2} \\
	& \leq \left( \sum_{x=1}^d \|G_x-G_x'\|_2^2 \right)^{1/2} \, .
	\end{align*}
	With these two estimates at hand, we can conclude that
	\[ \forall\ \epsilon>0,\ \P\left( \left(\sum_{x=1}^d \|G_x\|_{\infty}^2 \right)^{1/2} > 2\sqrt{d} +\epsilon \right) \leq e^{-D\epsilon^2} \, , \]
	which, taking $\epsilon=\sqrt{d}$, is exactly what we wanted to prove.
	
	We will now make us of the local version of the Gaussian concentration inequality, recalled in Theorem \ref{th:g-local}. In the case of our function $\hat{f}$ and our subset $\hat{\Omega}$, we have that, if $(G_1,\ldots,G_d),(G_1',\ldots,G_d')\in\hat{\Omega}$, then
	\begin{align*}
	\left|\hat{f}(G_1,\ldots,G_d)-\hat{f}(G_1',\ldots,G_d')\right| & = \left| \left\|\left(\frac{1}{d}\sum_{x=1}^dG_x\otimes\bar{G}_x\right)P\right\|_{\infty} - \left\|\left(\frac{1}{d}\sum_{x=1}^dG_x'\otimes\bar{G}_x'\right)P\right\|_{\infty} \right| \\
	& \leq \left\|\left(\frac{1}{d}\sum_{x=1}^d(G_x\otimes\bar{G}_x-G_x'\otimes\bar{G}_x')\right)P\right\|_{\infty} \\
	& \leq \left\|\frac{1}{d}\sum_{x=1}^d(G_x\otimes\bar{G}_x-G_x'\otimes\bar{G}_x')\right\|_{\infty} \\
	& \leq \frac{1}{d}\sum_{x=1}^d\left(\|G_x\|_{\infty}+\|G_x'\|_{\infty}\right)\|G_x-G_x'\|_{\infty} \\
	& \leq \frac{1}{d}\left(\sum_{x=1}^d\left(\|G_x\|_{\infty}+\|G_x'\|_{\infty}\right)^2\right)^{1/2} \left(\sum_{x=1}^d\|G_x-G_x'\|_{\infty}^2\right)^{1/2} \\
	& \leq \frac{\sqrt{2}}{d} \left( \left(\sum_{x=1}^d\|G_x\|_{\infty}^2\right)^{1/2} + \left(\sum_{x=1}^d\|G_x'\|_{\infty}^2\right)^{1/2} \right) \left(\sum_{x=1}^d\|G_x-G_x'\|_{\infty}^2\right)^{1/2} \\
	& \leq \frac{6\sqrt{2}}{\sqrt{d}}\left(\sum_{x=1}^d\|G_x-G_x'\|_{\infty}^2\right)^{1/2} \\
	& \leq \frac{6\sqrt{2}}{\sqrt{d}}\left(\sum_{x=1}^d\|G_x-G_x'\|_2^2\right)^{1/2} \, ,
	\end{align*}
	where the first inequality is by the triangle inequality, the third inequality is also by the triangle inequality (after noticing that $G_x\otimes\bar{G}_x-G_x'\otimes\bar{G}_x'=G_x\otimes(\bar{G}_x-\bar{G}_x')+(G_x-G_x')\otimes\bar{G}_x'$), the fourth inequality is by Cauchy-Schwarz inequality, the fifth inequality is because $(a+b)^2\leq 2(a^2+b^2)$ and $\sqrt{a+b}\leq\sqrt{a}+\sqrt{b}$ for any $a,b\geq 0$, and the sixth inequality is by assumption on $\hat{\Omega}$. Putting together this upper bound on the Lipschitz constant of $\hat{f}$ on $\hat{\Omega}$ with the upper bound on the probability of the complement of $\hat{\Omega}$, we eventually get
	\[ \forall\ \epsilon>0,\ \P\left( \hat{f}(G_1,\ldots,G_d) \gtrless \E \hat{f} \pm \epsilon \right) \leq e^{-Dd\epsilon^2/72}+e^{-Dd} \, , \] 
	which is exactly the announced result.
\end{proof}

\begin{proposition} \label{prop:MPS-1'}
	Let $T$ be defined as in equation \eqref{eq:transfer-MPS}. Then,
	\[ \forall\ 0<\epsilon<\sqrt{d},\ \P \left( \left| \bra{\psi}T\ket{\psi} - 1 \right| \leq \frac{\epsilon}{\sqrt{d}} \right) \geq 1- 4e^{-D\epsilon^2/72} \, . \]	
\end{proposition}

\begin{proof}
	Observe first of all that $\bra{\psi}T\ket{\psi}=\|T\ketbra{\psi}{\psi}\|_{\infty}$. So we will apply Lemma \ref{lem:dev-ineq} to the case where $P=\ketbra{\psi}{\psi}$. 
	We know from Proposition \ref{prop:MPS-1} that
	\[ \E\left\| \left(\frac{1}{d}\sum_{x=1}^dG_x\otimes\bar{G}_x\right)\ketbra{\psi}{\psi} \right\|_{\infty} = 1 \, , \]
	so that by Lemma \ref{lem:dev-ineq} (with $\epsilon/\sqrt{d}$ playing the role of $\epsilon$)
	\[ \P\left( \left| \left\| \left(\frac{1}{d}\sum_{x=1}^dG_x\otimes\bar{G}_x\right)\ketbra{\psi}{\psi} \right\|_{\infty} - 1 \right| > \frac{\epsilon}{\sqrt{d}} \right) \leq 2\left(e^{-D\epsilon^2/72}+e^{-Dd}\right) \leq 4e^{-D\epsilon^2/72} \, , \]
	which is precisely what we wanted to show.
\end{proof}

\begin{proposition} \label{prop:MPS-2'}
Let $T$ be defined as in equation \eqref{eq:transfer-MPS}. Then,  
\begin{align*} 
	\forall\ 0<\epsilon<\sqrt{d}, & \ \P\left( \left\|T\left(\Id-\ketbra{\psi}{\psi}\right) \right\|_{\infty} \leq \frac{40+\epsilon}{\sqrt{d}} \right) \geq 1-2e^{-D\epsilon^2/72} \, , \\
	& \ \P\left( \left\|\left(\Id-\ketbra{\psi}{\psi}\right)T \right\|_{\infty} \leq \frac{40+\epsilon}{\sqrt{d}} \right) \geq 1-2e^{-D\epsilon^2/72} \, .
\end{align*}	
\end{proposition}

\begin{proof}
Let us start with the first deviation probably. We will apply Lemma \ref{lem:dev-ineq} to the case where $P=\Id-\ketbra{\psi}{\psi}$. 
We know from Proposition \ref{prop:MPS-2} that
\[ \E\left\| \left(\frac{1}{d}\sum_{x=1}^dG_x\otimes\bar{G}_x\right)\left(\Id-\ketbra{\psi}{\psi}\right) \right\|_{\infty} \leq \frac{40}{\sqrt{d}} \, , \]
so that by Lemma \ref{lem:dev-ineq} (with $\epsilon/\sqrt{d}$ playing the role of $\epsilon$)
\[ \P\left( \left\| \left(\frac{1}{d}\sum_{x=1}^dG_x\otimes\bar{G}_x\right)\left(\Id-\ketbra{\psi}{\psi}\right) \right\|_{\infty} > \frac{40+\epsilon}{\sqrt{d}} \right) \leq e^{-D\epsilon^2/72}+e^{-Dd} \leq 2e^{-D\epsilon^2/72} \, , \]
which is precisely what we wanted to show.

As for the second deviation probability, it follows from the first one applied to $T^*$ instead of $T$. Indeed, $T^*$ is distributed as $T$ so we know by what precedes that
\[ \P\left( \left\|T^*\left(\Id-\ketbra{\psi}{\psi}\right) \right\|_{\infty} > \frac{40+\epsilon}{\sqrt{d}} \right) \leq 2e^{-D\epsilon^2/72} \, . \]
Now, $(\Id-\ketbra{\psi}{\psi})T = \left(T^*(\Id-\ketbra{\psi}{\psi})\right)^*$, so that $\left\|\left(\Id-\ketbra{\psi}{\psi}\right)T \right\|_{\infty} = \left\|T^*\left(\Id-\ketbra{\psi}{\psi}\right) \right\|_{\infty}$, and the proof is thus complete.
\end{proof}

\begin{theorem} \label{th:gap-MPS}
Let $T$ be the random MPS transfer operator, as defined in equation \eqref{eq:transfer-MPS}. Then,  
\[ \P\left( \Delta(T) \geq 1-\frac{95}{\sqrt{d}} \right) \geq 1-10e^{-D/72} \, . \]
\end{theorem}

\begin{proof}
First, we know from Proposition \ref{prop:MPS-CP} that
\[ \P \left( \mathcal T(\Id) \ngeq \left(1-\frac{6}{\sqrt{d}}\right)\Id \right) \leq 2e^{-D/4} \, . \]
Second, we know from Proposition \ref{prop:MPS-1'} (with $\epsilon=1$) that 
\[ \P \left( \left| \bra{\psi}T\ket{\psi} -1 \right| > \frac{1}{\sqrt{d}} \right) \leq 4e^{-D/72} \, , \]
and from Proposition \ref{prop:MPS-2'} (with $\epsilon=1$) that
\[ \P\left( \left\|T\left(\Id-\ketbra{\psi}{\psi}\right) \right\|_{\infty} > \frac{41}{\sqrt{d}} \right) \leq 2e^{-D/72} \ \ \text{and} \ \ \P\left( \left\|\left(\Id-\ketbra{\psi}{\psi}\right)T \right\|_{\infty} > \frac{41}{\sqrt{d}} \right) \leq 2e^{-D/72} \, . \]
Now, we also know by Proposition \ref{prop:spectral-gap} that, if the three following conditions are satisfied 
\begin{align*} 
	& \mathcal T(\Id) \geq \left(1-\frac{6}{\sqrt{d}}\right)\Id\ \, , \\
	& \left| \bra{\psi}T\ket{\psi} -1 \right| \leq \frac{1}{\sqrt{d}} \, , \\  
	& \left\|T\left(\Id-\ketbra{\psi}{\psi}\right) \right\|_{\infty}, \left\|\left(\Id-\ketbra{\psi}{\psi}\right)T \right\|_{\infty} \leq \frac{41}{\sqrt{d}} \, ,
\end{align*} 
then $\Delta(T)\geq 1-95/\sqrt{d}$. So the four deviation probabilities above imply by the union bound that
\begin{align*}
\P\left( \Delta(T) < 1-\frac{95}{\sqrt{d}} \right)  
& \leq \P \left( \mathcal T(\Id) \ngeq \left(1-\frac{6}{\sqrt{d}}\right)\Id \right) + \P\left( \left| \bra{\psi}T\ket{\psi} -1 \right| > \frac{1}{\sqrt{d}} \right) \\
& + \P\left( \left\| T\left(\Id-\ketbra{\psi}{\psi}\right) \right\|_{\infty} > \frac{41}{\sqrt{d}} \right) + \P\left( \left\| \left(\Id-\ketbra{\psi}{\psi}\right)T \right\|_{\infty} > \frac{41}{\sqrt{d}} \right) \\
& \leq 10e^{-D/72} \, ,
\end{align*}
which is precisely what we wanted to show.
\end{proof}

\subsection{The case of PEPS} \hfill\par\smallskip
\label{sec:transfer-operator-PEPS}

\subsubsection{Computing the typical value of the transfer CP map on the identity} \hfill\par\smallskip

\begin{lemma} \label{lem:GH*}
	Let $G,H$ be two independent $D\times dD$ matrices whose entries are complex Gaussians with mean $0$ and variance $1$. Then,
	\[ \P\left( \left\|\frac{1}{dD}GH^*\right\|_\infty \leq \frac{16\sqrt{\ln 6}}{\sqrt{d}} \right) \geq 1-e^{-cD} \, , \]
	where $c>0$ is a universal constant.
\end{lemma}

\begin{proof}
	By definition, what we want to show is that, with probability larger than $1-e^{-cdD}$, the supremum over all unit vectors $\ket{\phi},\ket{\varphi}\in\C^D$ of $|\bra{\phi}GH^*\ket{\varphi}|$ is at most $16\sqrt{\ln 6}\sqrt{d}D$. With this goal in mind, let us first fix unit vectors $\ket{\phi},\ket{\varphi}\in\C^D$ and set $\ket{g_\phi}:=G^*\ket{\phi},\ket{h_\varphi}:=H^*\ket{\varphi}$, so that $\bra{\phi}GH^*\ket{\varphi}=\braket{g_\phi}{h_\varphi}$. Observe that $\ket{g_\phi},\ket{h_\varphi}\in\C^{dD}$ are independent Gaussian vectors with mean $0$ and variance $1$. Hence by Theorem \ref{th:g-global},
	\[ \forall\ \epsilon>0,\ \P\left( \|g_\phi\| > \sqrt{dD}(1+\epsilon) \text{ or } \|h_\varphi\| > \sqrt{dD}(1+\epsilon) \right) \leq 2e^{-\epsilon^2dD} \, . \]
	This implies that, defining the subset $\hat{\Omega}$ of the set of pairs of vectors in $\C^{dD}$ as 
	\[ \hat{\Omega}:=\left\{ (a,b) \st \|a\| \leq 2\sqrt{dD} \text{ and } \|b\| \leq 2\sqrt{dD} \right\} \, , \]
	we have
	\[ \P\left( (g_\phi,h_\varphi)\notin\hat{\Omega} \right) \leq 2e^{-dD} \, . \]
	Now, defining the function $\hat{f}$ of pairs of vectors in $\C^{dD}$ as
	\[ \hat{f}:(a,b)\mapsto |\braket{a}{b}| \, , \]
	we have that, if $(a,b),(a',b')\in\hat{\Omega}$, then
	\begin{align*} 
		\left| f(a,b)-f(a',b') \right| & = \left| |\braket{a}{b}| - |\braket{a'}{b'}| \right| \\
		& \leq \left| \braket{a}{b} - \braket{a'}{b'} \right| \\
		& = \left| \braket{a-a'}{b} + \braket{a'}{b-b'} \right| \\
		& \leq \|b\|\|a-a'\| + \|a'\|\|b-b'\| \\
		& \leq 2\sqrt{dD}\left(\|a-a'\|+\|b-b'\|\right) \\
		& \leq 2\sqrt{2}\sqrt{dD}\left(\|a-a'\|^2+\|b-b'\|^2\right)^{1/2} \, . 
	\end{align*}
	This means that $\hat{f}$ is $2\sqrt{2}\sqrt{dD}$-Lipschitz on $\hat{\Omega}$. Hence by Theorem \ref{th:g-local},
	\begin{equation} \label{eq:ip-individual} \forall\ \epsilon>0,\ \P\left( |\braket{g_\phi}{h_\varphi}| > \epsilon dD \right) \leq 2e^{-dD} + e^{-\epsilon^2dD/8} \, . 
	\end{equation}
	Now, fix $0<\delta<1$ and let $\mathcal N_\delta$ be a $\delta$-net in the unit sphere of $\C^D$ (i.e.~a subset of the unit sphere of $\C^D$ such that, for any unit vector $\ket{\phi}\in\C^D$ there exists $\ket{\phi'}\in\mathcal N_\delta$ such that $\|\phi-\phi'\|\leq \delta$). We know that $\mathcal N_\delta$ can be chosen such that $|\mathcal N_\delta|\leq (2/\delta)^{2D}$ (see e.g. \cite[Lemma 5.3]{AS}). By the union bound, we thus get from equation \eqref{eq:ip-individual} that
	\begin{equation} \label{eq:ip-net} \forall\ \epsilon>0,\ \P\left( \exists\ \ket{\phi},\ket{\varphi}\in\mathcal N_\delta :  |\bra{\phi}GH^*\ket{\varphi}| > \epsilon dD \right) \leq \left(\frac{2}{\delta}\right)^{4D} \left(2e^{-dD} + e^{-\epsilon^2dD/8} \right) \, . 
	\end{equation}
	Finally, define $M$, resp.~$M_\delta$, as the supremum over all unit vectors $\ket{\phi},\ket{\varphi}\in\C^D$, resp.~all $\ket{\phi},\ket{\varphi}\in\mathcal N_\delta$, of $|\bra{\phi}GH^*\ket{\varphi}|$. Then, given unit vectors $\ket{\phi},\ket{\varphi}\in\C^D$, letting $\ket{\phi'},\ket{\varphi'}\in\mathcal N_\delta$ be such that $\|\phi-\phi'\|,\|\varphi-\varphi'\|\leq\delta$, we have
	\[ |\bra{\phi}GH^*\ket{\varphi}| = | \bra{\phi'}GH^*\ket{\varphi'} + \bra{\phi'}GH^*\ket{\varphi-\varphi'} + \bra{\phi-\phi'}GH^*\ket{\varphi'} + \bra{\phi-\phi'}GH^*\ket{\varphi-\varphi'} | \leq M_\delta + (2\delta+\delta^2)M \, . \]
	Hence taking the supremum on the left hand side, we see that
	\[ M \leq \frac{1}{1+3\delta} M_\delta \, . \]
	Choosing $\delta=1/3$ in equation \eqref{eq:ip-net}, and recalling that by the inequality above $M\leq 2M_{1/3}$, we eventually get
	\[ \forall\ \epsilon>0,\ \P\left( \exists\ \ket{\phi},\ket{\varphi}\in\C^D :  |\bra{\phi}GH^*\ket{\varphi}| > 2\epsilon dD \right) \leq 6^{4D} \left(2e^{-dD} + e^{-\epsilon^2dD/8} \right) \, .  \]
	To conclude, we just have to observe that the right hand side is smaller than $3e^{-4\ln 6 D}$ for $\epsilon=8\sqrt{\ln 6}/\sqrt{d}$.
\end{proof}

\begin{proposition} \label{prop:PEPS-CP}
	Let $T_N$ be defined as in equation \eqref{eq:transfer-PEPS} and $\mathcal T_N$ be its corresponding CP map. Then, 
	\[ \P\left( \left\| \mathcal T_N(\Id) - \Id \right\|_{\infty} \leq \left(1+\frac{28D}{\sqrt{d}}\right)^N\frac{28}{\sqrt{d}} \right) \geq 1-(D+1)^{2N} (N+2)e^{-cD} \, , \]
	which implies that
	\[ \P\left( \mathcal T_N(\Id) \geq \left(1- \left(1+\frac{28D}{\sqrt{d}}\right)^N\frac{28}{\sqrt{d}} \right) \Id \right) \geq 1- (D+1)^{2N} (N+2)e^{-cD} \, , \]
	where $c>0$ is a universal constant.
\end{proposition}

\begin{proof}
	Recall that
	\[ \mathcal T_N(\Id) = \frac{1}{D^Nd^N} \sum_{a_1,b_1,\ldots,a_N,b_N=1}^D \sum_{x_1,\ldots,x_N=1}^d G_{a_Na_1x_1}G^*_{b_Nb_1x_1}\otimes \cdots\otimes G_{a_{N-1}a_Nx_N}G^*_{b_{N-1}b_Nx_N} \, , \]
	where the $G_{a_{i-1}a_ix_i}$'s are independent $D\times D$ matrices whose entries are independent complex Gaussians with mean $0$ and variance $1/D$. This means that $\mathcal T_N(\Id)$ is distributed as
	\begin{equation} \label{eq:T_N(Id)} 
		\frac{1}{D^{2N}d^N} \sum_{a_1,b_1,\ldots,a_N,b_N=1}^D G_{a_Na_1}G^*_{b_Nb_1}\otimes \cdots\otimes G_{a_{N-1}a_N}G^*_{b_{N-1}b_N} \, , 
	\end{equation}
	where the $G_{a_{i-1}a_i}$'s are independent $D\times dD$ matrices whose entries are independent complex Gaussians with mean $0$ and variance $1$.
	
	Let us first look at the term corresponding to $b_1=a_1,\ldots,b_N=a_N$ in expression \eqref{eq:T_N(Id)}, i.e.
	\[ W := \frac{1}{D^N} \sum_{a_1,\ldots,a_N=1}^D \left(\frac{1}{dD}G_{a_Na_1}G^*_{a_Na_1} \right) \otimes\cdots\otimes \left(\frac{1}{dD}G_{a_{N-1}a_N}G^*_{a_{N-1}a_N} \right) \, . \]
	For each $1\leq a,a'\leq D$, $G_{aa'}G_{aa'}^*$ is a $D\times D$ Wishart matrix of parameter $dD$. So by Theorem \ref{th:Wishart} (applied with $n=D$ and $s=dD$), we know that
	\[ \P \left( \left\| \frac{1}{dD}G_{aa'}G_{aa'}^* - \Id \right\|_{\infty} > \frac{6}{\sqrt{d}} \right) \leq 2e^{-D/4} \, .  \]
	Hence by the union bound,
	\begin{equation} \label{eq:deviation-1site} 
		\P \left( \exists\ 1\leq a,a'\leq D : \left\| \frac{1}{dD}G_{aa'}G_{aa'}^* - \Id \right\|_{\infty} > \frac{6}{\sqrt{d}} \right) \leq 2D^2e^{-D/4} \, . 
	\end{equation}
	Now, assume that $A_1,\ldots,A_N$ are $D\times D$ matrices satisfying $\|A_i-\Id\|_{\infty}\leq \epsilon$, $1\leq i\leq N$. Then,
	\begin{equation} \label{eq:Ntensors} 
		\left\|A_1\otimes\cdots\otimes A_N - \Id^{\otimes N}\right\|_{\infty} \leq \left(\sum_{i=1}^N (1+\epsilon)^{i-1}\right)\epsilon \leq (1+\epsilon)^N\epsilon \, . 
	\end{equation}
	The first inequality can easily be shown by induction, after noticing that
	\[ A_1\otimes\cdots\otimes A_N - \Id^{\otimes N} = \left(A_1\otimes\cdots\otimes A_{N-1}-\Id^{\otimes (N-1)}\right)\otimes A_N + \Id^{\otimes (N-1)}\otimes(A_N-I) \, , \]
	so that
	\begin{align*}
		\left\|A_1\otimes\cdots\otimes A_N - \Id^{\otimes N}\right\|_{\infty} & \leq \left\|A_1\otimes\cdots\otimes A_{N-1}-\Id^{\otimes (N-1)}\right\|_{\infty} \|A_N\|_{\infty} + \|\Id\|^{N-1}_{\infty}\|A_N-\Id\|_{\infty} \\
		& \leq \left\|A_1\otimes\cdots\otimes A_{N-1}-\Id^{\otimes (N-1)}\right\|_{\infty}(1+\epsilon) + \epsilon \, .
	\end{align*}
    Putting together equations \eqref{eq:deviation-1site} and \eqref{eq:Ntensors}, we thus get that the probability that 
    \[ \exists\ 1\leq a_1,\ldots,a_N\leq D : \left\| \left(\frac{1}{dD}G_{a_Na_1}G^*_{a_Na_1} \right) \otimes\cdots\otimes \left(\frac{1}{dD}G_{a_{N-1}a_N}G^*_{a_{N-1}a_N} \right) - \Id \right\|_{\infty} > \left(1+\frac{6}{\sqrt{d}}\right)^N\frac{6}{\sqrt{d}} \]
    is smaller than $2D ^2e^{-D/4}$.
    And consequently, just noticing that
    \[ W-\Id = \frac{1}{D^N} \sum_{a_1,\ldots,a_N=1}^D \left( \left(\frac{1}{dD}G_{a_Na_1}G^*_{a_Na_1} \right) \otimes\cdots\otimes \left(\frac{1}{dD}G_{a_{N-1}a_N}G^*_{a_{N-1}a_N} \right) -\Id \right) \, , \]
    we get by the triangle inequality that
    \[ \P\left(  \left\| W - \Id \right\|_{\infty} > \left(1+\frac{6}{\sqrt{d}}\right)^N\frac{6}{\sqrt{d}} \right) \leq 2D ^2e^{-D/4} \, . \]
    
    Let us now look at the other terms in expression \eqref{eq:T_N(Id)}. For each $1\leq a_1,\ldots,a_N\leq D$, define
    \[ W_{a_1,\ldots,a_N} := \sum_{q=1}^N \sum_{\substack{I\subset[N] \\ |I|=q}} \sum_{\substack{b_i\neq a_i,\,i\in I \\ b_i=a_i,\,i\notin I}} \left(\frac{1}{dD}G_{a_Na_1}G^*_{b_Nb_1} \right) \otimes\cdots\otimes \left(\frac{1}{dD}G_{a_{N-1}a_N}G^*_{b_{N-1}b_N} \right) \, , \]
    so that $W':=\mathcal T_N(\Id)-W$ can be written as
    \[ W'=\frac{1}{D^N} \sum_{a_1,\ldots,a_N=1}^D W_{a_1,\ldots,a_N} \, . \]
    Fix $1\leq a_1,\ldots,a_N\leq D$. By the triangle inequality, we have
    \[ \left\|W_{a_1,\ldots,a_N}\right\|_\infty \leq \sum_{q=1}^N \sum_{\substack{I\subset[N] \\ |I|=q}} \sum_{\substack{b_i\neq a_i,\,i\in I \\ b_i=a_i,\,i\notin I}} \prod_{i=1}^N \left\|\frac{1}{dD}G_{a_{i-1}a_i}G^*_{b_{i-1}b_i} \right\|_\infty \, .  \]
    Now, given $1\leq q\leq N$ and $I\subset[N]$ such that $|I|=q$, set $\bar{I}:=\cup_{i\in I}\{i,i+1\}$. We then have that, for $i\in\bar{I}$, $G_{b_{i-1}b_i}$ is independent from $G_{a_{i-1}a_i}$, while for $i\notin\bar{I}$, $G_{b_{i-1}b_i}=G_{a_{i-1}a_i}$. Hence, for $i\in\bar{I}$, we know from Lemma \ref{lem:GH*} that
    \[ \P\left( \left\|\frac{1}{dD}G_{a_{i-1}a_i}G^*_{b_{i-1}b_i} \right\|_\infty > \frac{16\sqrt{\ln 6}}{\sqrt{d}} \right) \leq e^{-cD} \, . \]
    While for $i\notin\bar{I}$, we know from Theorem \ref{th:Wishart} that
    \[ \P \left( \left\| \frac{1}{dD}G_{a_{i-1}a_i}G^*_{b_{i-1}b_i} \right\|_{\infty} = \left\| \frac{1}{dD}G_{a_{i-1}a_i}G^*_{a_{i-1}a_i} \right\|_{\infty} > 1+\frac{6}{\sqrt{d}} \right) \leq 2e^{-D/4} \, .  \]
    And therefore, by the union bound,
    \begin{align*}
    	\P\left( \prod_{i=1}^N \left\|\frac{1}{dD}G_{a_{i-1}a_i}G^*_{b_{i-1}b_i} \right\|_\infty > \left(\frac{16\sqrt{\ln 6}}{\sqrt{d}}\right)^{|\bar{I}|} \left(1+\frac{6}{\sqrt{d}}\right)^{N-|\bar{I}|} \right) & \leq \sum_{i\in\bar{I}} \P\left( \left\|\frac{1}{dD}G_{a_{i-1}a_i}G^*_{b_{i-1}b_i} \right\|_\infty > \frac{16\sqrt{\ln 6}}{\sqrt{d}} \right) \\
    	& + \sum_{i\notin\bar{I}} \P\left( \left\|\frac{1}{dD}G_{a_{i-1}a_i}G^*_{b_{i-1}b_i} \right\|_\infty > 1+\frac{6}{\sqrt{d}} \right) \\
    	& \leq |\bar{I}|e^{-cD} + 2(N-|\bar{I}|)e^{-D/4} \\
    	& \leq Ne^{-c'D} \, .
    \end{align*}
    Since $|\bar{I}|\geq|I|+1=q+1$, we thus have by the union bound again that
    \[ \P\left( \sum_{\substack{b_i\neq a_i,\,i\in I \\ b_i=a_i,\,i\notin I}} \prod_{i=1}^N \left\|\frac{1}{dD}G_{a_{i-1}a_i}G^*_{b_{i-1}b_i} \right\|_\infty > D^q\left(\frac{16\sqrt{\ln 6}}{\sqrt{d}}\right)^{q+1} \left(1+\frac{6}{\sqrt{d}}\right)^{N-q-1} \right) \leq D^qNe^{-c'D} \, . \]
    And consequently, once more by the union bound, the probability that
    \[ \sum_{q=1}^N \sum_{\substack{I\subset[N] \\ |I|=q}} \sum_{\substack{b_i\neq a_i,\,i\in I \\ b_i=a_i,\,i\notin I}} \prod_{i=1}^N \left\|\frac{1}{dD}G_{a_{i-1}a_i}G^*_{b_{i-1}b_i} \right\|_\infty > \sum_{q=1}^N {N\choose q}D^q\left(\frac{16\sqrt{\ln 6}}{\sqrt{d}}\right)^{q+1} \left(1+\frac{6}{\sqrt{d}}\right)^{N-q-1}  \]
    is smaller than
    \[  \left(\sum_{q=1}^N {N\choose q}D^q\right) Ne^{-c'D} \leq (D+1)^N Ne^{-c'D} \, . \]
    Hence, just noticing that
    \[ \sum_{q=1}^N {N\choose q}D^q\left(\frac{16\sqrt{\ln 6}}{\sqrt{d}}\right)^{q+1} \left(1+\frac{6}{\sqrt{d}}\right)^{N-q-1} \leq \left(1+\frac{6}{\sqrt{d}}+\frac{16\sqrt{\ln 6}D}{\sqrt{d}}\right)^N\frac{16\sqrt{\ln 6}}{\sqrt{d}} \leq \left(1+\frac{28D}{\sqrt{d}}\right)^N\frac{22}{\sqrt{d}} \, , \]
    we eventually get that
    \[ \P\left( \left\|W_{a_1,\ldots,a_N}\right\|_\infty > \left(1+\frac{28D}{\sqrt{d}}\right)^N\frac{22}{\sqrt{d}} \right) \leq (D+1)^N Ne^{-c'D} \, . \]
    And therefore, by the triangle inequality and the union bound,
    \[ \P\left(\|W'\|_\infty > \left(1+\frac{28D}{\sqrt{d}}\right)^N\frac{22}{\sqrt{d}} \right) \leq D^N(D+1)^N Ne^{-c'D} \, . \]
    
    We are now left with combining the results of the two parts of the proof. Indeed, since $\|\mathcal T_N(\Id)-\Id\|_\infty \leq \|W-\Id\|_\infty+\|W'\|_\infty$, we eventually obtain that
    \[ \P\left( \left\|\mathcal T_N(\Id)-\Id\right\|_\infty > \left(1+\frac{6}{\sqrt{d}}\right)^N\frac{6}{\sqrt{d}} + \left(1+\frac{28D}{\sqrt{d}}\right)^N\frac{22}{\sqrt{d}} \right) \leq 2D^2e^{-D/4} + (D+1)^{2N} Ne^{-c'D} \, ,  \]
    which yields the announced result.
\end{proof}

\subsubsection{Computing the typical overlap of the transfer operator with the maximally entangled state} \hfill\par\smallskip

\begin{proposition} \label{prop:PEPS-1}
	Let $T_N$ be defined as in equation \eqref{eq:transfer-PEPS}. Then, $\bra{\psi^{\otimes N}}T_N\ket{\psi^{\otimes N}}\in\R$ and
	\[ \P\left( \left| \bra{\psi^{\otimes N}}T_N\ket{\psi^{\otimes N}} - 1 \right| \leq \frac{42N}{\sqrt{d}}+D^2\left(\frac{84}{\sqrt{d}}\right)^N \right) \geq 1-6e^{-D^3/72} \, . \]
\end{proposition}

\begin{proof}
Let $T$ be defined as in equation \eqref{eq:transfer-MPS}. Observe that 
\begin{equation} \label{eq:trace-T^N}
\bra{\psi^{\otimes N}}T_N\ket{\psi^{\otimes N}}=\Tr\left(\tilde{T}^N\right) = \lambda_1(\tilde{T})^N+\cdots+ \lambda_{D^2}(\tilde{T})^N \, ,
\end{equation} 
where $\tilde{T}$ is distributed as $T$ with $\tilde{d}=D^2d$, $\tilde{D}=D$. Hence, we first know by Fact \ref{fact:real-trace} that $\bra{\psi^{\otimes N}}T_N\ket{\psi^{\otimes N}}\in\R$. Second, we know from Propositions \ref{prop:MPS-1'} and \ref{prop:MPS-2'} (combined with observations from Proposition \ref{prop:spectral-gap}) that
\[ \forall\ 0<\epsilon<\sqrt{d},\ \P\left( \left| \lambda_1(T) - 1 \right| \leq \frac{40+2\epsilon}{\sqrt{d}}\ \text{and}\ |\lambda_2(T)|,\ldots,|\lambda_{D^2}(T)|\leq \frac{80+4\epsilon}{\sqrt{d}} \right) \geq 1-6e^{-D\epsilon^2/72} \, .  \]
Consequently, 
\[ \forall\ 0<\epsilon<D\sqrt{d},\ \P\left( \left| \lambda_1(\tilde{T})-1 \right|\leq \frac{40+2\epsilon}{D\sqrt{d}}\ \text{and}\ \left|\lambda_2(\tilde{T})\right|,\ldots,\left|\lambda_{D^2}(\tilde{T})\right|\leq \frac{80+4\epsilon}{D\sqrt{d}} \right) \geq 1-6e^{-D\epsilon^2/72} \, .  \]
And therefore, taking $\epsilon=D$, we get that 
\[ \P\left( \left| \lambda_1(\tilde{T})-1 \right|\leq \frac{42}{\sqrt{d}}\ \text{and}\ \left|\lambda_2(\tilde{T})\right|,\ldots,\left|\lambda_{D^2}(\tilde{T})\right|\leq \frac{84}{\sqrt{d}} \right) \geq 1-6e^{-D^3/72} \, .  \]
By equation \eqref{eq:trace-T^N}, this implies that, with probability greater than $1-6e^{-D^3/72}$, the two following hold
\begin{align*} 
 & \bra{\psi^{\otimes N}}T_N\ket{\psi^{\otimes N}} \geq \left(1-\frac{42}{\sqrt{d}}\right)^N - (D^2-1)\left(\frac{84}{\sqrt{d}}\right)^N \geq 1- \frac{42N}{\sqrt{d}}-D^2\left(\frac{84}{\sqrt{d}}\right)^N \\
 & \bra{\psi^{\otimes N}}T_N\ket{\psi^{\otimes N}} \leq \left(1+\frac{42}{\sqrt{d}}\right)^N + (D^2-1)\left(\frac{84}{\sqrt{d}}\right)^N \leq 1 + \frac{42N}{\sqrt{d}}+D^2\left(\frac{84}{\sqrt{d}}\right)^N \, ,
\end{align*}
which is precisely what we wanted to show.
\end{proof}	

\subsubsection{Upper bounding the typical norm of the projection of the transfer operator on the orthogonal of the maximally entangled state} \hfill\par\smallskip

In the sequel, we will make extensive use of the following simple observation: given positive random variables $X,Y$ and positive numbers $x,y$,
\begin{align*} 
& \P(X+Y>x+y) \leq \P(X>x\ \text{or}\ Y>y) \leq \P(X>x) + \P(Y>y) \, , \\
& \P(XY>xy) \leq \P(X>x\ \text{or}\ Y>y) \leq \P(X>x) + \P(Y>y) \, .
\end{align*}
Indeed, if $X+Y>x+y$, resp.~$XY>xy$, then necessarily either $X>x$ or $Y>y$.

We now gather three deviation inequalities that we will also use repeatedly later on.

\begin{lemma} \label{lem:dev-norms}
Let $G_1,\ldots,G_d,H_1,\ldots,H_d$ be independent $D\times D$ matrices whose entries are independent complex Gaussians with mean $0$ and variance $1/D$. Then,
\begin{align*}
\forall\ \epsilon>0, & \ \P \left( \left\| \frac{1}{d}\sum_{x=1}^d G_x\otimes \bar{G}_x \right\|_{\infty} > (1+\epsilon)\left(1+\frac{41}{\sqrt{d}}\right) \right) \leq e^{-cDd\min(\epsilon,\epsilon^2)} \, , \\
& \ \P \left( \left\| \frac{1}{d}\sum_{x=1}^d G_x\otimes \bar{H}_x \right\|_{\infty} > (1+\epsilon)\frac{10}{\sqrt{d}} \right) \leq e^{-cDd\min(\epsilon/\sqrt{d},\epsilon^2/d)} \, , \\
& \ \P \left( \left\| \frac{1}{d}\sum_{x=1}^d (G_x\otimes \bar{G}_x - H_x\otimes \bar{H}_x) \right\|_{\infty} > (1+\epsilon)\frac{20}{\sqrt{d}} \right) \leq e^{-cDd\min(\epsilon/\sqrt{d},\epsilon^2/d)} \, , \\
\end{align*}
where $c>0$ is a universal constant.
\end{lemma}

\begin{proof}
The proof follows step by step that of Lemma \ref{lem:dev-ineq}, just slightly generalizing it in one point, so some details are skipped here.
We first fix $\delta>0$ and define the following subset of the set of $d$-uples of $D\times D$ matrices:
\[ \hat{\Omega}_{\delta} := \left\{ (A_1,\ldots,A_d) : \left(\sum_{x=1}^d \|A_x\|_{\infty}^2 \right)^{1/2} \leq (2+\delta)\sqrt{d} \right\} \, . \]
Reasoning as in the proof of Lemma \ref{lem:dev-ineq} we obtain that, for $G_1,\ldots,G_d$ independent $D\times D$ matrices whose entries are independent complex Gaussians with mean $0$ and variance $1/D$,	
\[ \P\left( (G_1,\ldots,G_d)\notin\hat{\Omega}_{\delta} \right) \leq e^{-Dd\delta^2}\, . \]
We then use the local version of the Gaussian concentration inequality, recalled in Theorem \ref{th:g-local}. 
In our case, the functions that we are looking at are 
\begin{align*} 
& \hat{f}_1:(G_1,\ldots,G_d) \mapsto \left\| \frac{1}{d}\sum_{x=1}^d G_x\otimes \bar{G}_x \right\|_{\infty} \, , \\
& \hat{f}_2:(G_1,\ldots,G_d,H_1,\ldots,H_d) \mapsto \left\| \frac{1}{d}\sum_{x=1}^d G_x\otimes \bar{H}_x \right\|_{\infty} \, .
\end{align*}
Arguing again exactly as in the proof of Lemma \ref{lem:dev-ineq} it can be shown that, for $(G_1,\ldots,G_d),(G_1',\ldots,G_d')\in\hat{\Omega}_{\delta}$
\[ \left| \hat{f}_1(G_1,\ldots,G_d) - \hat{f}_1(G_1',\ldots,G_d') \right| \leq \frac{2\sqrt{2}(2+\delta)}{\sqrt{d}}\left(\sum_{x=1}^d\|G_x-G_x'\|_2^2\right)^{1/2} \, , \]
while for $(G_1,\ldots,G_d,H_1,\ldots,H_d),(G_1',\ldots,G_d',H_1',\ldots,H_d')\in\hat{\Omega}_{\delta}\times\hat{\Omega}_{\delta}$
\[ \left| \hat{f}_2(G_1,\ldots,G_d,H_1,\ldots,H_d) - \hat{f}_1(G_1',\ldots,G_d',H_1',\ldots,H_d') \right| \leq \frac{\sqrt{2}(2+\delta)}{\sqrt{d}}\left(\sum_{x=1}^d\|G_x-G_x'\|_2^2+\sum_{x=1}^d\|H_x-H_x'\|_2^2\right)^{1/2} \, . \]
Finally we know that $\E \hat{f}_1 \leq 1+41/\sqrt{d}$ and $\E \hat{f}_2 \leq 10/\sqrt{d}$. Putting everything together, we thus get that
\begin{align*}
\forall\ \epsilon>0, & \ \P\left( \left\| \frac{1}{d}\sum_{x=1}^d G_x\otimes \bar{G}_x \right\|_{\infty} > (1+\epsilon)\left(1+\frac{41}{\sqrt{d}}\right) \right) \leq e^{-Dd\epsilon^2/8(2+\delta)^2} + e^{-Dd\delta^2} \, , \\
& \ \P \left( \left\| \frac{1}{d}\sum_{x=1}^d G_x\otimes \bar{H}_x \right\|_{\infty} > (1+\epsilon)\frac{10}{\sqrt{d}} \right) \leq e^{-32D\epsilon^2/(2+\delta)^2} + 2e^{-Dd\delta^2} \, . \\
\end{align*}
Choosing $\delta_1,\delta_2>0$ in the two deviation probabilities above satisfying, respectively, $e^{-Dd\epsilon^2/8(2+\delta_1)^2} = e^{-Dd\delta_1^2}$ and $e^{-32D\epsilon^2/(2+\delta_2)^2} = 2e^{-Dd\delta_2^2}$, we eventually obtain that there exist universal constants $c_1,c_2>0$ such that
\begin{align*}
\forall\ \epsilon>0, & \ \P\left( \left\| \frac{1}{d}\sum_{x=1}^d G_x\otimes \bar{G}_x \right\|_{\infty} > (1+\epsilon)\left(1+\frac{41}{\sqrt{d}}\right) \right) \leq e^{-c_1Dd\min(\epsilon,\epsilon^2)} \, , \\
& \ \P \left( \left\| \frac{1}{d}\sum_{x=1}^d G_x\otimes \bar{H}_x \right\|_{\infty} > (1+\epsilon)\frac{10}{\sqrt{d}} \right) \leq e^{-c_2Dd\min(\epsilon/\sqrt{d},\epsilon^2/d)} \, , \\
\end{align*}
which are precisely the first two deviation inequalities. As for the third one, we simply have to recall that, for each $1\leq x\leq d$, $G_x\otimes \bar{G}_x-H_x\otimes \bar{H}_x \sim G_x\otimes \bar{H}_x-H_x\otimes \bar{G}_x$. Therefore, for all $\epsilon>0$,
\begin{align*}
\P \left( \left\| \frac{1}{d}\sum_{x=1}^d (G_x\otimes \bar{G}_x-H_x\otimes \bar{H}_x )\right\|_{\infty} > (1+\epsilon)\frac{20}{\sqrt{d}} \right) & = \P \left( \left\| \frac{1}{d} \sum_{x=1}^d (G_x\otimes \bar{H}_x-H_x\otimes \bar{G}_x )\right\|_{\infty} > (1+\epsilon)\frac{20}{\sqrt{d}} \right) \\
& \leq \P \left( \left\| \frac{1}{d}\sum_{x=1}^d G_x\otimes \bar{H}_x\right\|_{\infty} + \left\| \frac{1}{d}\sum_{x=1}^d H_x\otimes \bar{G}_x \right\|_{\infty} > (1+\epsilon)\frac{20}{\sqrt{d}} \right) \\
& \leq \P \left( \left\| \frac{1}{d}\sum_{x=1}^d G_x\otimes \bar{H}_x \right\|_{\infty} > (1+\epsilon)\frac{10}{\sqrt{d}} \right) \\
& \ + \P \left( \left\| \frac{1}{d} \sum_{x=1}^d H_x\otimes \bar{G}_x \right\|_{\infty} > (1+\epsilon)\frac{10}{\sqrt{d}} \right) \, .
\end{align*}
And since we know by what precedes that the latter sum of deviation probabilities is upper bounded by 
\[  2e^{-c_2Dd\min(\epsilon/\sqrt{d},\epsilon^2/d)}\leq e^{-c_3Dd\min(\epsilon/\sqrt{d},\epsilon^2/d)} \, , \] 
the third deviation inequality is proved as well.
\end{proof}

\begin{lemma} \label{lem:dev-norm-sym-n}
Let $G_1,\ldots,G_d$ and $H_1^i,\ldots,H_d^i$, $1\leq i\leq N$, be independent $D\times D$ matrices whose entries are independent complex Gaussians with mean $0$ and variance $1/D$. Then, for all $\epsilon>0$, the probability that
\[ \left\| \frac{1}{d^N}\sum_{x_1,\ldots,x_N=1}^d \left( \underset{i=1}{\overset{N}{\bigotimes}}\, G_{x_i}\otimes\bar{G}_{x_i} - \underset{i=1}{\overset{N}{\bigotimes}}\, H_{x_i}^i\otimes\bar{H}_{x_i}^i \right) \right\|_{\infty} > (1+\epsilon)^N \left(1+\frac{41}{\sqrt{d}}\right)^{N-1}\frac{20N}{\sqrt{d}} \]
is smaller than $3Ne^{-cDd\min(\epsilon/\sqrt{d},\epsilon^2/d)}$, where $c>0$ is a universal constant.
\end{lemma}

\begin{proof}
	We will show by induction that the result is true for $N$ being any $n\in\N$. To simplify notation, we will set
	\[ Y := \left\| \frac{1}{d}\sum_{x=1}^d G_x\otimes \bar{G}_x \right\|_{\infty} \, ,  \] 
	and for each $n\in\N$,
	\[ X_n := \left\| \frac{1}{d^n}\sum_{x_1,\ldots,x_n=1}^d \left( \underset{i=1}{\overset{n}{\bigotimes}}\, G_{x_i}\otimes\bar{G}_{x_i} - \underset{i=1}{\overset{n}{\bigotimes}}\, H_{x_i}^i\otimes\bar{H}_{x_i}^i \right) \right\|_{\infty} \, . \]
	
	We know from Lemma \ref{lem:dev-norms} that
	\[ \forall\ \epsilon>0,\ \P \left( X_1 > (1+\epsilon)\frac{20}{\sqrt{d}} \right) \leq e^{-cDd\min(\epsilon/\sqrt{d},\epsilon^2/d)} \leq 3e^{-cDd\min(\epsilon/\sqrt{d},\epsilon^2/d)} \, . \]
	So the statement is true for $n=1$. 
	
	Let us now assume that the statement is true for some $n\in\N$ and show that this implies that it is true also for $n+1$. 
	Observe that, setting for each $1\leq i\leq n+1$ and $1\leq x_i\leq d$, $M_{x_i}:=G_{x_i}\otimes \bar{G}_{x_i}$ and $N_{x_i}^i:=H_{x_i}^i\otimes \bar{H}_{x_i}^i$, we have for each $1\leq x_1,\ldots,x_{n+1}\leq d$,
	\[ \underset{i=1}{\overset{n+1}{\bigotimes}}\, M_{x_i} - \underset{i=1}{\overset{n+1}{\bigotimes}}\, N_{x_i}^i = \left(\underset{i=1}{\overset{n}{\bigotimes}}\, M_{x_i}\right) \otimes \left(M_{x_{n+1}}-N_{x_{n+1}}^{n+1}\right) + \left( \underset{i=1}{\overset{n}{\bigotimes}}\, M_{x_i} - \underset{i=1}{\overset{n}{\bigotimes}}\, N_{x_i}^i \right) \otimes N_{x_{n+1}}^{n+1} \, . \]
	So by the triangle inequality,
	\begin{equation} \label{eq:recursion-ineq} X_{n+1} \leq Y^nX_1 + X_n'Y' \, . \end{equation}
	Now, we know from Lemma \ref{lem:dev-norms} that
	\[ \forall\ \epsilon>0,\ \P \left( Y > (1+\epsilon)\left(1+\frac{41}{\sqrt{d}}\right) \right) \leq e^{-cDd\min(\epsilon,\epsilon^2)} \, , \]
	and therefore also that
	\[ \forall\ \epsilon>0,\ \P \left( Y^n > (1+\epsilon)^n\left(1+\frac{41}{\sqrt{d}}\right)^n \right) \leq e^{-cDd\min(\epsilon,\epsilon^2)}  \, . \]
	What is more, we know by the initialisation step and by the recursion hypothesis that
	\begin{align*}
	\forall\ \epsilon>0, & \ \P \left( X_1 > (1+\epsilon)\frac{20}{\sqrt{d}} \right) \leq e^{-cDd\min(\epsilon/\sqrt{d},\epsilon^2/d)} \, , \\
	& \ \P \left( X_n > (1+\epsilon)^n\left(1+\frac{41}{\sqrt{d}}\right)^{n-1}\frac{20n}{\sqrt{d}} \right) \leq 3ne^{-cDd\min(\epsilon/\sqrt{d},\epsilon^2/d)} \, .
	\end{align*}
	Consequently, for all $\epsilon>0$,
	\begin{align*}
	\P \left( X_{n+1} > (1+\epsilon)^{n+1}\left(1+\frac{41}{\sqrt{d}}\right)^{n}\frac{20(n+1)}{\sqrt{d}} \right) & \leq \P \left( Y^nX_1 + X_nY' > (1+\epsilon)^{n+1}\left(1+\frac{41}{\sqrt{d}}\right)^{n}\frac{20(n+1)}{\sqrt{d}} \right) \\
	& \leq \P \left( Y' > (1+\epsilon)\left(1+\frac{41}{\sqrt{d}}\right) \right) + \P \left( Y^n > (1+\epsilon)^n\left(1+\frac{41}{\sqrt{d}}\right)^n \right) \\
	& \ + \P \left( X_1 > (1+\epsilon)\frac{20}{\sqrt{d}} \right) + \P \left( X_n > (1+\epsilon)^n\left(1+\frac{41}{\sqrt{d}}\right)^{n-1}\frac{20n}{\sqrt{d}} \right) \\
	& \leq 2e^{-cDd\min(\epsilon,\epsilon^2)} + (3n+1)e^{-cDd\min(\epsilon/\sqrt{d},\epsilon^2/d)} \\
	& \leq 3(n+1)e^{-cDd\min(\epsilon/\sqrt{d},\epsilon^2/d)} \, .
	\end{align*}
	So the statement is indeed true for $n+1$, which concludes the proof.
\end{proof}

\begin{corollary} \label{cor:dev-norm-sym-n}
	Let $G_1,\ldots,G_d$ and $H_1^i,\ldots,H_d^i$, $1\leq i\leq N$, be independent $D\times D$ matrices whose entries are independent complex Gaussians with mean $0$ and variance $1/D$. And define the random variables $G:=(G_1,\ldots,G_d)$ and $H:=(H_1^1,\ldots,H_d^1,\ldots,H_1^N,\ldots,H_d^N)$. Set also
	\begin{equation} \label{def:eta} 
	\eta\equiv \eta(N,d,D) := (4\sqrt{d}N)^{2(N+1)} \left(1+\frac{41}{\sqrt{d}}\right)^{N-1}\frac{20N}{\sqrt{d}} e^{-cD^3/d} \, . 
	\end{equation}
	Then, for all $d^{3/4}\geq D$, the probability over $G$ that
	\[ \E_H \left\| \frac{1}{d^N}\sum_{x_1,\ldots,x_N=1}^d \left( \underset{i=1}{\overset{N}{\bigotimes}}\, G_{x_i}\otimes\bar{G}_{x_i} - \underset{i=1}{\overset{N}{\bigotimes}}\, H_{x_i}^i\otimes\bar{H}_{x_i}^i \right) \right\|_{\infty} > (1+\eta)\left(1+\frac{D}{\sqrt{d}}\right)^N \left(1+\frac{41}{\sqrt{d}}\right)^{N-1}\frac{20N}{\sqrt{d}} \]
	is smaller than $3Ne^{-cD^3/d}$, where $c>0$ is a universal constant.
\end{corollary}

\begin{proof}
We know from Lemma \ref{lem:dev-norm-sym-n} that, for all $\epsilon>0$, the probability over $G,H$ that
\[ \left\| \frac{1}{d^N}\sum_{x_1,\ldots,x_N=1}^d \left( \underset{i=1}{\overset{N}{\bigotimes}}\, G_{x_i}\otimes\bar{G}_{x_i} - \underset{i=1}{\overset{N}{\bigotimes}}\, H_{x_i}^i\otimes\bar{H}_{x_i}^i \right) \right\|_{\infty} > \left(1+\epsilon\right)^N \left(1+\frac{41}{\sqrt{d}}\right)^{N-1}\frac{20N}{\sqrt{d}} \]
is smaller than $3Ne^{-cDd\min(\epsilon/\sqrt{d},\epsilon^2/d)}$. Taking $\epsilon=D/\sqrt{d}$, this implies that, for all $d\geq D$ (so that $\epsilon/\sqrt{d}\leq 1$), there exists a set $\Omega$ of $G$'s of measure larger than $1-3Ne^{-cD^3/d}$ such that, for all fixed $G\in\Omega$, the random variable 
\[ X_H:= \left\| \frac{1}{d^N}\sum_{x_1,\ldots,x_N=1}^d \left( \underset{i=1}{\overset{N}{\bigotimes}}\, G_{x_i}\otimes\bar{G}_{x_i} - \underset{i=1}{\overset{N}{\bigotimes}}\, H_{x_i}^i\otimes\bar{H}_{x_i}^i \right) \right\|_{\infty} \, \]
satisfies
\[ \forall\ \delta>0,\ \P\left( X_H > \left(1+\delta\right)^N \left(1+\frac{41}{\sqrt{d}}\right)^{N-1}\frac{20N}{\sqrt{d}} \right) \leq 3Ne^{-cDd\min(\delta/\sqrt{d},\delta^2/d)} \, . \]
Assume that the latter holds. Then, setting $M:=(1+41/\sqrt{d})^{N-1}20N/\sqrt{d}$, we can re-write
\[ \E X_H = \int_0^{\infty}\P(X_H=u)udu = \int_0^{(1+D/\sqrt{d})^NM}\P(X_H=u)udu + \int_{(1+D/\sqrt{d})^NM}^{\infty}\P(X_H=u)udu \, . \]
Now on the one hand,
\[ \int_0^{(1+D/\sqrt{d})^NM}\P(X_H=u)udu \leq \left(1+\frac{D}{\sqrt{d}}\right)^NM \, . \]
While on the other hand,
\begin{align*} 
\int_{(1+D/\sqrt{d})^NM}^{\infty}\P(X_H=u)udu & = NM^2\int_{D/\sqrt{d}}^{\infty}\P\left(X_H=(1+t)^NM\right) (1+t)^{2N-1}dt \\
& \leq NM^2 \int_{D/\sqrt{d}}^{\infty}\P\left(X_H>(1+t)^NM\right) (1+t)^{2N-1}dt \\
& \leq 3N^2M^2\left( \int_{D/\sqrt{d}}^{\sqrt{d}} e^{-cDt^2} (1+t)^{2N-1}dt + \int_{\sqrt{d}}^{\infty} e^{-cD\sqrt{d}t} (1+t)^{2N-1}dt \right) \, ,
\end{align*}
where the first equality is by change of variables and the last inequality is by assumption on $X_H$.
Yet, we clearly have
\[ \int_{D/\sqrt{d}}^{\sqrt{d}} e^{-cDt^2} (1+t)^{2N-1}dt \leq \frac{(1+\sqrt{d})^{2N}}{2N} e^{-cD^3/d} \leq (2\sqrt{d})^{2N}e^{-cD^3/d} \, . \]
And it can easily be shown by successive integrations by parts that, as soon as $D\sqrt{d}\geq 1/c$,
\[ \int_{\sqrt{d}}^{\infty} e^{-cD\sqrt{d}t} (1+t)^{2N-1}dt \leq (1+\sqrt{d})^{2N-1}\left(\sum_{q=0}^{2N-1}(2N-1)^q\right) \frac{e^{-cD\sqrt{d}}}{cD\sqrt{d}} \leq (4\sqrt{d}N)^{2N}e^{-cD\sqrt{d}} \, . \]
Hence in the end, for all $d^{3/4}\geq D \geq 28$,
\[ \int_{(1+D/\sqrt{d})^NM}^{\infty}\P(X_H=u)udu \leq 3N^2M^2\left( (2\sqrt{d})^{2N}e^{-cD^3/d} + (4\sqrt{d}N)^{2N}e^{-cD\sqrt{d}} \right) \leq (4\sqrt{d}N)^{2(N+1)}M^2e^{-cD^3/d} \, . \]
Putting everything together, we eventually obtain that, for all $d^{3/4}\geq D \geq 28$,
\[ \E X_H \leq \left(1+\eta\right)\left(1+\frac{D}{\sqrt{d}}\right)^N M \, , \]
which is exactly the announced result.
\end{proof}

\begin{proposition} \label{prop:PEPS-2}
	Let $T_N$ be defined as in equation \eqref{eq:transfer-PEPS}. Then, for all $\sqrt{d}\geq D$,
	\begin{align*} 
		& \P\left( \left\|T_N(\Id-\ketbra{\psi^{\otimes N}}{\psi^{\otimes N}})\right\|_{\infty} > (1+\eta)\left(1+\frac{93D}{\sqrt{d}}\right)^N\frac{60N}{\sqrt{d}} \right) \leq 4N(D+1)^{N}e^{-cD^3/d} \, , \\
		& \P\left( \left\|(\Id-\ketbra{\psi^{\otimes N}}{\psi^{\otimes N}})T_N\right\|_{\infty} > (1+\eta)\left(1+\frac{93D}{\sqrt{d}}\right)^N\frac{60N}{\sqrt{d}} \right) \leq 4N(D+1)^{N}e^{-cD^3/d} \, ,
	\end{align*}
	where $\eta\equiv\eta(N,d,D)$ is as defined in equation \eqref{def:eta} and $c>0$ is a universal constant.
\end{proposition}

\begin{proof}
	To begin with, arguing as in the proof of Proposition \ref{prop:MPS-2}, note that 
	\begin{equation} \label{eq:T-psi-PEPS} 
	\left\|T_N(\Id-\ketbra{\psi^{\otimes N}}{\psi^{\otimes N}})\right\|_{\infty} \leq 2 \left\|T_N-\ketbra{\psi^{\otimes N}}{\psi^{\otimes N}}\right\|_{\infty} \, .
	\end{equation}
	
	Next, observe that 
	\[ \ketbra{\psi^{\otimes N}}{\psi^{\otimes N}} = \E \frac{1}{D^Nd^n} \sum_{a_1,\ldots,a_N=1}^D \sum_{x_1,\ldots,x_N=1}^d H_{a_Na_1x_1}^1 \otimes \bar{H}_{a_Na_1x_1}^1\otimes \cdots\otimes H_{a_{N-1}a_Nx_N}^N \otimes \bar{H}_{a_{N-1}a_Nx_N}^N \, , \]
	where the $H_{a_{i-1}a_ix_i}^i$'s are independent $D\times D$ matrices whose entries are independent complex Gaussians with mean $0$ and variance $1/D$.
	We thus get by Jensen inequality that $\left\|T_N-\ketbra{\psi^{\otimes N}}{\psi^{\otimes N}}\right\|_{\infty}$ is upper bounded by
	\[ \E_H \left\| \frac{1}{D^N} \sum_{a_1,b_1,\ldots,a_N,b_N=1}^D \frac{1}{d^N} \sum_{x_1,\ldots,x_N=1}^d \underset{i=1}{\overset{N}{\bigotimes}}\left( G_{a_{i-1}a_ix_i} \otimes \bar{G}_{b_{i-1}b_ix_i} - \delta_{a_{i-1}b_{i-1}}\delta_{a_ib_i}H_{a_{i-1}a_ix_i}^i \otimes \bar{H}_{a_{i-1}a_ix_i}^i \right) \right\|_{\infty} \, , \]
	which, by the triangle inequality, is itself upper bounded by
	\[  \frac{1}{D^N} \sum_{a_1,b_1,\ldots,a_N,b_N=1}^D \E_H \left\| \frac{1}{d^N} \sum_{x_1,\ldots,x_N=1}^d \underset{i=1}{\overset{N}{\bigotimes}}\left( G_{a_{i-1}a_ix_i} \otimes \bar{G}_{b_{i-1}b_ix_i} - \delta_{a_{i-1}b_{i-1}}\delta_{a_ib_i}H_{a_{i-1}a_ix_i}^i \otimes \bar{H}_{a_{i-1}a_ix_i}^i \right) \right\|_{\infty} \, . \]
	
	To simplify notation latter on, let us set
	\[ M:= (1+\eta) \left(1+\frac{D}{\sqrt{d}}\right)^N \left(1+\frac{41}{\sqrt{d}}\right)^{N-1}\frac{20N}{\sqrt{d}} \ \text{and}\ M':= \left(1+\frac{D}{\sqrt{d}}\right)^N \left(1+\frac{51D}{\sqrt{d}}\right)^N\frac{10}{\sqrt{d}} \, , \]
	\[ p:=3Ne^{-cD^3/d} \ \text{and}\ p':= (D+1)^{N} Ne^{-cD^3/d} \, . \]
	
	First of all, we know from Corollary \ref{cor:dev-norm-sym-n} that, for any $1\leq a_1,\ldots,a_N\leq D$, for all $d^{3/4}\geq D$,
	\begin{equation} \label{eq:dev1} \P_G \left( \E_H \left\| \frac{1}{d^N} \sum_{x_1,\ldots,x_N=1}^d \underset{i=1}{\overset{N}{\bigotimes}}\left( G_{a_{i-1}a_ix_i} \otimes \bar{G}_{a_{i-1}a_ix_i} - H_{a_{i-1}a_ix_i}^i \otimes \bar{H}_{a_{i-1}a_ix_i}^i \right) \right\|_{\infty} > M \right) \leq p \, . \end{equation}
		
		We will now show that, for any $1\leq a_1,\ldots,a_N\leq D$, for all $\sqrt{d}\geq D$,
		\begin{equation} \label{eq:dev2} \P\left( \sum_{q=1}^N \sum_{\substack{I\subset[N] \\ |I|=q}} \sum_{\substack{b_i\neq a_i,\,i\in I \\ b_i=a_i,\,i\notin I}} \left\| \frac{1}{d^N} \sum_{x_1,\ldots,x_N=1}^d \underset{i=1}{\overset{N}{\bigotimes}}\, G_{a_{i-1}a_ix_i} \otimes \bar{G}_{b_{i-1}b_ix_i} \right\|_{\infty} > M' \right) \leq p' \, . \end{equation}
		Note that the latter random variable can be re-written as
		\[ \sum_{q=1}^N \sum_{\substack{I\subset[N] \\ |I|=q}} \sum_{\substack{b_i\neq a_i,\,i\in I \\ b_i=a_i,\,i\notin I}} \underset{i=1}{\overset{N}{\prod}}\, Z_{a_{i-1}a_ib_{i-1}b_i} \, , \]
		where we introduced the notation, for each $1\leq a,a',b,b'\leq D$,
		\[ Z_{aa'bb'}:= \left\| \frac{1}{d}\sum_{x=1}^d G_{aa'x} \otimes \bar{G}_{bb'x} \right\|_{\infty} \, . \]
		Given $1\leq q\leq N$ and $I\subset[N]$ such that $|I|=q$, we define $\bar{I}:=\cup_{i\in I}\{i,i+1\}$. We then have that, for $i\in\bar{I}$, $G_{b_{i-1}b_ix_i}$ is independent from $G_{a_{i-1}a_ix_i}$, while for $i\notin\bar{I}$, $G_{b_{i-1}b_ix_i}=G_{a_{i-1}a_ix_i}$. Hence, for $i\in\bar{I}$, we know from Lemma \ref{lem:dev-norms} with $\epsilon=D/\sqrt{d}$ that, for all $d\geq D$,
		\[ \P\left( Z_{a_{i-1}a_ib_{i-1}b_i} > \left(1+\frac{D}{\sqrt{d}}\right) \frac{10}{\sqrt{d}} \right) \leq e^{-cD^3/d} \, . \]
		While for $i\notin\bar{I}$, we know again from Lemma \ref{lem:dev-norms} with $\epsilon=D/\sqrt{d}$ that, for all $\sqrt{d}\geq D$,
		\[ \P\left( Z_{a_{i-1}a_ib_{i-1}b_i} = Z_{a_{i-1}a_ia_{i-1}a_i}  > \left(1+\frac{D}{\sqrt{d}}\right) \left(1+\frac{41}{\sqrt{d}}\right) \right) \leq e^{-cD^3} \, . \]
		And therefore, for all $\sqrt{d}\geq D$,
		\begin{align*}
		\P\left( \underset{i=1}{\overset{N}{\prod}}\, Z_{a_{i-1}a_ib_{i-1}b_i} > \left(1+\frac{D}{\sqrt{d}}\right)^N \left(\frac{10}{\sqrt{d}}\right)^{|\bar{I}|} \left(1+\frac{41}{\sqrt{d}}\right)^{N-|\bar{I}|} \right) & \leq \sum_{i\in\bar{I}} \P\left(Z_{a_{i-1}a_ib_{i-1}b_i} > \left(1+\frac{D}{\sqrt{d}}\right) \frac{10}{\sqrt{d}} \right) \\
		& \ + \sum_{i\notin\bar{I}} \P\left( Z_{a_{i-1}a_ib_{i-1}b_i} > \left(1+\frac{D}{\sqrt{d}}\right) \left(1+\frac{41}{\sqrt{d}}\right) \right) \\
		& \leq |\bar{I}|e^{-cD^3/d} + (N-|\bar{I}|)e^{-cD^3} \\
		& \leq Ne^{-cD^3/d} \, .
		\end{align*}
		Since $|\bar{I}|\geq |I|+1=q+1$, we thus have that, for all $\sqrt{d}\geq D$,
		\[ \P\left( \sum_{\substack{b_i\neq a_i,\,i\in I \\ b_i=a_i,\,i\notin I}} \underset{i=1}{\overset{N}{\prod}}\, Z_{a_{i-1}a_ib_{i-1}b_i} > \left(1+\frac{D}{\sqrt{d}}\right)^N D^q\left(\frac{10}{\sqrt{d}}\right)^{q+1} \left(1+\frac{41}{\sqrt{d}}\right)^{N-q-1} \right) \leq D^qNe^{-cD^3/d} \, . \]
		And consequently, for all $\sqrt{d}\geq D$, the probability that
		\[ \sum_{q=1}^N \sum_{\substack{I\subset[N] \\ |I|=q}} \sum_{\substack{b_i\neq a_i,\,i\in I \\ b_i=a_i,\,i\notin I}} \underset{i=1}{\overset{N}{\prod}}\, Z_{a_{i-1}a_ib_{i-1}b_i} > \left(1+\frac{D}{\sqrt{d}}\right)^N \sum_{q=1}^N {N \choose q} D^q \left(\frac{10}{\sqrt{d}}\right)^{q+1} \left(1+\frac{41}{\sqrt{d}}\right)^{N-q-1} \]
		is smaller than  
		\[ Ne^{-cD^3/d} \sum_{q=1}^N {N \choose q} D^q \leq (D+1)^{N} Ne^{-cD^3/d} \, . \]
		Simply noticing that
		\[ \sum_{q=1}^N {N \choose q} D^q \left(\frac{10}{\sqrt{d}}\right)^{q+1} \left(1+\frac{41}{\sqrt{d}}\right)^{N-q-1} \leq \left(1+\frac{41}{\sqrt{d}}+\frac{10D}{\sqrt{d}}\right)^N\frac{10}{\sqrt{d}} \leq \left(1+\frac{51D}{\sqrt{d}}\right)^N\frac{10}{\sqrt{d}} \, , \]
		we eventually get what we claimed, namely that, for all $\sqrt{d}\geq D$,
		\[ \P \left( \sum_{q=1}^N \sum_{\substack{I\subset[N] \\ |I|=q}} \sum_{\substack{b_i\neq a_i,\,i\in I \\ b_i=a_i,\,i\notin I}} \underset{i=1}{\overset{N}{\prod}}\, Z_{a_{i-1}a_ib_{i-1}b_i} > \left(1+\frac{D}{\sqrt{d}}\right)^N \left(1+\frac{41D}{\sqrt{d}}\right)^N\frac{10}{\sqrt{d}} \right) \leq (D+1)^{N} Ne^{-cD^3/d} \, . \]
		
		We now just have to combine the two results of equations \eqref{eq:dev1} and \eqref{eq:dev2} to obtain our final result. Indeed, observe that on the one hand
		\begin{align*} 
		M+M' & = (1+\eta)\left(1+\frac{D}{\sqrt{d}}\right)^N \left(1+\frac{41}{\sqrt{d}}\right)^{N-1}\frac{20N}{\sqrt{d}} + \left(1+\frac{D}{\sqrt{d}}\right)^N \left(1+\frac{51D}{\sqrt{d}}\right)^N\frac{10}{\sqrt{d}} \\ 
		& \leq (1+\eta)\left(1+\frac{D}{\sqrt{d}}\right)^N \left(1+\frac{51D}{\sqrt{d}}\right)^N\frac{30N}{\sqrt{d}} \\
		& \leq (1+\eta)\left(1+\frac{93D}{\sqrt{d}}\right)^N\frac{30N}{\sqrt{d}} \, , 
		\end{align*}
		while on the other hand
		\[ p+p'= 3Ne^{-cD^3/d} + (D+1)^{N} Ne^{-cD^3/d} \leq 4N(D+1)^{N}e^{-cD^3/d} \, . \]
		Hence, we have shown that, for all $\sqrt{d}\geq D$, 
		\[ \frac{1}{D^N} \sum_{a_1,b_1,\ldots,a_N,b_N=1}^D \E_H\left\| \frac{1}{d^N} \sum_{x_1,\ldots,x_N=1}^d \underset{i=1}{\overset{N}{\bigotimes}}\left( G_{a_{i-1}a_ix_i} \otimes \bar{G}_{b_{i-1}b_ix_i} - \delta_{a_{i-1}b_{i-1}}\delta_{a_ib_i}H_{a_{i-1}a_ix_i}^i \otimes \bar{H}_{a_{i-1}a_ix_i}^i \right) \right\|_{\infty} \] 
		is larger than
		\[ (1+\eta)\left(1+\frac{93D}{\sqrt{d}}\right)^N\frac{30N}{\sqrt{d}} \]
		with probability smaller than $4N(D+1)^{N}e^{-cD^3/d}$, which implies that the same holds for $\left\|T_N-\ketbra{\psi^{\otimes N}}{\psi^{\otimes N}}\right\|_{\infty}$. And this in turn implies by equation \eqref{eq:T-psi-PEPS} that the same holds for $\left\|T_N\left(\Id-\ketbra{\psi^{\otimes N}}{\psi^{\otimes N}}\right)\right\|_{\infty}/2$, as announced.
		
		Just as in the proof of Proposition \ref{prop:MPS-2'}, the second deviation probability follows from the first one, applied to $T_N^*$ instead of $T_N$.
\end{proof}

\subsubsection{Typical spectral gap of the transfer operator} \hfill\par\smallskip

\begin{theorem} \label{th:gap-PEPS}
	Let $T_N$ be the random PEPS transfer operator as defined in equation \eqref{eq:transfer-PEPS}. Then, for all $\sqrt{d}\geq D$, 
	\begin{equation} \label{eq:gap-PEPS} 
		\Delta(T_N) \geq 1 - 2\left(1+\frac{28D}{\sqrt{d}}\right)^N\frac{28}{\sqrt{d}} - \frac{42N}{\sqrt{d}}-D^2\left(\frac{84}{\sqrt{d}}\right)^N - 3(1+\eta)\left(1+\frac{93D}{\sqrt{d}}\right)^N\frac{60N}{\sqrt{d}}
	\end{equation}
    with probability larger than
    \[ 1-(D+1)^{2N} (N+2)e^{-cD}-6e^{-D^3/72}-8N(D+1)^{N}e^{-cD^3/d} \, , \]
    where $\eta\equiv\eta(N,d,D)$ is as defined in equation \eqref{def:eta} and $c>0$ is a universal constant.
	 
	In particular, if $d\simeq N^{\alpha}$ and $D\simeq N^{\beta}$ with $\alpha>8$ and $(\alpha+1)/3<\beta<(\alpha-2)/2$, then
    \begin{equation} \label{eq:gap-PEPS-specific} 
    \P\left( \Delta(T_N) \geq 1-\frac{C}{N^{\alpha/2-\beta-1}} \right) \geq 1-e^{-c'N^{3\beta-\alpha}} \, ,
    \end{equation}
	where $C,c'>0$ are universal constants.
\end{theorem}

\begin{proof}
	First, we know by Proposition \ref{prop:PEPS-CP} that
	\[ \P\left( \mathcal T_N(\Id) \ngeq \left(1- \left(1+\frac{28D}{\sqrt{d}}\right)^N\frac{28}{\sqrt{d}} \right) \Id \right) \leq (D+1)^{2N} (N+2)e^{-cD} \, . \]
	Second, we know from Proposition \ref{prop:PEPS-1} that
	\[ \P\left( \left| \bra{\psi^{\otimes N}}T_N\ket{\psi^{\otimes N}} - 1 \right| > \frac{42N}{\sqrt{d}}+D^2\left(\frac{84}{\sqrt{d}}\right)^N \right) \leq 6e^{-D^3/72} \, , \] 
	and from Proposition \ref{prop:PEPS-2} that
	\begin{align*} 
		& \P\left( \left\|T_N(\Id-\ketbra{\psi^{\otimes N}}{\psi^{\otimes N}})\right\|_{\infty} > (1+\eta)\left(1+\frac{93D}{\sqrt{d}}\right)^N\frac{60N}{\sqrt{d}} \right) \leq 4N(D+1)^{N}e^{-cD^3/d} \, , \\
		& \P\left( \left\|(\Id-\ketbra{\psi^{\otimes N}}{\psi^{\otimes N}})T_N\right\|_{\infty} > (1+\eta)\left(1+\frac{93D}{\sqrt{d}}\right)^N\frac{60N}{\sqrt{d}} \right) \leq 4N(D+1)^{N}e^{-cD^3/d} \, .
	\end{align*}
	Now, we also know by Proposition \ref{prop:spectral-gap} that, if the three following conditions are satisfied
	\begin{align*}
	& \mathcal T_N(\Id) \geq \left(1- \left(1+\frac{28D}{\sqrt{d}}\right)^N\frac{28}{\sqrt{d}} \right) \Id \, , \\
	& \left| \bra{\psi^{\otimes N}}T_N\ket{\psi^{\otimes N}} - 1 \right| \leq \frac{42N}{\sqrt{d}}+D^2\left(\frac{84}{\sqrt{d}}\right)^N \, , \\
	& \left\|T_N(\Id-\ketbra{\psi^{\otimes N}}{\psi^{\otimes N}})\right\|_{\infty}, \left\|(\Id-\ketbra{\psi^{\otimes N}}{\psi^{\otimes N}})T_N\right\|_{\infty} \leq (1+\eta)\left(1+\frac{93D}{\sqrt{d}}\right)^N\frac{60N}{\sqrt{d}} \, ,
    \end{align*}
then
\[ \Delta(T_N) \geq 1 - 2\left(1+\frac{28D}{\sqrt{d}}\right)^N\frac{28}{\sqrt{d}} - \frac{42N}{\sqrt{d}}-D^2\left(\frac{84}{\sqrt{d}}\right)^N - 3(1+\eta)\left(1+\frac{93D}{\sqrt{d}}\right)^N\frac{60N}{\sqrt{d}} \, . \]
This means that the above holds with probability larger than 
\[ 1-(D+1)^{2N} (N+2)e^{-cD}-6e^{-D^3/72}-8N(D+1)^{N}e^{-cD^3/d} \, , \] 
and the first assertion \eqref{eq:gap-PEPS} is proved. 

The second assertion \eqref{eq:gap-PEPS-specific} can easily be checked by inserting the values of $d$ and $D$ in equation \eqref{eq:gap-PEPS}. The only details to check are that for $\alpha>8$, $(\alpha+1)/3<(\alpha-2)/2$ (so that the range of possible values for $\beta$ is not empty).
\end{proof}

\section{Consequence: Typical correlation length in random MPS and PEPS}
\label{sec:decay-correlations-2}

In the previous section we showed that the transfer operators of our random MPS and PEPS are typically gapped. In this section we derive from the latter result that our random MPS and PEPS typically exhibit exponential decay of correlations. Compared to the statements in Section \ref{sec:decay-correlations-1}, those in the current section have one main advantage: they apply to a dimensional regime that goes beyond the one of injectivity.

\subsection{Preliminary facts} \hfill\par\smallskip

\begin{lemma} \label{lem:trace-powers}
Let $M$ be an $n\times n$ complex matrix satisfying the following: $\lambda_1(M)=\lambda$ for some $\lambda\in\C$ and, for all $2\leq i\leq n$, $\lambda_i(M)=\lambda\epsilon_i$ with $|\epsilon_i|\leq \epsilon$ for some $0<\epsilon<1$. Then, there exists a unit vector $\ket{\varphi}\in\C^n$ such that, for any $k,k'\in\N$ and any $n\times n$ complex matrices $A,A'$,
\begin{align*} 
& \Tr(M^k) = \lambda^k\left(1+\epsilon^{(k)}\right)\ \text{with}\ |\epsilon^{(k)}|\leq n\epsilon^k \, , \\
& \Tr(AM^k) = \lambda^k\left(\bra{\varphi}A\ket{\varphi}+\epsilon^{(k)}_A\right)\ \text{with}\ |\epsilon^{(k)}_A|\leq n\epsilon^k\|A\|_{\infty} \, , \\
& \Tr(AM^kA'M^{k'}) = \lambda^{k+k'}\left(\bra{\varphi}A\ket{\varphi}\bra{\varphi}A'\ket{\varphi}+\epsilon^{(k,k')}_{A,A'}\right)\ \text{with}\ |\epsilon^{(k,k')}_{A,A'}|\leq (\epsilon^k+\epsilon^{k'}+n\epsilon^{k+k'})\|A\|_{\infty}\|A'\|_{\infty} \, . \\
\end{align*}
\end{lemma}

\begin{proof}
By the Schur decomposition, we know that $M$ can be written in triangular form as $M = \lambda ( \ketbra{\varphi_1}{\varphi_1} + R )$, with $\bra{\varphi_1}R\ket{\varphi_1}=0$, $\bra{\varphi_j}R\ket{\varphi_i}=0$ for all $1\leq i<j\leq n$, and $\|R\|_{\infty}\leq \epsilon$. Hence, for any $\ell\in\N$, $M^{\ell}=\lambda^{\ell}(\ketbra{\varphi_1}{\varphi_1} + R^{(\ell)})$, with $R^{(\ell)}$ having the same form as $R$ and $\|R^{(\ell)}\|_{\infty}\leq \epsilon^{\ell}$. As a consequence, we have 
\begin{align*} 
& \Tr(M^k) = \lambda^k\left(1+\Tr(R^{(k)})\right) \, , \\
& \Tr(AM^k) = \lambda^k\left(\bra{\varphi_1}A\ket{\varphi_1}+\Tr(AR^{(k)})\right) \, , \\
& \Tr(AM^kA'M^{k'}) = \lambda^{k+k'}\left(\bra{\varphi_1}A\ket{\varphi_1}\bra{\varphi_1}A'\ket{\varphi_1}+ \bra{\varphi_1}AR^{(k)}A'\ket{\varphi_1} + \bra{\varphi_1}A'R^{(k')}A\ket{\varphi_1} + \Tr(AR^{(k)}A'R^{(k')})\right) \, . 
\end{align*}
We now just need to upper bound the error terms in the three expressions above. For the first one we clearly have
\[ |\Tr(R^{(k)})|\leq (n-1)\epsilon^k \, . \]
For the second one we get by H\"{o}lder inequality that
\[ |\Tr(AR^{(k)})|\leq \|R^{(k)}\|_1\|A\|_{\infty}\leq (n-1)\epsilon^k\|A\|_{\infty} \, . \]
And for the third one we get again by H\"{o}lder inequality that
\begin{align*} 
& |\bra{\varphi_1}AR^{(k)}A'\ket{\varphi_1}| \leq \|AR^{(k)}A'\|_{\infty} \leq \|R^{(k)}\|_{\infty}\|A\|_{\infty}\|A'\|_{\infty} \leq \epsilon^k\|A\|_{\infty}\|A'\|_{\infty} \, , \\
& |\bra{\varphi_1}A'R^{(k')}A\ket{\varphi_1}| \leq \|A'R^{(k')}A\|_{\infty} \leq \|R^{(k')}\|_{\infty}\|A\|_{\infty}\|A'\|_{\infty} \leq \epsilon^{k'}\|A\|_{\infty}\|A'\|_{\infty} \, , \\
& |\Tr(AR^{(k)}A'R^{(k')})|\leq \|R^{(k)}\|_2\|R^{(k')}\|_2\|A\|_{\infty}\|A'\|_{\infty}\leq (n-1)\epsilon^{k+k'}\|A\|_{\infty}\|A'\|_{\infty} \, . 
\end{align*}
And we thus get precisely the announced result.
\end{proof}

\begin{corollary} \label{cor:correlation}
Let $M$ be an $n\times n$ complex matrix satisfying the following: $\lambda_1(M)=\lambda$ for some $\lambda\in\C$ and, for all $2\leq i\leq n$, $\lambda_i(M)=\lambda\epsilon_i$ with $|\epsilon_i|\leq \epsilon$ for some $0<\epsilon<1$. Let $k,k'\in\N$ be such that $k\leq k'$ and $\log (n)/\log(1/\epsilon)-2 \leq k'$. Then, for any $n\times n$ complex matrices $A,A'$,
\[ \left| \frac{\Tr(AM^kA'M^{k'})}{\Tr(M^{k+k'+2})} - \frac{\Tr(AM^{k+k'+1})\Tr(A'M^{k+k'+1})}{\left(\Tr(M^{k+k'+2})\right)^2} \right| \leq \frac{10}{|\lambda|^2(1-\epsilon^k)^2} \epsilon^k\|A\|_{\infty}\|A'\|_{\infty} \, . \]
\end{corollary}

\begin{proof}
We know by Lemma \ref{lem:trace-powers} (and using the notation introduced there) that
\begin{align*}
& \Tr(M^{k+k'+2}) = \lambda^{k+k'+2}\left(1+\epsilon^{(k+k'+2)}\right) \, , \\
& \Tr(AM^{k+k'+1}) = \lambda^{k+k'+1}\left(\bra{\varphi}A\ket{\varphi}+\epsilon^{(k+k'+1)}_A\right) \, , \\
& \Tr(A'M^{k+k'+1}) = \lambda^{k+k'+1}\left(\bra{\varphi}A'\ket{\varphi}+\epsilon^{(k+k'+1)}_{A'}\right) \, , \\
& \Tr(AM^kA'M^{k'}) = \lambda^{k+k'}\left(\bra{\varphi}A\ket{\varphi}\bra{\varphi}A'\ket{\varphi}+\epsilon^{(k,k')}_{A,A'}\right) \, . \\
\end{align*}
Set $\alpha:=\bra{\varphi}A\ket{\varphi}$ and $\alpha':=\bra{\varphi}A'\ket{\varphi}$, which are clearly such that $|\alpha|\leq\|A\|_{\infty}$ and $|\alpha'|\leq\|A'\|_{\infty}$. Set also
\[ \gamma := \left| \frac{\Tr(AM^kA'M^{k'})}{\Tr(M^{k+k'+2})} - \frac{\Tr(AM^{k+k'+1})\Tr(A'M^{k+k'+1})}{\left(\Tr(M^{k+k'+2})\right)^2} \right| \, . \]
We then have
\begin{align*}
\gamma & = \frac{1}{|\lambda|^2|1+\epsilon^{(k+k'+2)}|^2} \left| \left(\alpha\alpha'+\epsilon^{(k,k')}_{A,A'}\right)\left(1+\epsilon^{(k+k'+2)}\right) - \left(\alpha+\epsilon^{(k+k'+1)}_A\right)\left(\alpha'+\epsilon^{(k+k'+1)}_{A'}\right) \right| \\
& = \frac{1}{|\lambda|^2|1+\epsilon^{(k+k'+2)}|^2} \left| \left(1+\epsilon^{(k+k'+2)}\right)\epsilon^{(k,k')}_{A,A'} + \epsilon^{(k+k'+2)}\alpha\alpha' - \alpha\epsilon^{(k+k'+1)}_{A'} - \alpha'\epsilon^{(k+k'+1)}_A - \epsilon^{(k+k'+1)}_A\epsilon^{(k+k'+1)}_{A'} \right| \\
& \leq \frac{1}{|\lambda|^2(1-n\epsilon^{k+k'+2})^2} \left( (1+n\epsilon^{k+k'+2})(\epsilon^k+\epsilon^{k'}+n\epsilon^{k+k'}) + n\epsilon^{k+k'+2} + 2n\epsilon^{k+k'+1} + n^2\epsilon^{2(k+k'+1)} \right) \|A\|_{\infty}\|A'\|_{\infty} \\
& \leq \frac{10}{|\lambda|^2(1-\epsilon^k)^2}\epsilon^k \|A\|_{\infty}\|A'\|_{\infty} \, ,
\end{align*}
which is exactly what we wanted to show.
\end{proof}

\begin{lemma} \label{lem:chi-T}
	Let $\ket{\chi}\in(\C^d\otimes\C^D\otimes\C^D)^{\otimes M}$ be the $M$-site column tensor of a translation-invariant MPS (if $M=1$) or PEPS (if $M>1$). Denote by $T$ its associated transfer operator on $(\C^D\otimes\C^D)^{\otimes M}$ and by $\ket{\chi^{N}}\in(\C^d)^{\otimes MN}$ its contraction on $N$ sites. Fix also $0\leq k\leq N-2$. Then, for any Hermitian operators $A,A'$ on $(\C^d)^{\otimes M}$, 
	\begin{align*}
	& \braket{\chi}{\chi} = \Tr(T) \, , \\
	& \braket{\chi^{N}}{\chi^{N}} = \Tr\left(T^N\right) \, , \\
	& \bra{\chi^{N}}A_1\otimes\Id_{N-1}\ket{\chi^{N}} = \Tr\left(\tilde{A}T^{N-1}\right) \, , \\
	& \bra{\chi^{N}}A_1\otimes\Id_k\otimes A'_1\otimes\Id_{N-k-2}\ket{\chi^{N}} = \Tr\left(\tilde{A}T^k\tilde{A}'T^{N-k-2}\right) \, ,
	\end{align*}
	where $\tilde{A},\tilde{A}'$ are operators on $(\C^D\otimes\C^D)^{\otimes M}$ satisfying $\|\tilde{A}\|_{\infty}\leq\braket{\chi}{\chi}\|A\|_{\infty}, \|\tilde{A}'\|_{\infty}\leq\braket{\chi}{\chi}\|A'\|_{\infty}$.
\end{lemma}

\begin{proof}
	The first two equalities are simply by definition of $T$. So let us turn to the last two equalities. Given a Hermitian operator $B_H$ on $(\C^d)^{\otimes M}\equiv\mathrm{H}$, its corresponding operator $\tilde{B}_E$ on $(\C^D\otimes\C^D)^{\otimes M}\equiv\mathrm{E}$ is defined by 
	\[ \tilde{B}_E := \Tr_H\left[(B_H\otimes\Id_E)\ketbra{\chi}{\chi}_{HE}\right] \, . \]
	Hence, for any unit vector $\ket{\varphi}_E$, we have
	\begin{align*}
	\left| \bra{\varphi_E}\tilde{B}_E\ket{\varphi_E} \right| & = \left| \Tr_E\left[\Tr_H\left[(B_H\otimes\Id_E)\ketbra{\chi}{\chi}_{HE}\right]\ketbra{\varphi}{\varphi}_E \right] \right| \\
	& = \left| \Tr_{HE}\left[(B_H\otimes\Id_E)\ketbra{\chi}{\chi}_{HE}(\Id_H\otimes\ketbra{\varphi}{\varphi}_E) \right] \right| \\
	& = \left| \Tr_{HE}\left[(B_H\otimes\ketbra{\varphi}{\varphi}_E)\ketbra{\chi}{\chi}_{HE} \right] \right| \\
	& \leq \braket{\chi_{HE}}{\chi_{HE}} \|B_H\otimes\ketbra{\varphi}{\varphi}_E\|_{\infty} \\
	& = \braket{\chi_{HE}}{\chi_{HE}} \|B_H\|_{\infty} \, .
	\end{align*}
	And we thus have shown, as wanted, that $\|\tilde{B}_E\|_{\infty} \leq \braket{\chi_{HE}}{\chi_{HE}} \|B_H\|_{\infty}$.
\end{proof}	

The way $\tilde{A}$ is constructed from $A$ in the above proof is probably much easier to understand with a diagram than with a formula. In the MPS case, it is simply represented by:
\begin{center}
	\begin{tikzpicture} [scale=0.6]
	\begin{scope}[decoration={markings,mark=at position 0.5 with {\arrow{>}}}] 
	\draw[postaction={decorate}, color=brown] (10,0.6) -- (10,0); \draw[postaction={decorate}, color=brown] (10,2) -- (10,1.4);  
	\draw (10,1) node {{\small A}}; \draw (9.6,0.6) rectangle (10.4,1.4);
	\draw (10,-0.7) node {$A:\C^d\longrightarrow\C^d$};
	
	\begin{scope}[xshift=7cm]
	\draw[color=brown] (10,0) -- (10,0.6); \draw[color=brown] (10,1.4) -- (10,2); \draw[postaction={decorate}, color=gray] (11,0) -- (10,0); \draw[postaction={decorate}, color=gray] (10,0) -- (9,0);
	\draw[postaction={decorate}, color=gray] (11,2) -- (10,2); \draw[postaction={decorate}, color=gray] (10,2) -- (9,2); 
	\draw (10,1) node {{\small A}}; \draw (9.6,0.6) rectangle (10.4,1.4);
	\draw (10,-0.7) node {$\tilde{A}:\C^D\otimes\C^D\longrightarrow\C^D\otimes\C^D$};
	\end{scope}
	\end{scope}
	\end{tikzpicture}
\end{center}

Let us now make a simple but important observation. Let $\ket{\chi}\in(\C^d\otimes\C^D\otimes\C^D)^{\otimes M}$ be the $M$-site tensor of a translation-invariant MPS (if $M=1$) or PEPS (if $M>1$). Denote by $T$ its associated transfer operator on $(\C^D\otimes\C^D)^{\otimes M}$ and by $\ket{\chi^{N}}\in(\C^d)^{\otimes MN}$ its contraction on $N$ sites. We will focus here on a slightly less general setting than the one of Section \ref{sec:decay-correlations-1}, for the sake of readability. More precisely, we will consider only the case of Hermitian operators supported on one $M$-site column. Hence, for any $0\leq k\leq N-2$ and any Hermitian operators $A,A'$ on $(\C^d)^{\otimes M}$, we adapt the definition of the correlation function $\gamma_{\chi}(A,A',k)$ from equation \eqref{eq:def-C} to
\begin{equation*} 
\gamma_{\chi}(A,A',k) := \left| \frac{\bra{\chi^{N}}A_1\otimes\Id_k\otimes A'_1\otimes\Id_{N-k-2}\ket{\chi^{N}}}{\braket{\chi^{N}}{\chi^{N}}} - \frac{\bra{\chi^{N}}A_1\otimes\Id_{N-1}\ket{\chi^{N}}\bra{\chi^{N}}A'_1\otimes\Id_{N-1}\ket{\chi^{N}}}{\braket{\chi^{N}}{\chi^{N}}^2} \right| \, . 
\end{equation*}
By Lemma \ref{lem:chi-T} we know that the latter can actually be re-written as
\[ \gamma_{\chi}(A,A',k) = \left| \frac{\Tr(\tilde{A}T^k\tilde{A}'T^{N-k-2})}{\Tr(T^N)} - \frac{\Tr(\tilde{A}T^{N-1})\Tr(\tilde{A}'T^{N-1})}{(\Tr(T^N))^2} \right| \, , \]
where $\tilde{A},\tilde{A}'$ are operators on $(\C^D\otimes\C^D)^{\otimes M}$ satisfying $\|\tilde{A}\|_{\infty}\leq\Tr(T)\|A\|_{\infty}, \|\tilde{A}'\|_{\infty}\leq\Tr(T)\|A'\|_{\infty}$.

In what follows, we will need two easy deviation bounds for some scalar products in Gaussian variables. We gather them below.

\begin{lemma} \label{lem:gg}
Let $g$ be a Gaussian vector in $\C^n$ with mean $0$ and variance $1/n$. Then,
\[ \forall\ \epsilon>0,\ \P\left( \|g\| > 1+\epsilon \right) \leq e^{-n\epsilon^2} \, . \]
\end{lemma}

\begin{proof}
Lemma \ref{lem:gg} is a straightforward application of the Gaussian concentration inequality, as recalled in Theorem \ref{th:g-global}. Indeed, first by Jensen inequality 
\[ \E\|g\| \leq \left(\E\|g\|^2\right)^{1/2} =1 \, . \]
And second it is clear that $g\mapsto\|g\|$ is $1$-Lipschitz. Therefore, 
\[ \forall\ \epsilon>0,\ \P\left( \|g\| > 1+\epsilon \right) \leq e^{-n\epsilon^2} \, , \]
which is exactly the announced result.
\end{proof}

\begin{lemma} \label{lem:gg'}
Let $g,g'$ be independent Gaussian vectors in $\C^n$ with mean $0$ and variance $1/n$. Then,
\[ \forall\ \epsilon>0,\ \P\left( |\braket{g}{g'}| > \epsilon \right) \leq 2e^{-n\epsilon^2} \, . \]
\end{lemma}

\begin{proof}
Lemma \ref{lem:gg'} is simply a bound on the tails of the Gaussian distribution. Indeed, observe that $\braket{g}{g'}$ is distributed as a complex Gaussian $g_0$ with mean $0$ and variance $1/n$. And it is well-known that, for such $g_0$,
\[ \forall\ \epsilon>0,\ \P\left( |g_0| > \epsilon \right) \leq 2e^{-n\epsilon^2} \, , \]
which concludes the proof.
\end{proof}

\subsection{The case of MPS} \hfill\par\smallskip
\label{sec:decay-correlations-2-MPS}

\begin{lemma} \label{lem:trace-T}
Let $T$ be defined as in equation \eqref{eq:transfer-MPS}. Then,
\[ \forall\ \epsilon>0,\ \P\left( \Tr (T) \leq (1+\epsilon)^2 \right) \geq 1-e^{-d\epsilon^2} \, . \]
\end{lemma}

\begin{proof}
To begin with, observe that $\Tr(T)$ is distributed as $\|g\|^2$, where $g$ is a Gaussian vector in $\C^d$ with mean $0$ and variance $1/d$. Indeed,
\[ \Tr(T) = \frac{1}{d} \sum_{x=1}^d \Tr(G_x)\Tr(\bar{G}_x) = \sum_{x=1}^d g_x\bar{g}_x \, , \] 
where we have set, for each $1\leq x\leq d$, $g_x:=\Tr(G_x)/\sqrt{d}$, so that the $g_x$'s are distributed as independent complex Gaussians with mean $0$ and variance $1/d$. Now, for such vector $g$, we know by Lemma \ref{lem:gg} that
\[ \forall\ \epsilon>0,\ \P\left( \|g\|^2 > (1+\epsilon)^2 \right) = \P\left( \|g\| > 1+\epsilon \right) \leq e^{-d\epsilon^2} \, . \]
And the proof is thus complete.
\end{proof}	

\begin{theorem} \label{th:dc2-MPS}
Let $\ket{\chi^{N}}\in(\C^d)^{\otimes N}$ be the random $N$-site translation-invariant MPS whose random $1$-site tensor $\ket{\chi}\in\C^d\otimes\C^D\otimes\C^D$ is defined as in equation \eqref{eq:MPS}. Then, with probability larger than $1-e^{-c\min(D,d^{1/3})}$, for any $k\leq N-C_0\log D/\log d$ and any Hermitian operators $A,A'$ on $\C^d$,
\[ \gamma_{\chi}(A,A',k) \leq C'\times\left(\frac{C}{\sqrt{d}}\right)^k\|A\|_{\infty}\|A'\|_{\infty} \, , \]
where $C_0,c,C,C'>0$ are universal constants.
\end{theorem}

\begin{proof}
First of all, we know by Theorem \ref{th:gap-MPS} that there exist universal constants $\hat{c},\hat{C}>0$ such that, with probability larger than $1-e^{-\hat{c}D}$,
\[ |\lambda_1(T)|  \geq 1-\frac{\hat{C}}{\sqrt{d}}\ \ \text{and}\ \ \forall\ 2\leq i\leq D^2,\ |\lambda_i(T)| \leq |\lambda_1(T)|\times \frac{\hat{C}}{\sqrt{d}} \, . \]
Next, we know by Lemma \ref{lem:trace-T} (applied with, say, $\epsilon=1/d^{1/3}$) that, with probability larger than $1-e^{-d^{1/3}}$,
\[ \Tr(T)\leq \left(1+\frac{1}{d^{1/3}}\right)^2 \, . \]
Now, combining Corollary \ref{cor:correlation} and Lemma \ref{lem:chi-T} we know that, setting $\epsilon(T):=|\lambda_1(T)|/|\lambda_2(T)|$, we have, for any $k\leq N-C_0\log D/\log d$ and any Hermitian operators $A,A'$ on $\C^d$,
\[ \gamma_{\chi}(A,A',k) \leq \frac{10(\Tr(T))^2}{|\lambda_1(T)|^2\left(1-\epsilon(T)^k\right)^2} \epsilon(T)^k \|A\|_{\infty}\|A'\|_{\infty} \, .\]
Putting everything together, we thus eventually get that there exist universal constants $c,C, C'>0$ such that, with probability larger than $1-e^{-c\min(D,d^{1/3})}$, for any $k\leq N-C_0\log D/\log d$ and any Hermitian operators $A,A'$ on $\C^d$,
\[ \gamma_{\chi}(A,A',k) \leq C'\times\left(\frac{C}{\sqrt{d}}\right)^k\|A\|_{\infty}\|A'\|_{\infty} \, ,\]
exactly as announced.
\end{proof}

\subsection{The case of PEPS} \hfill\par\smallskip
\label{sec:decay-correlations-2-PEPS}

\begin{lemma} \label{lem:trace-T_N}
	Let $T_N$ be defined as in equation \eqref{eq:transfer-PEPS}. Then,
	\[ \forall\ \epsilon>0,\ \P\left( \Tr (T_N) \leq (1+4D\epsilon)^{2N} \right) \geq 1-2N(D+1)^Ne^{-d\epsilon^2} \, . \]
\end{lemma}

\begin{proof}
The proof technique is largely inspired from that of Proposition \ref{prop:PEPS-2}. So we might skip a few details here.

To begin with, observe that $\Tr(T_N)$ is distributed as 
\[ \frac{1}{D^N} \sum_{a_1,b_1,\ldots,a_N,b_N=1}^D \underset{i=1}{\overset{N}{\prod}}\, \braket{g_{a_{i-1}a_i}}{g_{b_{i-1}b_i}} \, , \]
where the $g_{a_{i-1}a_i}$'s are independent Gaussian vectors in $\C^d$ with mean $0$ and variance $1/d$. Indeed,
\begin{align*} 
\Tr(T_N) & = \frac{1}{D^N} \sum_{a_1,b_1,\ldots,a_N,b_N=1}^D \frac{1}{d^N} \sum_{x_1,\ldots,x_N=1}^d \underset{i=1}{\overset{N}{\prod}} \left( \Tr(G_{a_{i-1}a_ix_i})\Tr(\bar{G}_{b_{i-1}b_ix_i}) \right) \\ 
& = \frac{1}{D^N} \sum_{a_1,b_1,\ldots,a_N,b_N=1}^D \underset{i=1}{\overset{N}{\prod}} \left( \sum_{x_i=1}^d g_{a_{i-1}a_ix_i}\bar{g}_{b_{i-1}b_ix_i} \right) \, , 
\end{align*}
where we have set, for each $1\leq i\leq N$ and each $1\leq a_{i-1},a_i\leq D$, $1\leq x_i\leq d$, $g_{a_{i-1}a_ix_i}:= \Tr(G_{a_{i-1}a_ix_i})/\sqrt{d}$, so that the $g_{a_{i-1}a_ix_i}$'s are independent complex Gaussians with mean $0$ and variance $1/d$.

Now, given $1\leq a_1,\ldots,a_N\leq D$, we can re-write
\[ \sum_{b_1,\ldots,b_N=1}^D \underset{i=1}{\overset{N}{\prod}}\, \braket{g_{a_{i-1}a_i}}{g_{b_{i-1}b_i}} = \sum_{q=0}^N \sum_{\substack{I\subset[N] \\ |I|=q}} \sum_{\substack{b_i\neq a_i,\,i\in I \\ b_i=a_i,\,i\notin I}} \underset{i=1}{\overset{N}{\prod}}\, Z_{a_{i-1}a_ib_{i-1}b_i} \, , \]
where we introduced the notation, for each $1\leq a,a',b,b'\leq D$,
\[ Z_{aa'bb'}:=  \braket{g_{aa'}}{g_{bb'}} \, . \]
Given $1\leq q\leq N$ and $I\subset[N]$ such that $|I|=q$, we define $\bar{I}:=\cup_{i\in I}\{i,i+1\}$. We then have that, for $i\in\bar{I}$, $g_{b_{i-1}b_i}$ is independent from $g_{a_{i-1}a_i}$, while for $i\notin\bar{I}$, $g_{b_{i-1}b_i}=g_{a_{i-1}a_i}$. Hence, for $i\in\bar{I}$, we know from Lemma \ref{lem:gg'} that
\[ \P\left( \left|Z_{a_{i-1}a_ib_{i-1}b_i}\right| > \epsilon \right) \leq 2e^{-d\epsilon^2} \, , \]
while for $i\notin\bar{I}$, we know from Lemma \ref{lem:gg} that
\[ \P\left( \left|Z_{a_{i-1}a_ib_{i-1}b_i}\right| > (1+\epsilon)^2 \right) \leq e^{-d\epsilon^2} \, . \]
And therefore,
\begin{align*} 
\P\left( \underset{i=1}{\overset{N}{\prod}}\, \left|Z_{a_{i-1}a_ib_{i-1}b_i}\right| > \epsilon^{|\bar{I}|}(1+\epsilon)^{2(N-|\bar{I}|)} \right) & \leq \sum_{i\in\bar{I}} \P\left( \left|Z_{a_{i-1}a_ib_{i-1}b_i}\right| > \epsilon \right) + \sum_{i\in\bar{I}} \P\left( \left|Z_{a_{i-1}a_ib_{i-1}b_i}\right| > (1+\epsilon)^2 \right) \\
& \leq |\bar{I}|\times 2e^{-d\epsilon^2} + (N-|\bar{I}|)\times e^{-d\epsilon^2} \\
& \leq 2Ne^{-d\epsilon^2} \, .
\end{align*}
Since $|\bar{I}|\geq |I|+1=q+1$, we thus have
\[ \P\left( \sum_{\substack{b_i\neq a_i,\,i\in I \\ b_i=a_i,\,i\notin I}} \underset{i=1}{\overset{N}{\prod}}\, \left|Z_{a_{i-1}a_ib_{i-1}b_i}\right| > D^q\epsilon^{q+1}(1+\epsilon)^{2(N-q-1)} \right) \leq D^q\times 2Ne^{-d\epsilon^2} \, . \]
We then simply have to notice that, on the one hand,
\[ \sum_{q=0}^N{N \choose q}D^q\epsilon^{q+1}(1+\epsilon)^{2(N-q-1)} = \epsilon(1+\epsilon)^{2(N-1)}\left(D\epsilon+(1+\epsilon)^2\right)^N \leq (1+4D\epsilon)^{2N} \, , \]
while on the other hand,
\[ 2Ne^{-d\epsilon^2}\sum_{q=0}^N{N \choose q}D^q=2N(D+1)^Ne^{-d\epsilon^2} \, . \]
And consequently, we get in the end that
\[ \P\left( \sum_{q=0}^N \sum_{\substack{I\subset[N] \\ |I|=q}} \sum_{\substack{b_i\neq a_i,\,i\in I \\ b_i=a_i,\,i\notin I}} \underset{i=1}{\overset{N}{\prod}}\, \left|Z_{a_{i-1}a_ib_{i-1}b_i}\right| > (1+4D\epsilon)^{2N} \right) \leq 2N(D+1)^Ne^{-d\epsilon^2} \, , \]
which implies precisely the result we wanted to show.
\end{proof}	

\begin{theorem} \label{th:dc2-PEPS}
Let $\ket{\chi_M^{N}}\in(\C^d)^{\otimes MN}$ be the random $MN$-site translation-invariant PEPS whose random $M$-site column tensor $\ket{\chi_M}\in(\C^d\otimes\C^D\otimes\C^D)^{\otimes M}$ is defined as in equation \eqref{eq:PEPS'}. Assume that $N>C_0M$, for some universal constant $C_0>0$, and that $d\simeq M^{\alpha}$ and $D\simeq M^{\beta}$ with $\alpha>11$ and $(\alpha+1)/3<\beta<(\alpha-3)/2$. Then, with probability larger than $1-e^{-cM^{1/2}}$, for any $k\leq N-C_0M$ and any Hermitian operators $A,A'$ on $(\C^d)^{\otimes M}$,
\[ \gamma_{\chi}(A,A',k) \leq C'\times\left(\frac{C}{M^{\alpha/2-\beta-1}}\right)^k\|A\|_{\infty}\|A'\|_{\infty} \, , \]
where $c,C,C'>0$ are universal constants.
\end{theorem}

\begin{proof}
First of all, we know by Theorem \ref{th:gap-PEPS} that there exist universal constants $\hat{c},\hat{C}>0$ such that, with probability larger than $1-e^{-\hat{c}M^{3\beta-\alpha}}$,
\[ |\lambda_1(T_M)|  \geq 1-\frac{\hat{C}}{M^{\alpha/2-\beta-1}}\ \ \text{and}\ \ \forall\ 2\leq i\leq D^{2M},\ |\lambda_i(T_M)| \leq |\lambda_1(T_M)|\times \frac{\hat{C}}{M^{\alpha/2-\beta-1}} \, . \]
Next, we know by Lemma \ref{lem:trace-T_N} (applied with $\epsilon=\tilde{C}/M^{(\alpha-1)/2}$) that, with probability larger than $1-e^{-\tilde{c}M^{1/2}}$,
\[ \Tr(T_M)\leq \left(1+\frac{\tilde{C}}{M^{\alpha/2-\beta-3/2}}\right)^2 \, . \]
Now, combining Corollary \ref{cor:correlation} and Lemma \ref{lem:chi-T} we know that, setting $\epsilon(T_M):=|\lambda_1(T_M)|/|\lambda_2(T_M)|$, we have, for any $k\leq N-C_0M$ and any Hermitian operators $A,A'$ on $(\C^d)^{\otimes M}$,
\[ \gamma_{\chi}(A,A',k) \leq \frac{10(\Tr(T_M))^2}{|\lambda_1(T_M)|^2\left(1-\epsilon(T_M)^k\right)^2} \epsilon(T_M)^k \|A\|_{\infty}\|A'\|_{\infty} \, .\]
We now just have to observe that, for $\alpha>11$, $(\alpha+1)/3<(\alpha-3)/2$ (so that the range of possible values for $\beta$ is not empty), and that with the assumptions we made on $\alpha,\beta$, we have $\left(1+\tilde{C}/M^{\alpha/2-\beta-3/2}\right)^2\leq \tilde{C}'$ and $M^{3\beta-\alpha}\geq M>M^{1/2}$. Hence putting everything together, we eventually get that there exist universal constants $c,C, C'>0$ such that, with probability larger than $1-e^{-cM^{1/2}}$, for any $k\leq N-C_0M$ and any Hermitian operators $A,A'$ on $(\C^d)^{\otimes M}$,
\[ \gamma_{\chi}(A,A',k) \leq C'\times\left(\frac{C}{M^{\alpha/2-\beta-1}}\right)^k\|A\|_{\infty}\|A'\|_{\infty} \, ,\]
exactly as announced.	
\end{proof}	

Note that the setting of Theorem \ref{th:dc2-PEPS} is more or less the same as the one of Theorem \ref{th:dc1'-PEPS-block1}, with $d,D$ having to grow with $N$. The only difference is that in Theorem \ref{th:dc2-PEPS} $d$ and $D$ are both polynomial in $N$, while in Theorem \ref{th:dc1'-PEPS-block1} $d$ is polynomial in $N$ and $D$ is sub-polynomial in $N$. Anyway, the obtained typical correlation length is of the same order in both cases, namely $1/\log N$.

\section{Other implications} 
\label{sec:implications}

\subsection{New random constructions of quantum expanders} \hfill\par\smallskip
\label{sec:quantum-expander}

Let $S$ be an MPS transfer operator, of the form
\[ S=\sum_{x=1}^{d} K_x\otimes\bar{K}_x \, , \]
where the $K_x$'s are $D\times D$ matrices. $S$ can equivalently be seen as a completely positive (CP) map on the set of $D\times D$ matrices, having $d$ Kraus operators, defined by
\[ \mathcal{S}(X) := \sum_{x=1}^{d} K_x X K_x^* \, . \]

Given a unit vector $\ket{\varphi}\in\C^{D}\otimes\C^{D}$, written as
\[ \ket{\varphi} = \sum_{i,j=1}^{D}\varphi_{ij}\ket{ij} \, , \]
define $M_{\varphi}$, $D\times D$ matrix with unit Hilbert--Schmidt norm, as
\[ M_{\varphi}:\, = \sum_{i,j=1}^{D}\varphi_{ij}\ketbra{i}{j} \, . \] 
It is then easy to see that
\[ S\ket{\varphi} = \lambda\ket{\varphi} \ \Longleftrightarrow\ \mathcal{S}(M_{\varphi}) = \lambda M_{\varphi} \, . \]
So $S$ and $\mathcal{S}$ have the same eigenvalues and associated eigenvectors in one-to-one correspondence.

The notion of quantum expander was introduced in \cite{BASTS}. We here adopt the slightly generalized definition of \cite{Has1}, which applies to any completely positive trace-preserving (CPTP) map (not necessarily unital, i.e.~having the maximally mixed state as fixed state). A CPTP map on $n\times n$ matrices is called $k$-regular if it has Kraus rank at most $k$, and it is called $(1-\varepsilon)$-expanding if its second largest eigenvalue (in modulus) is at most $\varepsilon$. Such a CPTP map is called a quantum expander with parameters $(m,k,\varepsilon)$ if, in addition, its fixed state has entropy at least $\log m$. A `good' quantum expander should have $m$ as large at possible and $k,\varepsilon$ as small as possible.

Let $T$ be the random MPS transfer operator, as defined by equation \eqref{eq:MPS}, and let $\mathcal{T}$ be its associated random CP map, which is thus defined by
\begin{equation} \label{eq:CalT}
\mathcal{T}(X) := \frac{1}{d} \sum_{x=1}^d G_x X G_x^* \, . 
\end{equation}
We know by Theorem \ref{th:gap-MPS} that
\[ \P\left( 1-\frac{C}{\sqrt{d}} \leq \lambda_1(T) \leq 1+\frac{C}{\sqrt{d}} \ \text{and} \ |\lambda_2(T)|\leq \frac{C}{\sqrt{d}} \right) \geq 1-e^{-cD} \, ,  \]
which by the preceding discussion is equivalent to
\[ \P\left( 1-\frac{C}{\sqrt{d}} \leq \lambda_1(\mathcal{T}) \leq 1+\frac{C}{\sqrt{d}} \ \text{and} \ |\lambda_2(\mathcal{T})|\leq \frac{C}{\sqrt{d}} \right) \geq 1-e^{-cD} \, .  \]
Furthermore, denoting by $\ket{\phi}$ the eigenvector of $T$ with associated eigenvalue $\lambda_1(T)$, we also know by Theorem \ref{th:gap-MPS} that, with probability larger than $1-e^{-cD}$, there exists a vector $\ket{\phi'}$ with norm at most $1$ such that
\[ \ket{\phi} = \ket{\psi} + \frac{C}{\sqrt{d}}\ket{\phi'} \, . \]
This implies that the eigenvector of $\mathcal{T}$ with associated eigenvalue $\lambda_1(\mathcal{T})$ is $M_{\phi}$, which, with probability larger than $1-e^{-cD}$, is of the form
\[ M_{\phi} = M_{\psi} + \frac{C}{\sqrt{d}}M_{\phi'} = \frac{\Id}{\sqrt{D}} + \frac{C}{\sqrt{d}}M_{\phi'} \, , \]
where $M_{\phi'}$ has Hilbert--Schmidt norm at most $1$.

$\mathcal{T}$ is a priori not TP. Nevertheless, we will show that we can construct a random CPTP map $\hat{\mathcal{T}}$ which is with high probability a good approximation of $\mathcal{T}$ (at least as soon as $d\geq C_0D$ for some universal constant $C_0>0$). With this aim in view set
\begin{equation} \label{eq:Sigma}
\Sigma := \frac{1}{d} \sum_{x=1}^d G_x^*G_x \, . 
\end{equation}
Saying that $\mathcal{T}$ is close to being TP is equivalent to saying that $\Sigma$ is close to the identity, which is what we prove below.

\begin{lemma} \label{lem:Sigma}
Let $\Sigma$ be the random $D\times D$ matrix defined by equation \eqref{eq:Sigma}. Then,
\[ \P\left( \|\Sigma-\Id\|_{\infty} \leq \frac{C}{\sqrt{d}} \right) \geq 1-e^{-cD} \, , \]
where $c,C>0$ are universal constants.
\end{lemma}

\begin{proof}
Observe that, re-scaled by a factor $dD$, $\Sigma$ is a $D\times D$ Wishart matrix of parameter $dD$. The result thus immediately follows from Theorem \ref{th:Wishart}, applied with $n=D$ and $s=dD$.
\end{proof}

By Lemma \ref{lem:Sigma} we get in particular that %, as soon as $d>C^2D$, 
$\Sigma$ is invertible with probability larger than $1-e^{-cD}$. If this is the case, then we can define the map $\hat{\mathcal{T}}$ by
\begin{equation} \label{eq:hatCalT}
\hat{\mathcal{T}}(X) := \frac{1}{d} \sum_{x=1}^d \Sigma^{-1/2}G_x X G_x^*\Sigma^{-1/2} = \Sigma^{-1/2}\mathcal{T}(X)\Sigma^{-1/2} \, .
\end{equation}
$\hat{\mathcal{T}}$ is by construction a CPTP map. Let us now show that it is with high probability a good approximation of $\mathcal{T}$.

\begin{corollary} \label{cor:T-hatT}
Let $\mathcal{T}$ and $\hat{\mathcal{T}}$ be the random CP and CPTP maps defined by equations \eqref{eq:CalT} and \eqref{eq:hatCalT}. %If $d\geq C_0 D$, 
Then,
\[ \P\left( \left\|\hat{\mathcal{T}}-\mathcal{T}\right\|_{2\rightarrow 2} \leq \frac{C}{\sqrt{d}} \right) \geq 1-e^{-cD} \, , \]
where $c,C>0$ are universal constants.
\end{corollary}

\begin{proof}
For any $X$ such that $\|X\|_2\leq 1$, we have
\begin{align*}
\| \hat{\mathcal{T}}(X)-\mathcal{T}(X) \|_2 & = \| \Sigma^{-1/2}\mathcal{T}(X)\Sigma^{-1/2}-\mathcal{T}(X) \|_2 \\
& \leq \| (\Sigma^{-1/2}-\Id)\mathcal{T}(X)\Sigma^{-1/2} \|_2 + \| \mathcal{T}(X)(\Sigma^{-1/2}-\Id) \|_2 \\
& \leq \left( \|\Sigma^{-1/2}\|_{\infty} + 1 \right) \left\|\mathcal{T}(X)\right\|_{2} \|\Sigma^{-1/2}-\Id\|_{\infty} \, ,
\end{align*}
where the first inequality is by the triangle inequality, after noticing that $AYA-Y=(A-\Id)YA+Y(A-\Id)$, while the second inequality is by H\"{o}lder inequality. 

Now first of all, we know by Lemma \ref{lem:Sigma} that, with probability larger than $1-e^{-cD}$, $\|\Sigma^{-1/2}-\Id\|_{\infty}\leq C/\sqrt{d}$, so a fortiori $\|\Sigma^{-1/2}\|_{\infty}\leq 1+ C/\sqrt{d}$. What is more, with probability larger than $1-e^{-cD}$, for any $X$ such that $\|X\|_2\leq 1$,
\[ \left\| \mathcal{T}(X) \right\|_{2} \leq |\lambda_1(\mathcal{T})| \|X\|_2 \leq 1+\frac{C}{\sqrt{d}} \, . \]

Putting everything together, we thus eventually get that, with probability larger than $1-2e^{-cD}$, for any $X$ such that $\|X\|_2\leq 1$,
\[ \| \hat{\mathcal{T}}(X)-\mathcal{T}(X) \|_2 \leq \left(1+\frac{C}{\sqrt{d}}\right)^2\frac{C}{\sqrt{d}} \leq \frac{4C}{\sqrt{d}} \, . \]
which (suitably re-labelling $c,C$) is precisely the advertised result.
\end{proof}

\begin{proposition} \label{prop:hatrho}
Let $\hat{\mathcal{T}}$ be the random CPTP map defined by equation \eqref{eq:hatCalT}, and denote by $\hat{\rho}$ its fixed state. %If $d\geq C_0 D$, 
Then,
\[ \P\left( \left\|\hat{\rho}\right\|_{2} \leq \left( 1+\frac{C}{\sqrt{d}} \right) \frac{1}{\sqrt{D}} \right) \geq 1-e^{-cD} \, , \]
where $c,C>0$ are universal constants.
\end{proposition}

\begin{proof}
Since $\mathcal{T}$ is CP, we can assume without loss of generality that $M_{\phi}\geq 0$. We can thus define $\rho_{\phi}:=M_{\phi}/\|M_{\phi}\|_1$, which is the state such that $\mathcal{T}(\rho_{\phi})=\lambda_1(\mathcal{T})\rho_{\phi}$. Then,
\begin{align*}
\left\|\hat{\rho}-\rho_{\phi}\right\|_2 & = \left\|\hat{\mathcal{T}}(\hat{\rho})-\frac{1}{\lambda_1(\mathcal{T})}\mathcal{T}(\rho_{\phi})\right\|_2 \\
& \leq \left\|\hat{\mathcal{T}}(\hat{\rho})-\hat{\mathcal{T}}(\rho_{\phi})\right\|_2 + \left\|\hat{\mathcal{T}}(\rho_{\phi})-\frac{1}{\lambda_1(\mathcal{T})}\mathcal{T}(\rho_{\phi})\right\|_2 \\
& \leq |\lambda_2(\hat{\mathcal{T}})| \left\|\hat{\rho}-\rho_{\phi}\right\|_2 + \left\|\hat{\mathcal{T}}-\frac{1}{\lambda_1(\mathcal{T})}\mathcal{T}\right\|_{2\rightarrow 2} \left\|\rho_{\phi}\right\|_2 \, .
\end{align*} 
This means that
\[ \left\|\hat{\rho}-\rho_{\phi}\right\|_2 \leq \frac{1}{1-|\lambda_2(\hat{\mathcal{T}})|} \left\|\hat{\mathcal{T}}-\frac{1}{\lambda_1(\mathcal{T})}\mathcal{T}\right\|_{2\rightarrow 2} \left\|\rho_{\phi}\right\|_2 \, , \]
and therefore, by the triangle inequality, that
\[ \left\|\hat{\rho}\right\|_{2} \leq \left\|\rho_{\phi}\right\|_2 + \left\|\hat{\rho}-\rho_{\phi}\right\|_2 \leq \left( 1 + \frac{1}{1-|\lambda_2(\hat{\mathcal{T}})|} \left\|\hat{\mathcal{T}}-\frac{1}{\lambda_1(\mathcal{T})}\mathcal{T}\right\|_{2\rightarrow 2} \right) \left\|\rho_{\phi}\right\|_2 \, . \]
Now, first of all we know that, with probability larger than $1-e^{-cD}$, $|\lambda_1(\mathcal{T})|\leq 1+C/\sqrt{d}$, so that
\[ \left\|\hat{\mathcal{T}}-\frac{1}{\lambda_1(\mathcal{T})}\mathcal{T}\right\|_{2\rightarrow 2} \leq \left\|\hat{\mathcal{T}}-\mathcal{T}\right\|_{2\rightarrow 2} + \frac{2C}{\sqrt{d}}\left\|\mathcal{T}\right\|_{2\rightarrow 2} \leq \left\|\hat{\mathcal{T}}-\mathcal{T}\right\|_{2\rightarrow 2} + \frac{2C}{\sqrt{d}}\left(1+\frac{C}{\sqrt{d}}\right) \, .  \]
Next, we know by Corollary \ref{cor:T-hatT} that, with probability larger than $1-e^{-cD}$, 
\[ \left\|\hat{\mathcal{T}}-\mathcal{T}\right\|_{2\rightarrow 2} \leq \frac{C}{\sqrt{d}} \, . \]
Since we also know that, with probability larger than $1-e^{-cD}$, $|\lambda_2(\mathcal{T})|\leq C/\sqrt{d}$, this implies that, with probability larger than $1-2e^{-cD}$,
\[ |\lambda_2(\hat{\mathcal{T}})| \leq \frac{2C}{\sqrt{d}} \, . \]
Putting everything together we thus get that, with probability larger than $1-3e^{-cD}$,
\[ \left\|\hat{\rho}\right\|_{2} \leq \left( 1 + \frac{C'}{\sqrt{d}} \right) \left\|\rho_{\phi}\right\|_2 \, . \]

The only thing that now remains to be done is to upper bound the typical value of $\|\rho_{\phi}\|_2$. Since $\|\rho_{\phi}\|_2=1/\|M_{\phi}\|_1$, we actually have to lower bound the typical value of $\|M_{\phi}\|_1$. Yet, we know that, with probability larger than $1-e^{-cD}$,
\[ \left\|M_{\phi}\right\|_1 = \left\| \frac{\Id}{\sqrt{D}} + \frac{C}{\sqrt{d}} M_{\phi'} \right\|_1 \geq \left\| \frac{\Id}{\sqrt{D}} \right\|_1 - \frac{C}{\sqrt{d}} \left\| M_{\phi'} \right\|_1 \geq \left(1-\frac{C}{\sqrt{d}}\right)\sqrt{D} \, , \]
where the last inequality is because $\left\| M_{\phi'} \right\|_1\leq\sqrt{D}\left\| M_{\phi'} \right\|_2\leq \sqrt{D}$.

So in the end we have shown that, with probability larger than $1-4e^{-cD}$,
\[ \left\|\hat{\rho}\right\|_{2} \leq \left( 1 + \frac{C''}{\sqrt{d}} \right) \frac{1}{\sqrt{D}} \, , \]
which (suitably re-labelling $c,C$) is precisely the advertised result.
\end{proof}

\begin{theorem}
Let $\hat{\mathcal{T}}$ be the random CPTP map defined by equation \eqref{eq:hatCalT}. Then, with probability larger than $1-e^{-cD}$, $\hat{\mathcal{T}}$ is a quantum expander with parameters $((1-C/\sqrt{d})D,d,C/\sqrt{d})$, where $c,C>0$ are universal constants.
\end{theorem}

\begin{proof}
The fact that $\hat{\mathcal{T}}$ is $d$-regular is clear by definition. 
The fact that, with probability larger than $1-e^{-cD}$, $|\lambda_2(\hat{\mathcal{T}})|\leq C/\sqrt{d}$, so that $\hat{\mathcal{T}}$ is $(1-C/\sqrt{d})$-expanding, is established in the proof of Proposition \ref{prop:hatrho}. So it only remains to lower bound the typical entropy of $\hat{\rho}$. By concavity of $\log$ we have 
\[ S(\hat{\rho}) = -\Tr(\hat{\rho}\log\hat{\rho}) \geq -\log\Tr(\hat{\rho}^2) \, .\]
Combining this observation with Proposition \ref{prop:hatrho} we get that, with probability larger than $1-e^{-cD}$,
\[ S(\hat{\rho}) \geq -\log\left(\left( 1 + \frac{C}{\sqrt{d}} \right)^2 \frac{1}{D} \right) \geq \log \left(\left( 1 - \frac{6C}{\sqrt{d}} \right) D \right) \, , \]
which concludes the proof (after suitably re-labelling $c,C$).
\end{proof}	

\subsection{A model of random dissipative evolution} \hfill\par\smallskip
\label{sec:random-evolution}

In the recent past, the analysis of local random circuits has triggered great interest. Very broadly speaking, the main question that one tries to answer in this field is: after which depth does the action of a circuit composed of local random unitaries `resembles' that of a global random unitary (on any input many-body state)? Typical features of a global Haar-distributed unitary, that one would like to reproduce in a more economical way, include: scrambling \cite{BF1,BHH,HM}, decoupling \cite{BF2}, entanglement spreading \cite{HNRV,HNV} etc. Such models for random reversible evolutions have now been studied quite extensively. But what about similar models for random dissipative evolutions? 

Our random PEPS transfer operator, as defined by equation \eqref{eq:transfer-PEPS}, actually provides one such model of random evolution in the open system picture. Indeed, as explained in Section \ref{sec:quantum-expander} in the case of MPS, a random transfer operator can equivalently be seen (after renormalization) as a random quantum channel, i.e.~a random evolution of a system coupled to an environment. What is more, the quantum channel corresponding to a PEPS transfer operator is by construction acting on its input many-body state with some locality constraints, as a local circuit in the case of isolated systems. The only issue is that, in this particular context of chaotic quantum dynamics, a non translation-invariant model is usually more relevant than a translation-invariant one. 

So let us start with explaining how the results of Section \ref{sec:transfer-operator-PEPS} can be extended to the case of non translation-invariant random PEPS. The model that we are now considering is constructed in a very similar way to the one we have been looking at up to here: the only difference is that we sample the $N^2$ $1$-site tensors independently from one another, each being distributed as the $1$-site tensor defined by equation \eqref{eq:PEPS}. Let us denote by $T_N'$ the corresponding transfer operator, i.e.
\begin{equation} \label{eq:transfer-PEPS-ind}
T_N' := \frac{1}{D^N} \sum_{a_1,b_1,\ldots,a_N,b_N=1}^D \frac{1}{d^N} \sum_{x_1,\ldots,x_N=1}^d G_{a_Na_1x_1}^1 \otimes \bar{G}_{b_Nb_1x_1}^1\otimes \cdots\otimes G_{a_{N-1}a_Nx_N}^N \otimes \bar{G}_{b_{N-1}b_Nx_N}^N \, ,
\end{equation}
where the $G_{a_{i-1}a_ix_i}^i$'s are independent $D\times D$ matrices whose entries are independent complex Gaussians with mean $0$ and variance $1/D$. 

The analogues of Propositions \ref{prop:PEPS-1} and \ref{prop:PEPS-2} now read as follows.

\begin{proposition} \label{prop:PEPS-1-ind}
	Let $T_N'$ be defined as in equation \eqref{eq:transfer-PEPS-ind}. Then, $\bra{\psi^{\otimes N}}T_N'\ket{\psi^{\otimes N}}\in\R$ and there exist universal constants $C,c>0$ such that
	\[ \P\left( \left| \bra{\psi^{\otimes N}}T_N'\ket{\psi^{\otimes N}} - 1 \right| \leq \frac{CN}{\sqrt{d}}+D^2\left(\frac{C}{\sqrt{d}}\right)^N \right) \geq 1-Ne^{-cD^3} \, . \]
\end{proposition}

\begin{proposition} \label{prop:PEPS-2-ind}
	Let $T_N'$ be defined as in equation \eqref{eq:transfer-PEPS-ind}. Then, there exist universal constants $C,c>0$ such that, for all $\sqrt{d}\geq D$,
	\[ \P\left( \left\|T_N'(\Id-\ketbra{\psi^{\otimes N}}{\psi^{\otimes N}})\right\|_{\infty} \leq (1+\eta)\left(1+\frac{CD}{\sqrt{d}}\right)^N\frac{CN}{\sqrt{d}} \right) \geq 1- N^2D^{N}e^{-cD^3/d} \, , \]
	where $\eta\equiv\eta(N,d,D)= d^N N^{2N} (1+C/\sqrt{d})^N e^{-cD^3/d}$.
\end{proposition}

Note that the only difference with Propositions \ref{prop:PEPS-1} and \ref{prop:PEPS-2} is an extra $N$ factor in the deviation probabilities. This is quite intuitive to understand: contrary to $T_N$ which is made out of one single random tensor repeated $N$ times, $T_N'$ is made of $N$ independent random tensors, so that deviation probabilities for each of them have to somehow add up. We will not fully redo the proofs in this non translation-invariant case, but simply explain how the proofs in the translation-invariant case have to be modified.

\begin{proof}[Sketch of proof of Proposition \ref{prop:PEPS-1-ind}]
Similarly to the proof of Proposition \ref{prop:PEPS-1}, we begin with observing that
\[  \bra{\psi^{\otimes N}}T_N'\ket{\psi^{\otimes N}}=\Tr\left(\tilde{T}_1\cdots\tilde{T}_N\right) \, , \]
where the $\tilde{T}_i$'s are independent and distributed as $T$ with $\tilde{d}=D^2d$ and $\tilde{D}=D$. Next, for each $1\leq i\leq N$, write $\tilde{T}_i$ as $\tilde{T}_i=\lambda_i \ketbra{\psi}{\psi}+\epsilon_i M_i$ with $\bra{\psi}M_i\ket{\psi}=0$ and $\|M_i\|_{\infty}\leq 1$. We thus have
\[ \Tr\left(\tilde{T}_1\cdots\tilde{T}_N\right) = \lambda_1\cdots\lambda_N + \epsilon_1\cdots\epsilon_N \Tr(M_1\cdots M_N) \, . \]
Now, we know from Propositions \ref{prop:MPS-1'} and \ref{prop:MPS-2'} that, for each $1\leq i\leq N$, with probability larger than $1-e^{-cD^3}$, $|\lambda_i-1|\leq C/\sqrt{d}$ and $|\epsilon_i|\leq C/\sqrt{d}$. Therefore, with probability larger than $1-Ne^{-cD^3}$, 
the two following hold
\[ \bra{\psi^{\otimes N}}T_N'\ket{\psi^{\otimes N}} \geq \left(1-\frac{C}{\sqrt{d}}\right)^N - D^2\left(\frac{C}{\sqrt{d}}\right)^N \ \text{and} \  \bra{\psi^{\otimes N}}T_N'\ket{\psi^{\otimes N}} \leq \left(1+\frac{C}{\sqrt{d}}\right)^N + D^2\left(\frac{C}{\sqrt{d}}\right)^N \, , \]
which is precisely what we wanted to show.
\end{proof}

\begin{proof}[Sketch of proof of Proposition \ref{prop:PEPS-2-ind}]
Our first claim is: Lemma \ref{lem:dev-norm-sym-n} holds exactly the same when, instead of having only $d$ independent matrices $G_x$'s, one has $dN$ independent matrices $G_x^i$'s, only replacing $3N$ by $N^2$ in the deviation probability. Indeed, the induction proof works exactly alike. The only thing that changes is that we now have to define
\[ Y := \left\| \frac{1}{d}\sum_{x=1}^d G_x\otimes \bar{G}_x \right\|_{\infty} \ \text{and}\  X_n := \left\| \frac{1}{d^n}\sum_{x_1,\ldots,x_n=1}^d \left( \underset{i=1}{\overset{n}{\bigotimes}}\, G_{x_i}^i\otimes\bar{G}_{x_i}^i - \underset{i=1}{\overset{n}{\bigotimes}}\, H_{x_i}^i\otimes\bar{H}_{x_i}^i \right) \right\|_{\infty} \, , \]
and the upper bound \eqref{eq:recursion-ineq} becomes
\[ X_{n+1} \leq Y_1\cdots Y_n X_1 + X_n'Y_1' \, . \]
Then, we know on the one hand that $Y>(1+\epsilon)(1+C/\sqrt{d})$ with probability at most $e^{-cDd\min(\epsilon,\epsilon^2)}$, so that also, clearly, $Y_1\cdots Y_n >(1+\epsilon)^n(1+C/\sqrt{d})^n$ with probability at most $ne^{-cDd\min(\epsilon,\epsilon^2)}$. While we know on the other hand that, by recursion hypothesis, $X_1>(1+\epsilon)C/\sqrt{d}$ with probability at most $e^{-cDd\min(\epsilon/\sqrt{d},\epsilon^2/d)}$ and $X_n>(1+\epsilon)^n(1+C/\sqrt{d})^nCn/\sqrt{d}$ with probability at most $n^2e^{-cDd\min(\epsilon/\sqrt{d},\epsilon^2/d)}$. This implies, exactly as wanted, that $X_{n+1}>(1+\epsilon)^{n+1}(1+C/\sqrt{d})^{n+1}C(n+1)/\sqrt{d}$ with probability at most $(n+1)^2e^{-cDd\min(\epsilon/\sqrt{d},\epsilon^2/d)}$.

It then immediately follows that also Corollary \ref{cor:dev-norm-sym-n} holds exactly the same when, instead of having only $d$ independent matrices $G_x$'s, one has $dN$ independent matrices $G_x^i$'s, only replacing $3N$ by $N^2$ in the deviation probability. And from there we deduce that Proposition \ref{prop:PEPS-2} as well holds exactly the same for $T_N'$ instead of $T_N$, only replacing $4N(D+1)^{N}$ by, say, $2N^2(D+1)^{N}$ in the deviation probability.
\end{proof}

From Propositions \ref{prop:PEPS-1-ind} and \ref{prop:PEPS-2-ind}, one can then straightforwardly derive the analogue of Theorem \ref{th:gap-PEPS} in this independent case, as stated below. 

\begin{theorem} \label{th:gap-PEPS-ind}
	There exist universal constants $C,c>0$ such that, for all $\sqrt{d}\geq D$, 
	\[  \P\left( \Delta(T_N') \geq 1 - (1+\eta)\left(1+\frac{CD}{\sqrt{d}}\right)^N\frac{CN}{\sqrt{d}} +D^2\left(\frac{C}{\sqrt{d}}\right)^N \right) \geq 1- N^2D^{N}e^{-cD^3/d}  \, , \]
	where $\eta\equiv\eta(N,d,D)= d^N N^{2N} (1+C/\sqrt{d})^N e^{-cD^3/d}$.
	
	In particular, there exist universal constants $C',c'>0$ such that, if $d\simeq N^{\alpha}$ and $D\simeq N^{\beta}$ with $\alpha>8$ and $(\alpha+1)/3<\beta<(\alpha-2)/2$, then
	\[ \P\left( \Delta(T_N') \geq 1-\frac{C'}{N^{\alpha/2-\beta-1}} \right) \geq 1-e^{-c'N^{3\beta-\alpha}} \, . \]
\end{theorem}

This lower bound on the typical spectral gap of the PEPS transfer operator $T_N'$, together with the knowledge that $\ket{\psi^{\otimes N}}$ is typically close to its largest eigenvalue eigenvector, would now allow for an analysis quite similar to that carried on in Section \ref{sec:quantum-expander} for the MPS transfer operator $T$. In particular, one could study the following questions: How close typically is the fixed point of the quantum channel $\hat{\mathcal{T}}_N'$, corresponding to $T_N'$, to the maximally mixed state? And when iteratively applying $\hat{\mathcal{T}}_N'$ to an input state, how fast does the latter typically converge towards this fixed point?

\subsection{Miscellaneous final comments} \hfill\par\smallskip

Let us start with a few comments on the results of Section \ref{sec:transfer-operator}, about the typical spectral gap of random transfer operators. In the MPS case, we know that the scaling we obtain for the spectral gap is optimal. In the PEPS case though, the statements that we are able to make remain not fully satisfying. The main open question clearly is: could these results be improved so that $d,D$ growing polynomially with $N$ is not needed? With our current proof techniques, the exponents in this polynomial dependence could be optimized. But getting rid of this limitation would require totally different arguments.
Hence for now, the transfer operator approach does not seem to be the most suited to tackle the case of higher dimensional regular lattices. Indeed, even though it is extremely powerful in dimension $1$, it is doomed to yield results which are not independent of the system size in dimension $2$ or higher. 

Concerning the results of Section \ref{sec:parent-hamiltonian}, they are likely to be sub-optimal in both the MPS and PEPS cases. Indeed, there is no a priori obstruction for the validity regime $d> D^{\theta}$ to be improved to $\theta=2$ for MPS and $\theta=4$ for PEPS. So it would be nice to be able to get closer to this regime. It is indeed clear what are the two points in the proofs where we probably lose something. First it is when upper bounding the operator norm of realigned random matrices (Proposition \ref{prop:norm-M} in the case of MPS): we pick up local dimension factors which are quite likely not to be necessary. Second it is when upper bounding the operator norm of approximate ground space projectors on the complement of the ground space by their trace norm (Proposition \ref{prop:invariant'} in the case of MPS). In both cases, getting upper bounds with the optimal order of magnitude would require a careful analysis of the specific random matrix models under consideration. Using similar techniques as those used in the proofs of Section \ref{sec:transfer-operator-MPS}, this does not seem out of reach. 
What is more, a nice feature of the parent Hamiltonian approach is that it can a priori be easily generalized to any regular lattice. 
	
Let us also say, about this parent Hamiltonian analysis, that it could in principle be generalized to non-injective random MPS and PEPS. For instance, if our random MPS is in a dimensional regime where it is almost surely injective only after blocking together $K$ sites (with $K>1$), then its associated canonical parent Hamiltonian is not $2$-local but $(K+1)$-local. However, its construction remains entirely similar, and the strategy to try and lower bound its spectral gap as well. So we believe that the reasoning would carry through. It would just be more cumbersome, which is why we have restricted ourselves to writing it down properly only under the injectivity assumption. Of course, this guess that our results would generalize in a quite straightforward way to $K$-injective random MPS is only for $K$ a fixed constant. Instead, if $K$ were to depend on other parameters involved (such as $d,D,N$), then the analysis could become much more subtle.

These results on random parent Hamiltonians being typically gapped obviously trigger a new question: instead of constructing a translation-invariant ground state at random and then studying the spectral properties of the corresponding local translation-invariant Hamiltonian, what about directly constructing the Hamiltonian at random? This viewpoint is the one adopted in \cite{Lem}. It would be interesting to see if the results in the latter could be extended to non frustration-free situations (i.e.~to situations which are outside of the parent Hamiltonian picture).

Another, very different, route that one could explore is how to change our model of random MPS and PEPS in a meaningful (but manageable) manner? A natural idea would be to sample the $1$-site tensors in a non unitarily-invariant way. Indeed, in a model where either physical or bond indices would be favoured, or even some bond indices compared to others, interesting phenomena might arise. But being able to attack such problem seems to be quite challenging. 
It would also be useful to understand what happens when some symmetry is imposed on the random $1$-site tensors (i.e.~when the latter are sampled under the constraint that they are invariant under the action of some group). This situation of MPS and PEPS exhibiting a local symmetry is indeed very important in practice \cite{Cirac20}. 

\section*{Acknowledgements}

We would like to thank Andrea Coser, David Gosset and Ramis Movassagh for sharing various interesting thoughts at several points of this project. We are also extremely grateful to Henrik Wilming for pointing out to us a mistake in the first version of this work and to Marius Lemm for noticing an imprecision in a later one. Finally, we would like to thank the two anonymous referees for their numerous and insightful comments, which truly helped in improving the presentation of our results.
C.L.~acknowledges financial support from the French CNRS (project PEPS JCJC) and the French ANR (project Investissement d'Avenir ANR-11-LABX-0040). D.P.-G.~acknowledges financial support from the European Research Council (European Union's Horizon 2020 research and innovation programme, grant No.~648913), the Spanish MICINN (projects MTM2014-54240-P, MTM2017-88385-P and Severo Ochoa CEX2019-000904-S) and the Comunidad de Madrid (projects QUITEMAD+ S2013/ICE-2801 and P2018/TCS-4342).

\end{document}